\documentclass[final,leqno,onefignum]{siamltex}

\usepackage{enumerate,enumitem} 
\usepackage{fancyhdr}
\usepackage{epsfig,amsfonts,latexsym,amsmath,amssymb}

\DeclareMathOperator {\Id} {Id}

\DeclareMathOperator {\Div} {Div} 
 
\DeclareMathOperator {\curl} {curl} 
\DeclareMathOperator {\tr} {Tr}

\newcommand{\RR}{\mathbb{R}}
\newcommand{\de}{\delta}
\newcommand{\La}{\Lambda}

\newtheorem{remark}{Remark}[section]

\title{On the system of elastic-gravitational equations describing the
    oscillations of the earth}

\author{Maarten V. de Hoop\thanks{Rice University, Houston, TX, USA
    ({\tt mdehoop@rice.edu})} \and Sean Holman \thanks{University of
    Manchester, Manchester, UK ({\tt sean.holman@manchester.ac.uk})}
  \and Ha Pham \thanks{INRIA, Pau, France} ({\tt
    ha.howard@inria.fr})}

\begin{document}
\maketitle

\pagestyle{myheadings}
\thispagestyle{plain}
\markboth{DE HOOP, HOLMAN, and PHAM}{Oscillations of the earth}

\begin{abstract}
The linear equations of motion of a uniformly rotating, elastic and
self-gravitating earth model are analyzed under minimal regularity
assumptions. We present existence and uniqueness results for the
system, energy estimates, convergence of Galerkin approximations, and
propose a method based on a Volterra equation to deal with the
nonlocal self-gravitation.
\end{abstract}


\section{Introduction}

The study of Earth's oscillations is fundamental to seismology, as it
is a key part of the theory of earth's dynamic response to external or
internal forces. The same theory is applicable to the analysis of free
oscillations or normal modes and postseismic relaxation, and the
analysis of seismic surface waves and body waves.

We develop a comprehensive framework for the analysis of the
oscillations, establish well-posedness under appropriate and natural
conditions and energy estimates, and provide an anatomy describing,
for example, how the system of equations describing acousto-elastic
waves can be meaningfully extracted. The results in this paper form
the foundations for the characterization of the spectrum of the earth,
which we will give in a follow-up paper. Here, the key complication
arises from the presence of a liquid outer core.

We consider a bounded set $\widetilde{X} \subset \mathbb{R}^3$
representing the interior of the earth, with Lipschitz continuous
exterior boundary $\partial \tilde{X}$. The set $\widetilde{X}$ is
subdivided into solid and fluid regions, denoted by $\Omega^S$ and
$\Omega^F$ respectively. Ideally, the region $\overline{\Omega^F}$ is
the union of a deformed annulus corresponding to the outer core, and
some regions corresponding to the oceans; $\overline{\Omega^S}$ is the
union of a deformed ball corresponding with the inner core and a
deformed annulus representing essentially the mantle and crust.  Both
$\Omega^S$ and $\Omega^F$ can be further divided into subregions;
these subregions will be separated by $C^1$ (inner) interfaces which we will collectively label
as $\Sigma$. Thus, we have
$$
   \widetilde{X} = \Omega^S \cup \Omega^F \cup \Sigma
   \cup \partial \widetilde{X} ,
$$
$$
   \Omega^S = \bigcup_{i=1}^n \Omega_i^S ;\ \
   \Omega^F  = \bigcup_{j=1}^m
      (\Omega^{F}_{\text{int}})_j \cup \bigcup_{k=1}^l \Omega_k^O ,
$$
where $\Omega_k^O$ is an ocean layer, and
$(\Omega^{F}_{\text{int}})_j$ is an internal fluid region.  The outer
boundary $\partial \widetilde{X}$ is divided into the land surface
$(\partial \widetilde{X})_S $ and the ocean surface $(\partial
\widetilde{X})_F$, thus
$$
   \partial \widetilde{X} = \partial \widetilde{X}_S
          \cup \partial \widetilde{X}_F .
$$
The inner interfaces are subdivided into 
$$
   \Sigma = \Sigma^{SS} \cup \Sigma^{FF} \cup \Sigma^{FS} ;\ \
   \Sigma^{FS} = \Sigma^{FS}_{\text{int}} \cup \Sigma^{FS}_O ,
$$
where
\begin{itemize}
\item $\Sigma^{FF}$ is the union of all the different interfaces in
  the interior of $\overline{\Omega^F_{\text{int}}}$, that is, in between
  two inner fluid regions.
\item $\Sigma^{SS}$ is the union of all the different interfaces in
  the interior of $\overline{\Omega^S}$, that is, in between two solid
  regions.
\item $\Sigma^{FS}$ is the union of all the different interfaces
  separating a solid and a fluid region. We further distinguish
  $\Sigma^{FS}_{\text{int}}$ and $\Sigma^{FS}_O$.
  $\Sigma^{FS}_{\text{int}}$ is the union of interfaces between a
  solid layer $\Omega_i^S$ and an inner fluid region $\Omega_j^F$,
  that is, not an ocean layer, while $\Sigma^{FS}_O$ is the union of
  ocean floors, that is, union of interfaces between a solid region
  $\Omega_i^S$ and an ocean layer $\Omega_k^O$. The boundaries of
  ocean layers are composed of ocean floors $\Sigma^{FS}_O$ and ocean
  (free) surfaces $\partial \tilde{X}_F$, that is, $\partial \Omega^O
  = \Sigma^{FS}_O \cup \partial \widetilde{X}_F$.
  \item We will also write $\Sigma^F = \Sigma^{FF} \cup \Sigma^{FS}$ for the union of all interfaces involving a fluid.
\end{itemize}

For the purpose of a well-posedness result, we impose further
restrictions on the above model. In our model, the earth is made up of
`onion-like' layers of the solid regions $\Omega^S_i$ and fluid
regions $\Omega^F_{\text{int},j}$ (except for the oceans). We also
assume that the boundaries and interfaces of different types listed
above do not intersect one another in the interior. We will glue the
different regions together following certain boundary conditions
discussed in subsection \ref{BoCo}. We will also assume that for any
interface between a solid and an inner fluid region or between two
solid regions: $\partial \Sigma^{FS}_{\text{int}} = \emptyset$,
$\partial \Sigma^{SS} = \emptyset$. On the other hand, ocean edges,
which are boundaries of ocean floors $\Sigma^{FS}_O$, as well ocean
surfaces $\partial \widetilde{X}_F$ exist; hence $\partial
\Sigma^{FS}_O = \partial \widetilde{X}_F \neq \emptyset$.

We assume that prior to the occurrence of an earthquake, the earth is
in a state of mechanical equilibrium by requiring that the static
momentum equation \eqref{Equi1} be satisfied throughout $\Omega^S$ and
$\Omega^F$. In $\Omega^F$, the static momentum equation takes the
special form of \eqref{Equi2}. For our analysis, the fluid region
contains a `perfect fluid' characterized by \eqref{PerFluid}.

Denote by $u =u(t,x) $ the displacement which takes values in
$\mathbb{C}^3$. One hopes for existence and uniqueness of solutions to
the following equation of motion modelling the oscillations of an
elastic and self-gravitating earth, imposed with the boundary and
interface conditions listed in Table~\ref{unT},
\begin{equation}\label{globalseis}
   \rho^0 [\ddot u + 2 R_{\Omega} \cdot \dot u]
     + \rho^0 u \cdot \nabla\nabla (\Phi^0 + \Psi^s)
   + \rho^0 \nabla S(u) - \nabla \cdot (\La^{T^0} : \nabla u)
   = \rho^0 f,
\end{equation}
where
$$
   (\La^{T^0} : \nabla u)_{ij} =
   \sum_{k, l = 1}^3 \La^{T^0}_{ijkl} \partial_k u_l = \sum_{k, l =
   1}^3 \La^{T^0}_{jikl} \partial_l u_k ;
$$
$$
   (\nabla\nabla (\Phi^0 + \Psi^s))_{ij} = \partial_i
            \partial_j (\Phi^0 + \Psi^s) ;
$$
$$
   R_{\Omega} \cdot \dot{u} = \Omega \times \dot{u} 
\quad\text{with}\
   R_{\Omega} :=\Bigg (\sum_{j=1}^3 \epsilon_{ijk}
           \Omega_j \Bigg )_{i,k=1}^3 .
$$
Here, $\Omega \in \RR^3$ is the angular velocity of the earth's
rotation and $R_{\Omega}\cdot \dot{u}$ represents the induced Coriolis
force. $\Psi^s(x)$ is the corresponding (spatial) centrifugal
potential given by \eqref{centriPo}, $\Phi^0$ is the gravitational
potential of the reference state given by \eqref{eq:Phi0} and $\Psi^s$
is the centrifugal potential given by
\eqref{centriPo}. Furthermore, $\rho^0$ is the density, related to the
reference gravitational potential $\Phi^0$, given by \eqref{ellPhi},
$\La^{T^0}$ is the modified stiffness tensor defined by
\eqref{ModStiffTensor}, and $S(u)$ is the perturbation of the
gravitational potential due to the oscillation, given by equation
\eqref{eq:Phi1sol}. We will describe all the quantities involved in
this equation in more detail in Section~\ref{IntroEq}. Also, from now
on we will use the summation convention over repeated indices so that
for example
\[
\sum_{k, l =1}^3 \La^{T^0}_{jikl} \partial_l u_k =  \La^{T^0}_{jikl} \partial_l u_k.
\]

Since it is not always possible to hope for a classical solution, one
must explore various notions of weak solution. Coercitivity is the crucial
ingredient in any approach to proving existence and uniqueness of weak, or classical, solutions of \eqref{globalseis}. We thus briefly review the concept of coercivity. Let $H$ and $E$ be Hilbert spaces with $E\hookrightarrow H$ a dense and continuous embedding. A continuous sesquilinear form $a$ over $E\times E$ is said to be $E$ coercive relative to $H$ if
there exist $\alpha > 0$ and $\beta \in \RR$ so that
$$
   a(v,v) \geq \alpha \lVert v\rVert_E^2
       - \beta \lVert v\rVert^2_H ,\quad  \forall \quad v \in E .
$$
This definition also carries over to the unbounded operator $A$ in $H$
corresponding to $a$.  By \cite[Theorem XVII.3.3]{DautrayLionsV2}, if
coercitivity of $A$ holds, then $A$ is the infinitesimal generator of
a semigroup of class $\mathcal{C}^0$ in $H$. From this result,
\cite{Ball} gives the well-posedness for the Cauchy problem $u_t + Au
= f$, $u(0) = g$. This is called the semi-group approach. Another approach, the Galerkin method, also requires coercivity of $A$.

Before discussing the elastic-gravitational equation, we first visit
the instructive simpler case of the linear elastic equation with zero
displacement on the boundary,
$$
   \ddot{u} -\nabla\cdot(\mathbf{a} :\nabla u) = 0\ \text{on}\
          \Omega \ ; \ \ u|_{\partial \Omega} = 0 ,
$$
to motivate the type of assumptions imposed on $\La^{T^0}$ in order to
obtain coerciveness. In this case, if we assume that $\mathbf{a}$ is
strongly elliptic, that is, for all non-zero vectors $\alpha$ and $k$,
$$
   \big(\mathbf{a} :(\alpha \otimes k )\big): (\alpha \otimes k )
           > 0 ,
$$ 
then, e.g. by \cite{Valette}, the associated Dirichlet form is
$H^1_0(\Omega)$ coercive relative to $L^2(\Omega)$. That is, there exists
$c,d > 0$ so that
$$
   \int_{\Omega} (\mathbf{a}:\nabla u ) :
       \nabla \overline{u} \, d V
   \geq c \lVert u\rVert^2_{H^1_0(\Omega)}
         - d \lVert u \rVert^2_{L^2(\Omega)} \quad
   \forall \quad u \in H^1_0(\Omega) .
$$
With this result, \cite[Theorem XVII.3.4]{DautrayLionsV2} gives that $$
\tilde{A} = \begin{pmatrix} 0 & \Id \\ A & 0 \end{pmatrix} ,$$ where $
Au = -\nabla\cdot(\mathbf{a} :\nabla u)$, is the infinitesimal
generator of a semigroup in $H^1_0(\Omega) \times L^2(\Omega)$.  The
well-posedness for a weak solution of the linear elastic equation then
follows with the help of \cite{Ball}.

In fact, strong ellipticity is a necessary condition in order to use
semigroup theory to establish well-posedness. Indeed, by Theorem 3.7
in \cite{Marsden}, if strong ellipticity strictly fails, $ \tilde{A}$
cannot generate a semigroup on any sub-Banach space $Y$ of $L^2(\Omega)$ with
$D(A) \subset Y$. If we wish to consider cases in which the
displacement $u$ is nonzero at the boundary of $\Omega$, then strong
ellipticity of $\mathbf{a}$ is not sufficient to obtain
coercitivity. Indeed, \cite{Valette} contains a counter
example\footnote{ It should be noted that this counter example
  actually shows that \cite[Proposition 1.5]{Marsden} is incorrect as
  stated.} showing that strong ellipticity does not imply coercivity
for the Dirichlet form on $H^1(\Omega)$ relative to $L^2(\Omega)$. The
stronger requirement, which we will need and which is referred to as
{\it pointwise stability}\footnote{In \cite[Proposition 3.1]{Marsden},
  pointwise stability implies strong ellipticity. However, the reverse
  statement is not true, as can be seen in the case of an isotropic
  and homogeneous elastic tensor, given in \cite[Proposition
    3.13]{Marsden}.} in \cite{Marsden}, is that, for any symmetric two
tensor $C$, $(\mathbf{a}:C):C > 0$.

In our problem with the elastic-gravitational equation, by integration
by parts (requiring $u$ and $v$ to have sufficient regularity), we
obtain the preliminary weak formulation
$$
   \dfrac{d}{dt} \int_{\tilde{X}}
          \rho^0 \dot{u} \cdot \overline{v} \, dV
   + \int_{\tilde{X}} 2 \rho^0 R_{\Omega} \dot{u}
    \cdot \overline{v} \, dV + a_{\text{original}}\big(u, v\big)
   = \int_{\tilde{X}} \rho^0 f \cdot \overline{v} \, dV ,
$$
where the sesquilinear form $a_{\text{original}}$ is defined in
\eqref{aOrig}. Although the solid regions support waves in all
directions, which is equivalent to $\Lambda^{T^0}$ being strongly
elliptic\footnote{By Proposition 3.9 in \cite{Marsden}, the $4$-tensor
  $\mathbf{a}$ admits progressive elastic plane waves in all possible
  directions if and only if $\mathbf{a}$ is strongly elliptic.}, this
only gives $H^1_0$-coercivity, and is not sufficient to obtain $H^1$
coercivity, as shown by the counterexample in \cite{Valette}. This
necessitates a stronger assumption of pointwise stability in the solid
regions. On the other hand, the fluid regions do not support shear
waves, and in its current form, $(\Lambda^{T^0} : \nabla u) : \nabla
\overline{u}$ in the fluid regions is not an expression of definite
sign. This issue is rectified by using an idea of \cite{Valette} to
modify the problem, so that in the fluid regions, the highest order
terms occurring in the weak formulation take the form
$$
   ( \Lambda^{T^0} : \nabla u ) : \nabla \overline{v}
     + \sigma_N  \nabla u : \nabla \overline{v}^T
         - \sigma_N (\nabla\cdot u)( \nabla \cdot \overline{v})
   = (\text{positive constant}) (\nabla \cdot u)
           (\nabla \cdot \overline{v}) .
$$
Here, $p^0$ is the initial hydrostatic pressure given in
\eqref{InHP}, and $\sigma_N$ is any regular scalar function which is equal to
$-p^0$ in $\Omega^F$, and $0$ outside of a small neighborhood of
$\Omega^F$. This neighborhood is sufficiently small so that none of
the solid-solid interfaces intersect the support of $\sigma_N$. With
this modification, one can assure positivity of the highest order
terms in the fluid regions, although we must compensate this by adding
hypotheses on the size of $p^0$ near the fluid-solid boundary.

In short, $a_{\text{original}}$ is not coercive and not closed, hence
the main technical challenge of the problem lies in finding the
`proper' variational formulation of the problem, one whose associated
sesquilinear form is coercive. Using the idea of \cite{Valette} described in the previous paragraph, one
derives and works with an equivalent problem of the form
$$
   \dfrac{d}{dt} \int_{\tilde{X}} \rho^0 \dot{u}\cdot \overline{v} \,
   dV + \int_{\tilde{X}} 2 \rho^0 R_{\Omega} \dot{u}
   \cdot \overline{v} \, dV + a_2\big(u, v\big)
   = \int_{\tilde{X}} \rho^0 f \cdot \overline{v} \ dV,
$$ 
where the sesquilinear form $a_2$, whose definition is given in
\eqref{a2} and \eqref{DefE}, is Hermitian, closed, and coercive in a
Hilbert space $E$ relative to $H = L^2(\tilde{X}, \rho^0 \, d
V)$. Furthermore, $E$ is dense in $H$. As a result, the `equivalent'
variational problem can be formulated in the setting of
$E\hookrightarrow H\hookrightarrow E'$ ($E'$ here denotes the dual space of $E$ in which $H$ is identified as a subspace by $H \ni h \mapsto (h,\cdot)_H \in E'$) on $(0,T)$ for $T > 0$ as
follows,
\begin{equation}\label{VartionalP2}
\begin{aligned}
& \text{Given} \ g \in E , h \in H,  f\in L^2(0,T; H) \\[0.2cm] 
& \text{Find}\ u\ \text{satisfying}:
     u \in \mathcal{C}^0( [0,T]; E), \dot{u} \in \mathcal{C}^0([0,T]; H) ,
  \text{such that} \\
& \forall v \in E , \dfrac{d}{dt} (\dot{u}, v)_H
  + (2R_{\Omega} \dot{u} , v)_H + a_2(u,v) = (f, v)_H ,\
  \text{in}\ \mathcal{D}'(0,T) ; \\
& u(0) = g\ ;\ \dot{u}(0) = h . 
\end{aligned}
\end{equation}
In subsection~\ref{A2Op}, we show that $a_2$ corresponds to a
self-adjoint unbounded operator $(A_2, D(A_2))$ densely defined on $H$
with the property that $A_2 \in \mathcal{L}(E,E')$. Using the operator $A_2$ we can rewrite the main
equation in \eqref{VartionalP2} in its equivalent form
$$ 
   \forall v \in E ,  \dfrac{d}{dt} (\dot{u}, v)_H
        + (2R_{\Omega} \dot{u}, v)_H + \langle A_2 u, v\rangle_{E',E} = (f, v)_H ,\
   \text{in}\ \mathcal{D}'(0,T) ,
$$
or simply
\begin{equation}\label{IVP}
    \ddot{u} + 2 R_{\Omega} \dot{u} + A_2 u = f ,\
   \text{in}\ \mathcal{D}'(0,T; E') .
\end{equation}
One can consider \eqref{IVP} as a single equation of second order in time, or as a system of two equations both first-order in time. In the first approach, it is possible to study Problem \eqref{IVP} directly using the Galerkin
approximation and `parabolic regularisation' for the second order
evolution equation, cf.\cite{LionsMagenesV1}. In the second approach,
one studies the well-posedness of the associated first-order problem
obtained from \eqref{IVP} by reduction of order. The associated first
order problem is formulated on $\mathcal{D}'(0,T; H \times E')$ as
$$
   \dfrac{d}{dt} \begin{pmatrix} u \\ \dot{u} \end{pmatrix}
   - \tilde{A}_2 \begin{pmatrix} u \\ \dot{u} \end{pmatrix}
   = \begin{pmatrix} 0 \\ f\end{pmatrix};\ \
   \tilde{A}_2 := \begin{pmatrix}  0 & \Id \\ -A_2  & -2R_{\Omega}
                                   \end{pmatrix} .
$$
Here $\tilde{A}_2$ makes sense as a densely defined operator on
$\mathcal{H} = E\times H$ with domain $D(\tilde{A}_2) = D(A_2) \times
E$. Well-posedness results for this system can be obtained either by
semi-group theory, cf. \cite{DautrayLionsV5}, or by the Galerkin method
and a `regularisation' technique for first-order evolution equations,
cf. \cite{LionsMagenesV1}. The main tool in semigroup theory we will
use is a corollary of the Hille-Yosida theorem, the Hille-Phillips
theorem, cf. \cite{HillePhilips}.  From either approach, one can arrive at
the following well-posedness result.

\medskip

\begin{theorem}\label{maintheorem} {\it Suppose that the hypotheses of
  Theorem~\ref{CoerciveA2} are satisfied. For $g \in E$, $h\in
  L^2(\widetilde{X}, \rho^0\,dx)$ and $f \in L^2(0,T; E) $ there
  exists a unique solution $u$ to the problem \eqref{VartionalP2} with
$$
   u \in \mathcal{C}^0([0,T]; E) ,\ \dot{u} \in  \mathcal{C}^0([0,T];H ) .
$$
This solution is a called a weak solution of 
\begin{equation}\label{MainProblem}
   \ddot{u} + 2 R_{\Omega} \dot{u} + A_2 u = f ,\
   u(0) = g ,\ \dot{u}(0) = h .
\end{equation}
If $f \in H^1(0,T, H), g \in D(A_2) $ and $h\in E$, the
unique variational solution $u$ obtains more regularity, $u \in
\mathcal{C}^1([0,T], H) \,\cap \,\mathcal{C}^2((0,T), H)$, and solves
\eqref{MainProblem} in a strong (classical) sense that is, in
$\mathcal{C}^0([0,T], H)$. In addition, if we denote by $$U(t)
= \begin{pmatrix} u \\ \dot{u}\end{pmatrix} ,$$ then $U(t)$ is given by the formula
\eqref{VectorVSol}.}
\end{theorem}
\medskip

In order to leverage the existing theory in functional analysis
relating to the well-posedness of evolution equations, and then
leading to Galerkin type numerical schemes, it is critical to consider
carefully the spaces on which the variational formulation is
defined. We comment that Wahr (\cite{Wahr}) also applied some results
of functional analysis to find time harmonic solutions for
\eqref{globalseis} including also a body force and traction force at
the surface. In that work the author considers a quotient space which
removes the zero frequency eigenfunctions of the operator $A_2$. The
Hilbert space $E$ which we introduce here, coming from an idea in
\cite{Valette}, is different from the space and inner product used in
\cite{Wahr} and obviates the need to take the quotient. We also add
that there are technical problems with the inner product constructed
in \cite{Wahr} beginning with the claim that the product may be
extended from the domain of the operator $A_2$ to all of $L^2 \times
L^2$.

In addition to showing the proof of Theorem \ref{maintheorem} using the semi-group approach, we also in section \ref{Galerkin1} prove that, under some hypotheses on regularity, Galerkin approximations for the solution of problem \eqref{VartionalP2} converge strongly. Also, in section \ref{Voltsec} we provide another approach to well-posedness using a Volterra equation and viewing the terms in \eqref{VartionalP2} of less than second order as perturbations. This has the advantage in potential numerical applications that the non-local self-gravitation terms are treated separately.

The outline of the remainder of the paper is as follows. In section \ref{IntroEq} we discuss in detail all terms appearing in equation \eqref{globalseis} including the physical meaning of each term. We also discuss the boundary conditions applied at all of the interfaces. In section \ref{Prelimsec} we introduce rigorously preliminary formulations of the problem \eqref{globalseis} and the initial weak formulation given by the bilinear form $a_{\text{original}}$ as described above. Then in section \ref{weak:sec} we go through the detailed technical conversion of the preliminary weak formulation using $a_{\text{original}}$ to the formulation using $a_2$ for which we can prove well-posedness. Section \ref{functsec} contains some background on the abstract framework in which we can apply functional analysis to prove well-posedness, and establishes coercivity of the form $a_2$. The next section, section \ref{semigroup:sec}, then applies semi-group theory to prove the well-posedness, and in particular to prove Theorem \ref{maintheorem} quoted above. In section \ref{Galerkin1} we establish the convergence of Galerkin approximations, and then in section \ref{Voltsec} we study the Volterra equation approach described above. There are also several appendices concerning conservation of the physical energy, proof of a lemma on regularity of vector valued distributions, and some geometrical background which is important for the boundary conditions at the interfaces.

\section{Mathematical model of rotating and self-gravitating earth}\label{IntroEq}

In this section we will discuss the equation \eqref{globalseis} giving details on the physical meaning of each term, as well as details on the boundary conditions. However, in order to motivate the application, we first give a small amount of background in seismology using a recent
Nepal earthquake (Mw 7.8; 2015, April 25). The equation \eqref{globalseis} is intended to model the oscillations of the earth resulting from such a source, and in particular allow modelling of the seismic waves measured in a seismogram. In a seismogram, one displays the three components of displacement at a particular station on earth's surface. An example of a single station seismogram and one component of seismograms for different
epicentral distances are shown in Figure~\ref{fig:2}.
\begin{figure}
\centering
\hspace*{0.2cm}
\includegraphics[width=105.5mm]{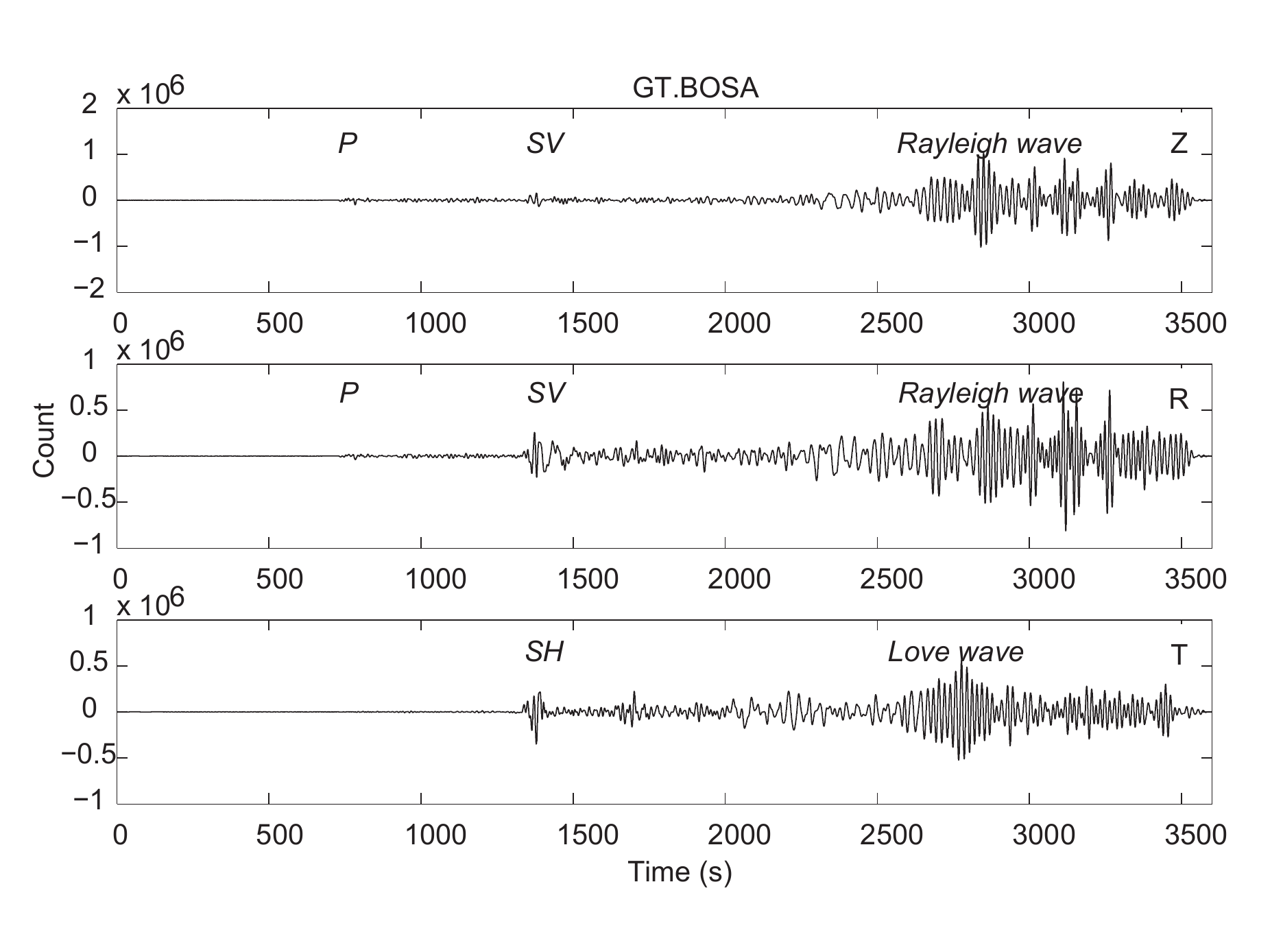} \\
\includegraphics[width=100mm]{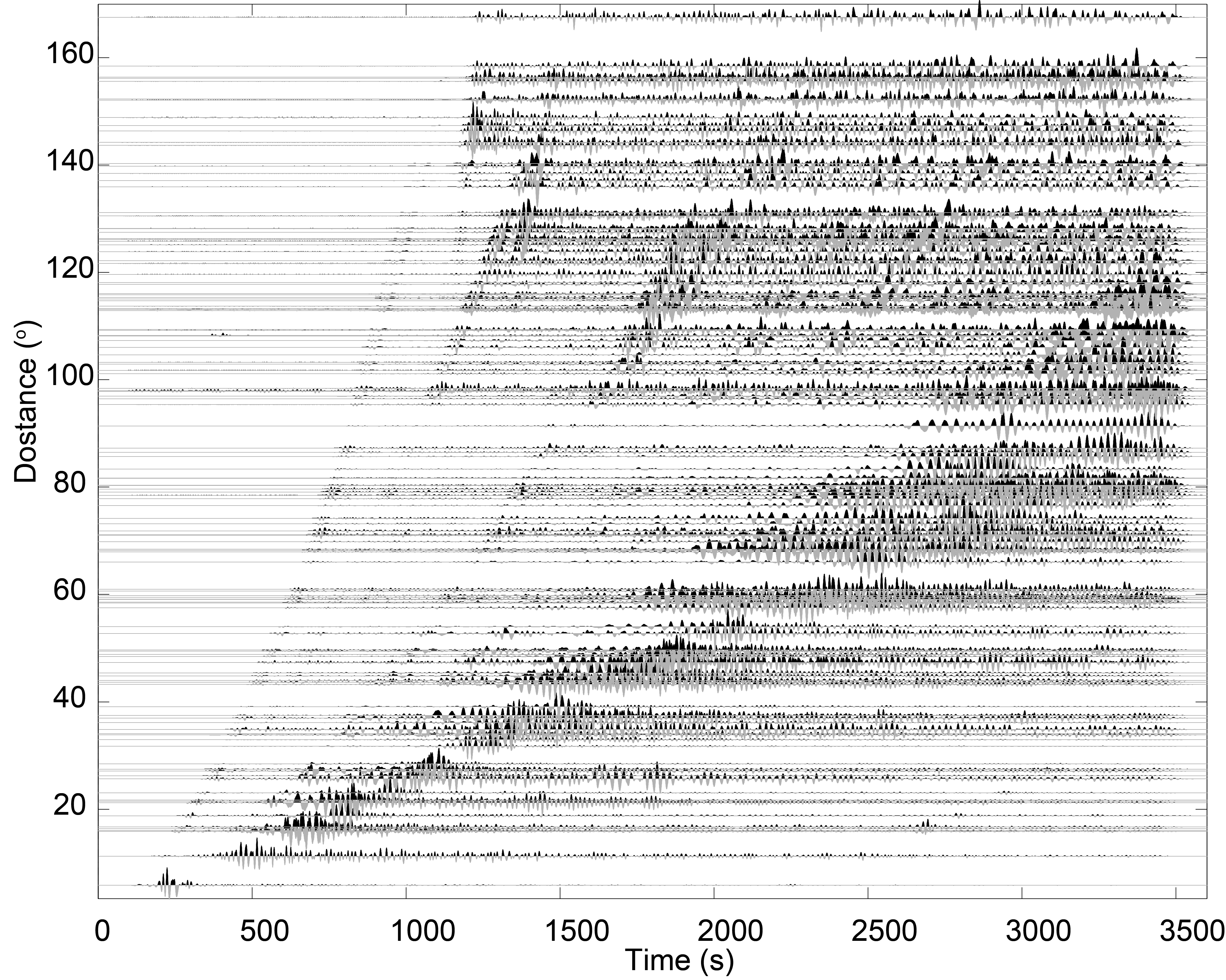}
\caption{Top: A seismogram, that is, three (vertical ($Z$), radial
  ($R$) and transverse ($T$)) components of the displacement as a
  function of time. We indicate the presence of (\textit{P}- and
  \text{S}-polarized) body waves and (Rayleigh and Love) surface
  waves. Bottom: $Z$ component of displacement as a function of time
  and angular epicentral distance (acknowledgment: Chunquan Yu).}
\label{fig:2}
\end{figure}
We have also indicated
some body-wave and surface-wave phases.



\subsection{System of elastic-gravitational equations} 

We now begin defining more precisely the elastic-gravitational equation which we plan to study in the remainder of the paper. Denote by $u = u(t,x)$ the displacement which takes values in
$\mathbb{C}^3$ (the physical displacement being the real part of $u$). For $f \in L^2( (-T,T) \times \tilde{X}, \mathbb{C}^3 )$ the basic system of equations describing free
oscillations of the earth is
\begin{equation}
\label{eq:systemu}
  \rho^0 \left [ \ddot u + 2 R_{\Omega} \cdot \dot u \right ]
     + \rho^0 u \cdot \nabla\nabla (\Phi^0 + \Psi^s)
   + \rho^0 \nabla S(u) - \nabla \cdot (\La^{T^0} : \nabla u) = \rho^0 f ;
\end{equation}
where
$$(\La^{T^0} : \nabla u)_{ij} =
   \sum_{k, l = 1}^3 \La^{T^0}_{ijkl} \partial_k u_l = \sum_{k, l =
   1}^3 \La^{T^0}_{jikl} \partial_l u_k,$$
and
$$(\nabla\nabla (\Phi^0 + \Psi^s))_{ij} = \partial_i
  \partial_j (\Phi^0 + \Psi^s)$$
while
\begin{equation}\label{ROm}
   R_{\Omega} \cdot \dot{u} = \Omega \times \dot{u} 
\quad\text{with}\
   R_{\Omega} := (\sum_{j=1}^3 \epsilon_{ijk} \Omega_j )_{i,k=1}^3 .
\end{equation}
We describe below the physical meaning of the parameters $\rho^0$,
$R_{\Omega}$, $\Lambda^{T^0}$, potentials $\Phi^0$ and $\Psi^s$, and operator
$S$ appearing in \eqref{eq:systemu}.

\subsection{Earth's rotation}

$\Omega \in \RR^3$ is the angular velocity of the earth's rotation and $R_{\Omega}\cdot
\dot{u}$ represents the induced Coriolis force. $\Psi^s(x)$ is the
corresponding (spatial) centrifugal potential with
\begin{equation}\label{centriPo}
   \Psi^s(x) := -\tfrac{1}{2} \left (\Omega^2x^2-(\Omega\cdot x)^2 \right ) .
\end{equation}


\subsection{Initial prestressed state}

$\Phi^0$ is the reference gravitational potential and $\rho^0$ the reference density. We are considering oscillations about the reference state corresponding to these quantities which satisfy the relation
\begin{equation}\label{ellPhi}
\Delta \Phi^0 = 4 \pi G \rho^0
\end{equation}
where $G$ is the gravitational constant. We assume that $\rho^0 \in L^\infty(\tilde{X})$ and thus $\Phi^0 \in H^2(\RR^3)$ by elliptic regularity. In fact for well-posedness $\rho^0$ is required to be in the space $W^{1,\infty}(\tilde{X} \setminus \Sigma)$ and to be bounded from below by a positive constant. $W^{1,\infty}$ is the space of $C^0$ functions whose weak gradient is in $L^\infty$, or equivalently the space of uniformly Lipschitz functions. Thus $W^{1,\infty}(\tilde{X} \setminus \Sigma)$ is the space of functions which are uniformly Lipschitz in $\tilde{X}$ except for possibly having jumps across some of the interfaces in $\Sigma$.

Making use of the Green's function
\[
E_3(x) = -\frac{1}{4 \pi |x|}
\]
for the Laplacian in three dimensions $\Phi^0$ may be written explicitly as
\begin{equation}
\label{eq:Phi0}
   \Phi^0 = 4\pi G E_3 \ast \rho^0
\end{equation}
in $\RR^3$. Since $\Phi^0 \in H^2(\RR^3)$, $\Phi^0$ is continuous across all of the boundaries $\Sigma$. The sum $\Phi^0 + \Psi^s$ is referred to as the geopotential. 

Denote by $p^0$ the initial hydrostatic pressure, 
 \begin{equation}\label{InHP}
 p^0:= \begin{cases}  \text{hydrostatic pressure} & \text{in}\ \Omega^F\\ -\tfrac{1}{3} \tr T^0  & \text{in}\ \Omega^S \end{cases}
 \end{equation} 
 and by $T^0$ the initial static stress 
 \begin{equation}\label{staticStress}
  T^0 = \begin{cases}
 -p^0 \Id & \text{in}\ \Omega^F \\
 -p^0 \Id + \tau^0  & \text{in}\ \Omega^S \end{cases}.
 \end{equation}
 $T^0$ has the symmetry
$$T^0_{ij} = T^0_{ji}.$$
Note that \eqref{staticStress} decomposes $T^0$ into its isotropic and deviatoric parts which are respectively $- p^0 \mathrm{Id}$ and $\tau^0$, and that from these definitions $\tr \tau^0 = 0 $. It is important to note that \eqref{staticStress} includes the physical assumption that the prestress is hydrostatic in $\Omega^F$. 

\subsection{Mechanical equilibrium} For a uniformly rotating earth model
prior to the occurrence of an earthquake the earth is assumed to be in
a state of mechanical equilibrium, that is, at rest with respect to a
set of Cartesian coordinates $x\in \mathbb{R}^3$ which are rotating
uniformly with angular velocity $\Omega$ \cite{Dahlen}. The mechanical
equilibrium condition is given by the static momentum equation,
satisfied throughout $\Omega^S$ and $\Omega^F$.
 \begin{equation}\label{Equi1}
 \text{Mechanical equilibrium}:\ \ \nabla \cdot T^0 = \rho^0 \nabla ( \Phi^0 + \Psi^s) =:  \rho^0 g_0'.
\end{equation}
Here we are making the definition $g_0' := \nabla (\Phi^0 + \Psi^s)$, and we remind the reader that $\Phi^0$ is the gravitational potential of the reference state given by \eqref{eq:Phi0} and $\Psi^s$ is the centrifugal potential given by \eqref{centriPo}. It is important to note that not all components of the deviatoric initial static stress, $\tau_0$, in the solid regions are determined by \eqref{Equi1}. Indeed, the equations (with appropriate boundary conditions given by \eqref{TractCont} below) only constrain three out of six independent components of $T_0$. In the fluid region the static momentum equation \eqref{Equi1} assumes the following form:
\begin{equation}\label{Equi2}
\text{Hydrostatic equilibrium in}\ \Omega^F:\ \ \ \  \nabla p^0 = - \rho^0 g_0' .
 \end{equation}
Taking the limit at the boundaries and interfaces the equilibrium
conditions take the form of the
\begin{equation}\label{TractCont}
\text{Traction Continuity Condition}:\ \begin{cases}
  \partial \tilde{X}&: \ \ \nu \cdot T^0 = 0 \\
 \Sigma^{SS} \cup \Sigma^{FF} \cup \Sigma^{FS} \ \ &:\ \ [\nu \cdot T^0]^+_- = 0 
\end{cases}
\end{equation}
where $\nu$ is a unit normal to the relevant surface oriented from the ``negative side" to the ``positive side." The notation $[\cdot]_-^+$ indicates the difference between the limits from each size of an interface (that is, the limit from the positive side minus the limit from the negative side). For the interior interfaces a choice of which side is positive and which is negative must be made for every interface in a consistent way, but the boundary conditions do not depend on these choices.
Along $\Sigma^{FF}$ and $\Sigma^{SS}$, the choice we will take is so that the normal vector fields along these interfaces point outward. For the exterior interfaces (that is, along $\partial \tilde{X}$) we take the interior of $\tilde{X}$ to be the negative side and the exterior to be the positive side so that $\nu$ is the outward pointing unit normal vector on $\partial \tilde{X}$. For the fluid-solid interface $\Sigma^{FS}$ we take the positive side to be the solid region and the negative the fluid region so that $\nu$ points from the fluid toward the solid.

In the next lemma we collect a few properties of the prestress in the
fluid region which follow from \eqref{Equi2}. We assume that $\nabla
\rho^0$ is non-vanishing.

\medskip

\begin{lemma} \label{parallellemma}
If $\rho_0$, $p^0$, and $g_0'$ are in $C^1$ up to the boundary on each component of $\Omega^F$ and satisfy \eqref{Equi2} in $\Omega^F$, then we have
\begin{equation}\label{EquiFluid}
\nabla \rho^0 \big|\big| g_0' \big|\big| \nabla p^0
\end{equation}
on $\Omega^F$. The notation $\big|\big|$ means that the two vectors are parallel. On any $C^1$ portion of $\Sigma^{FF}$ across which $\rho^0$ is not continuous
\[
\nabla \rho^0_{\pm} \big|\big| \nabla p^0_\pm \big|\big| (g_0')_\pm \big|\big| \nu
\]
where $\nabla \rho^0_\pm$ mean respectively the limit of $\nabla \rho^0$ from either the positive of negative side of $\Sigma^{FF}$.
\end{lemma}
\medskip

\begin{proof}
From the identity $\curl ( \nabla  f)  = 0$ and the definition $g_0' = \nabla (\Phi^0 + \Psi^s)$ we see that $\curl (g_0' )= 0$ and so
$$\curl \big( \rho^0 g_0' \big)
=  \nabla \rho^0 \times g_0'. $$
As a result, the equilibrium condition \eqref{Equi2} gives $0 = \nabla \rho^0 \times g_0'$, and so $\nabla \rho^0 \big|\big| \, g_0'$. Also, we have directly from \eqref{Equi2}
\[
\nabla p^0 \times g_0' = 0
\]
from which it follows that $\nabla \rho^0 \big|\big| \nabla p^0$.

To prove the second claim take any $x_0 \in \Sigma^{FF}$ such that $\Sigma^{FF}$ is $C^1$ in a neighborhood of $x_0$. Then we can find a set of local coordinates $\{x^j\}_{j=1}^3$ near $x_0$ such that locally $\Sigma^{FF} = \{x^3 = 0\}$ and the positive side is given by $\{ x^3 > 0\}$. Let $\varphi(x^1,x^2) \in C_c^\infty(\RR^2)$ and suppose that $\tilde{\varphi} \in C_c^\infty(\RR)$ is an odd function such $-\int_{-\infty}^0 \tilde{\varphi} = \int_{0}^\infty \tilde{\varphi} = 1$. Then for any locally defined vector field $v = v^1 \partial_{x^1} + v^2 \partial_{x^2}$ define the following vector fields using the local coordinates
\[
u_\epsilon = u^j_\epsilon \partial_{x^j}  = \frac{v^j}{\epsilon}  \varphi\left ( x^1, x^2 \right )  \tilde{\varphi} \left ( \frac{x^3}{\epsilon} \right )\partial_{x^j}.
\]
In the coordinates \eqref{Equi2} implies
\[
-\int_{\RR^3} (\partial_{x^j} p^0) u_\epsilon^j \ \sqrt{|e|} \ \mathrm{d} x = \int_{\RR^3} \rho^0 (g_0')_j u_\epsilon^j   \ \sqrt{|e|} \ \mathrm{d} x
\]
where $\sqrt{|e|}\ \mathrm{d} x$ is the volume form of the Euclidean metric in these coordinates. Using the continuity of $g_0'$, the right hand side is equal to
\[
\begin{split}
\int_{\RR^2} &\left ( \int_{\RR} \rho^0 (g_0')_j \frac{v^j}{\epsilon}   \varphi\left ( x^1, x^2 \right )  \tilde{\varphi} \left ( \frac{x^3}{\epsilon} \right ) \ \mathrm{d}x^3 \right ) \ \mathrm{d} x^1\ \mathrm{d} x^2\\
&\hskip1in  \underbrace{\longrightarrow}_{\epsilon \rightarrow 0^+} \ \int_{\{x^3 = 0\}} [\rho^0]_-^+ (g_0')_j v^j  \varphi\left ( x^1, x^2 \right ) \sqrt{|e|} \ \mathrm{d} x^1 \ \mathrm{d} x^2.
\end{split}
\]
Using integration by parts and the continuity of $p^0$, we find that the left hand side equals
\[
\int_{\RR^2} \left ( \int_{\RR} p^0  \partial_{x^j} (u_\epsilon^j \ \sqrt{|e|}) \ \mathrm{d}x^3 \right ) \ \mathrm{d} x^1\ \mathrm{d} x^2 \underbrace{\longrightarrow}_{\epsilon \rightarrow 0^+} 0.
\]
Since $\varphi \in C_c^\infty(\RR^2)$ is arbitrary and $[\rho^0]_-^+$ is not zero this implies that $v^j (g_0')_j = 0$, and so $g_0' \big|\big| \nu$. Combining this with the first part of the lemma the second claim is proven.
\end{proof}
\medskip

\subsection{Perturbation of gravitational potential} $S(u)$ denotes the perturbation, also written as $\Phi^1$, of the
gravitational potential caused by the redistribution of mass. This is the Eulerian perturbation of the Newtonian potential associated to
the field of displacement $u$. We have
\begin{equation}
\label{eq:Phi1sol}
   S(u) := E_3 \ast (-4\pi G \nabla \cdot (\rho^0 u)).
\end{equation}
Note that the divergence in this formula is taken in the weak sense since $\rho^0$ may not be continuous across the interfaces $\Sigma$. We then have that
\[
\Delta S(u) = -4 \pi G \nabla \cdot (\rho^0 u).
\]
The operator $S$ so defined is a pseudodifferential operator on $\RR^3$ of order $-1$.

\subsection{First Piola-Kirchhoff stress}

In \eqref{eq:systemu} $\La^{T^0}$ is the modified stiffness tensor defined by
\begin{equation}\label{ModStiffTensor}
   \Lambda^{T^0}_{ijkl} = \Xi_{ijkl} + T^0_{ik} \delta_{jl}
\end{equation}
where $T^0$ is the initial static stress appearing in \eqref{staticStress} and $\Xi_{ijkl}\in L^{\infty}(\tilde{X})$ is the stiffness tensor coming from the linearization of the constitutive function. The stiffness tensor possesses the classical symmetries \cite{Dahlen}
\begin{equation} \label{SymElastic}
   \Xi_{ijkl} = \Xi_{jikl} = \Xi_{ijlk} = \Xi_{klij}
\end{equation}
On the other
hand the first Piola-Kirchhoff stress tensor, $T^{PK1}$, satisfying
 $$T^{PK1} = \Lambda^{T^0}
: \nabla u $$ is not symmetric which reflects on the fact that the
invariant definition of $T^{PK1}$ is actually as a two-point tensor
(see \cite{Marsden}).  In fact, following the discussion in
\cite[Section 3.6.2]{Dahlen}, one can introduce the alternate representations,
\begin{equation}\label{ElasticTensor}
   \Lambda^{T^0}_{ijkl} = 
      \Gamma_{ijkl} + a (T^0_{ij} \delta_{kl} + T^0_{kl} \delta_{ij})
      + (1 + b) T^0_{ik} \delta_{jl}
      + b (T^0_{jk} \delta_{il} + T^0_{il} \delta_{jk}
                    + T^0_{jl} \delta_{ik}).
\end{equation}
Each choice of scalars $a,b$ defines a possible tensor $\Gamma$ possessing the symmetries \eqref{SymElastic}. $\Xi$ in \eqref{ModStiffTensor} is the elastic tensor with $a = b = 0$, which is also the choice of \cite{Valette}. Another choice adopted by \cite{Dahlen_1972} is $a = \tfrac{1}{2}, b =-\tfrac{1}{2}$. We use $\Gamma$ to denote from now on this choice of elasticity tensor (that is, with $a = -b = \tfrac{1}{2}$) so that the modified stiffness tensor is given by
\begin{equation}\label{NotationRel}
   \Lambda^{T^0}_{ijkl} = 
      \Gamma_{ijkl} + \tfrac{1}{2} (
      T^0_{ij} \delta_{kl} + T^0_{kl} \delta_{ij}
      + T^0_{ik} \delta_{jl}
      - T^0_{jk} \delta_{il} - T^0_{il} \delta_{jk}
                    - T^0_{jl} \delta_{ik})
\end{equation}

Now, the definition of an isotropic solid given in \cite{Dahlen} is as follows
\medskip

\begin{definition}[Isotropic solid]
An isotropic solid is one whose elasticity tensor 
is of the form 
\begin{equation}\label{isotropicEl}
\Gamma_{ijkl} = (\kappa - \tfrac{2}{3} \mu) \, \delta_{ij} \delta_{kl}
      + \mu \, (\delta_{ik} \delta_{jl} + \delta_{il} \delta_{jk}),
      \end{equation}
      where $\kappa$ is the isentropic incompressibility (or bulk modulus) and $\mu$ is the rigidity (or shear modulus). 
\end{definition}
\medskip

In the fluid regions, $\Omega^F$, of the earth $\Gamma$ is isotropic and the rigidity is identically zero so we have
     \begin{equation}\label{isoFluid}
      \Gamma_{ijkl} = \kappa \, \delta_{ij} \delta_{kl}.
     \end{equation} 
Using \eqref{isoFluid} and the relationship between $\Xi_{ijkl}$ and $\Gamma_{ijkl}$, which can be found by equating the right hand sides of \eqref{ModStiffTensor} and \eqref{NotationRel}, we obtain
\begin{equation}\label{PerFluid}
\begin{aligned}
\text{Perfect fluid}\ \Omega^F:\ \ \Xi_{ijkl} 
&= -p^0 \, (\de_{ij} \delta_{kl}              - \de_{jk} \delta_{il} - \delta_{ik}\delta_{jl})
   +  \kappa \, \delta_{ij} \delta_{kl}\\
&= p^0 (\gamma-1) \delta_{ij}\delta_{kl} + p^0 \delta_{ik} \delta_{jl} + p^0 \delta_{jk} \delta_{il},
\end{aligned}
\end{equation}
where $\gamma$ is the adiabatic index of the fluid. Using \eqref{PerFluid} we also find that in the fluid regions
\begin{equation}\label{TPK1}
T^{PK1}_{ij} = p^0 (\gamma - 1) \delta_{ij} (\nabla \cdot u) + p^0 (\nabla u)_{ij}.
\end{equation}


%
%


\subsection{Boundary conditions}\label{BoCo}\
The equations of motion \eqref{eq:systemu} are accompanied by linearized kinematic, dynamic and gravitational conditions on the boundaries and interfaces $\Sigma = \partial \tilde{X} \cup \Sigma^{SS} \cup \Sigma^{FF} \cup \Sigma^{FS}$.
The discussion here follows partly from \cite[Section 3.4]{Dahlen} although we will use \cite{Valette} for the dynamic boundary condition along $\Sigma^{FS}$, which is \eqref{BC1}. We also comment that the boundaries are required to have at least $C^1$ regularity.

First we specify some notation. The jump across a boundary between two regions $\Omega_-$ and $\Omega^+$ will be written as $[ u]^+_-:= u_+  - u_- $ where $\nu$ is the unit normal oriented from $\Omega_-$ to $\Omega^+$. Along $\Sigma^{FS}$, we chose the unit normal $\nu$ that points from $\Omega^F$ to $\Omega^S$ so in this case $\Omega^S$ is $\Omega^+$ and $\Omega^F$ is $\Omega_-$. On the earth's free surface, $\partial \tilde{X}$, $\nu$ will denote the outward pointing unit normal.

\begin{enumerate}
\item The Kinematic Boundary Conditions require that there is no slip along the welded solid-solid interfaces which means that
\begin{equation}\label{firstCon}
[u]^+_- = 0 \ \text{across $\Sigma^{SS}$}. 
\end{equation}
Along the fluid-solid and fluid-fluid interfaces tangential slip is allowed but it is required that there is no separation or interpenetration \cite{Dahlen}. This is assured by the linearized continuity condition 
\begin{equation} \label{FcontinuityBC}
[u\cdot \nu]^+_- = 0  \ \text{across $\Sigma^{F} = \Sigma^{FF} \cup \Sigma^{FS}$}.
\end{equation}
We call this the first-order tangential slip condition.
\item The Dynamic Boundary Conditions require that juxtaposed particles on either side of a welded or solid-solid boundary at time $t=0$ must remain juxtaposed \cite{Dahlen}. This condition in terms of $T^{PK1}$ is 
\[
[\nu \cdot  T^{PK1}]^+_-  = 0 \ \text{across $\Sigma^{SS}$}.
\]
On the outer free surface $\partial \tilde{X}$ 
\begin{equation}\label{lastCon}
\nu \cdot  T^{PK1} = 0 .
\end{equation}
To model the case in which there is an applied traction force at the surface the right hand side of \eqref{lastCon} can be made nonzero although we will not consider this here. Along $\Sigma^{FS}$ and $\Sigma^{FF}$, since there may be tangential slip, juxtaposed particles on either side of the boundary need not remain juxtaposed after deformation. However, it is required that there is no shear traction along $\Sigma^{F} = \Sigma^{FF} \cup \Sigma^{FS}$. To model this requirement we use the condition\footnote{By \cite{Valette}, \eqref{BC1} is equivalent to the boundary conditions along $\Sigma^{F}$ used in \cite{WoodhouseDahlen78}.}
\begin{equation}\label{BC1}
[\nu \cdot  T^{PK1}]^+_- 
 = - \nu \nabla^{\Sigma} \cdot ( p^0 [ u]^+_-) -p^0 W [u]^+_-
\end{equation}
where $\nabla^\Sigma \cdot$ is the surface divergence and $W$ is the Weingarten operator for the surface (see Appendix~\ref{Wein} for the definitions). We comment that \eqref{BC1} corresponds precisely with formula (3.81) in \cite{Dahlen}. Furthermore, \cite{Dahlen} includes an extra condition at the fluid-solid boundary given by \cite[Formula (3.82)]{Dahlen}. It can be checked that this extra condition is automatically satisfied when $\Xi_{ijkl}$ takes the form \eqref{PerFluid} in the fluid region.

\item Gravitational Boundary Conditions: The following continuity conditions are satisfied on all $\Sigma = \partial \tilde{X} \cup \Sigma^{SS} \cup \Sigma^{FF} \cup \Sigma^{FS}$,
\[
\begin{aligned}
\big[S(u)\big]^+_- = 0 ,\\
\Big[ \nabla S(u) \cdot \nu + 4\pi G  \rho^0  u \cdot \nu\Big]^+_- = 0 .
\end{aligned}
\]
\end{enumerate}
For a summary of all the boundary conditions including the conditions
\eqref{firstCon} to \eqref{lastCon} and the traction continuity
condition at the boundaries \eqref{TractCont} see
table~\ref{unT}.

\begin{table}
\caption{Linearized Boundary Conditions satisfied by $u$ and $T^0$ }
\centering
\begin{tabular}{l | r}
\textbf{Boundary Type} & \textbf{Linearized Boundary Conditions}\\ \hline \\[-0.2cm]
 Earth's free surface, $\partial \tilde{X}$ & $ T^0 \cdot \nu = 0$ \\ &$ \nu \cdot T^{PK1} =0 $\\ \hline \\[-0.2cm]
 Solid - Solid,  $\Sigma^{SS}$ & $[T^0 \cdot \nu ]^+_- = 0 $ \\
  & $[ \nu \cdot T^{PK1}]^+_- = 0 ;\ [u]^+_- = 0$\\ \hline \\[-0.2cm]
  Fluid - Fluid, $\Sigma^{FF}$ & $ [T^0 \cdot \nu ]^+_- = 0 $ \\
   & $ [\nu \cdot  T^{PK1}]^+_-  = - \nu \nabla^{\Sigma} \cdot ( p^0 [ u]^+_-) -p^0 W [u]^+_-;\ \ [u \cdot \nu]^+_- = 0  $\\ \hline \\[-0.2cm]
  Fluid - Solid, $\Sigma^{FS}$ & $ [T^0 \cdot \nu]^+_- = 0 $ \\
  & $ [\nu \cdot  T^{PK1}]^+_-  = - \nu \nabla^{\Sigma} \cdot ( p^0 [ u]^+_-) -p^0 W [u]^+_-;\ \ [u \cdot \nu]^+_- = 0  $ \\ \hline \\[-0.3cm]
 All boundaries and interfaces $\Sigma$ & $\big[S(u)\big]^+_- = 0;\ \ \Big[ \nabla S(u) \cdot \nu + 4\pi G  \rho^0  u \cdot \nu\Big]^+_- = 0$
 \end{tabular}
 \label{unT}
\end{table}

\subsection{Representation of the source} \label{ssec:2.7}

A source typically represents a rupture process. The rupture process
involves a complicated slip function that is variable in space and
time. To infer source parameters, one approximates the rupture as a
constant slip on a geometrically flat fault. The faulting is then
approximated by a double couple of equivalent body forces. (Processes
described by single couples would generate large torques and thus
affect the rotation of the earth. Double sets of couples do not
generate a net torque.) For background on modeling earthquake sources see for example \cite{SteinWys,AkiRichards}.

The equivalent body force takes the form
\[
   f_j(x,t) = -M_{ij} \partial_i \delta(x - \tilde{x})
                   H(t - \tilde{t})
            = -M_{ij} \partial_i \delta(x - \tilde{x})
                \partial_t^{-1} \delta(t - \tilde{t}) ,
\]
in which $\tilde{x}$ coincides with the epicenter and $\tilde{t}$ with
the origin time. The step function $H$ signifies an idealized time-rise
function, and the spatial delta function is also an idealised representation of a highly localised source. More generally,
\[
   f_j = -\partial_i S_{ij} ,
\]
where $S$ stands for the stress glut tensor. 
In the lowest-order far-field approximation,
\begin{equation}
   \partial_t S_{ij}(x,t) = M_{ij} \delta(x - \tilde{x})
                   \delta(t - \tilde{t}) .
\end{equation}
The seismic moment tensor $M$ is itself defined as the spatial and temporal integrals of the stress glut rate over the support of the source.

Combining force couples of different orientations into the seismic
moment tensor, $M$, gives a general description of various seismic
sources. A body force couple consists of two forces, $f$ say, oriented
in opposite directions acting together and offset, $\sigma$, in a direction
normal to the forces. If $f$ is aligned with the $x_1$-axis and $\sigma$
with the $x_2$-axis, then $M_{12}$ is given by $f_1 \sigma_2$; to model a
couple acting at a point, one takes the limit $\sigma_2 \downarrow 0$ such
that $f_1 \sigma_2$ stays constant. 

A component of the seismic moment tensor is written in the form
\begin{equation} \label{eq:2.26}
   M_{ij} = M_0 \, (n_i d_j + d_i n_j) ,
\end{equation}
where $n$ signifies the unit vector normal to the fault plane, and $d$
the unit slip vector. Naturally, this tensor is symmetric. The seismic
moment has a vanishing trace:
\[
   M_{ii} = 2 M_0 \, n_i d_i = 0 ,
\]
because the slip vector lies in the fault plane. Thus this moment
tensor is purely deviatoric. A non-vanishing trace implies a volume
change (explosion) which does not exist for a pure double-couple
source. Moreover, $\det M = 0$. Indeed, the vanishing of the trace and
determinant guarantee a representation of the form
(\ref{eq:2.26}). The so-called scalar moment $M_0$ represents the norm
of the moment tensor,
\[
   M_0 = \tfrac{1}{\sqrt{2}} (M_{ij} M_{ij})^{1/2} ;
\]
one also introduces $\hat{M}$ according to
\[
   M = \sqrt{2} M_0 \hat{M} \quad \mbox{so that} \quad \hat{M}_{ij} = \frac{1}{\sqrt{2}} (n_i d_j + d_i n_j).
\]

We now discuss how to recover the fault geometry from a moment
tensor. Three axes (in dimension three) play a role: the T, P, and
null axes oriented along $t$, $p$ and $b$. These vectors are
orthogonal to one another. In particular,
\[
   \tfrac{1}{2} \epsilon_{ijk} t_j p_k = -b_i .
\]
One has the relations
\begin{eqnarray}
   t_i &=& n_i + d_i ,
\label{eq:T}
\\
   p_i &=& n_i - d_i ,
\label{eq:P}
\\
   b_i &=& \epsilon_{ijk} n_j d_k.\nonumber
\end{eqnarray}
An explicit
calculation shows that $t$, $p$ and $b$ are eigenvectors of the
seismic moment tensor:
\[
   M_{ij} t_i = M_0 \, t_j ,\quad
   M_{ij} p_i = -M_0 \, p_j ,\quad
   M_{ij} b_i = 0 .
\]
One typically displays the moment tensor using ``beach balls''. We let
$\tilde{\nu}$ denote a unit vector on the sphere at the hypocenter;
one uses white in the beach ball plot if $-1 \le \tilde{\nu}_i
\hat{M}_{ij} \tilde{\nu}_j < 0$ and black if $0 < \tilde{\nu}_i
\hat{M}_{ij} \tilde{\nu}_j \le 1$. One orients the sphere using
(global) spherical polar coordinates: The radial, geocentric
colatitudal and longitudinal unit vectors are pointing up, south and
east, respectively. The standard method of display uses a zenithal
equal-area projection onto the plane orthogonal to the radial
direction at the hypocenter. We indicate in Figure \ref{fig:1a} such a
projection $P'$ of a direction $P$ identified with angles,
$(\theta,\psi)$, where $\psi$ is the angle in the mentioned plane
relative to the geographical North.

\begin{figure}
\centering
\includegraphics[width=110mm]{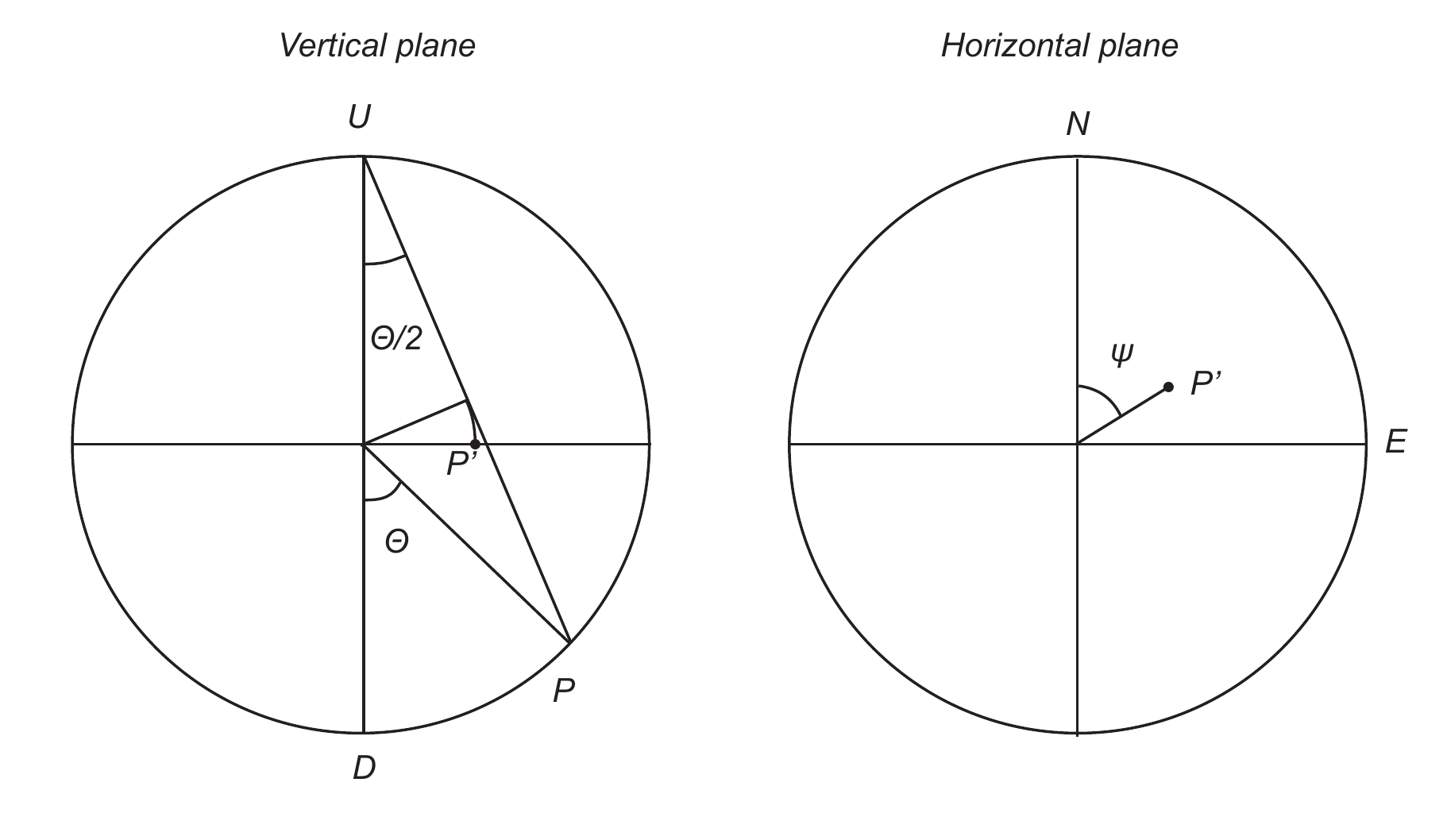}
\caption{Earthquake source: Method of display of a moment tensor on
  the focal sphere. Left: ``Vertical'' plane ($U$ points radially
  outwards from the center of the earth); right: ``Horizontal'' or tangential plane ($N$ points
  towards the geographical North). The point $P$ on the sphere
  indicates the direction $\tilde{\nu}$ (associated with angles
  $(\theta, \psi)$); $P'$ is the zenithal equal-area projection of
  $P$.}
\label{fig:1a}
\end{figure}
In Figure~\ref{fig:1b} we show
\begin{figure}
\centering
\includegraphics[width=55mm]{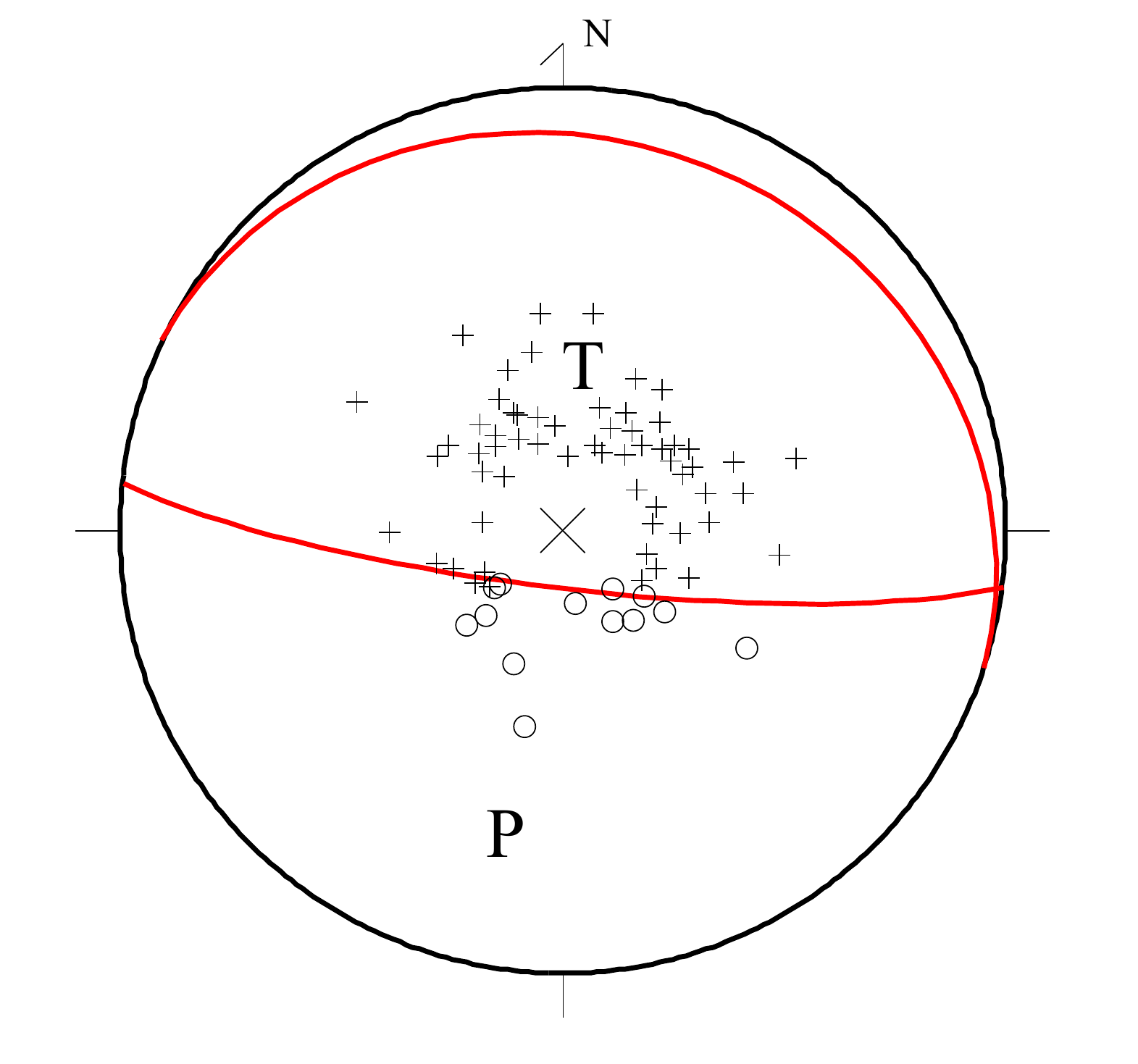}
\includegraphics[width=55mm]{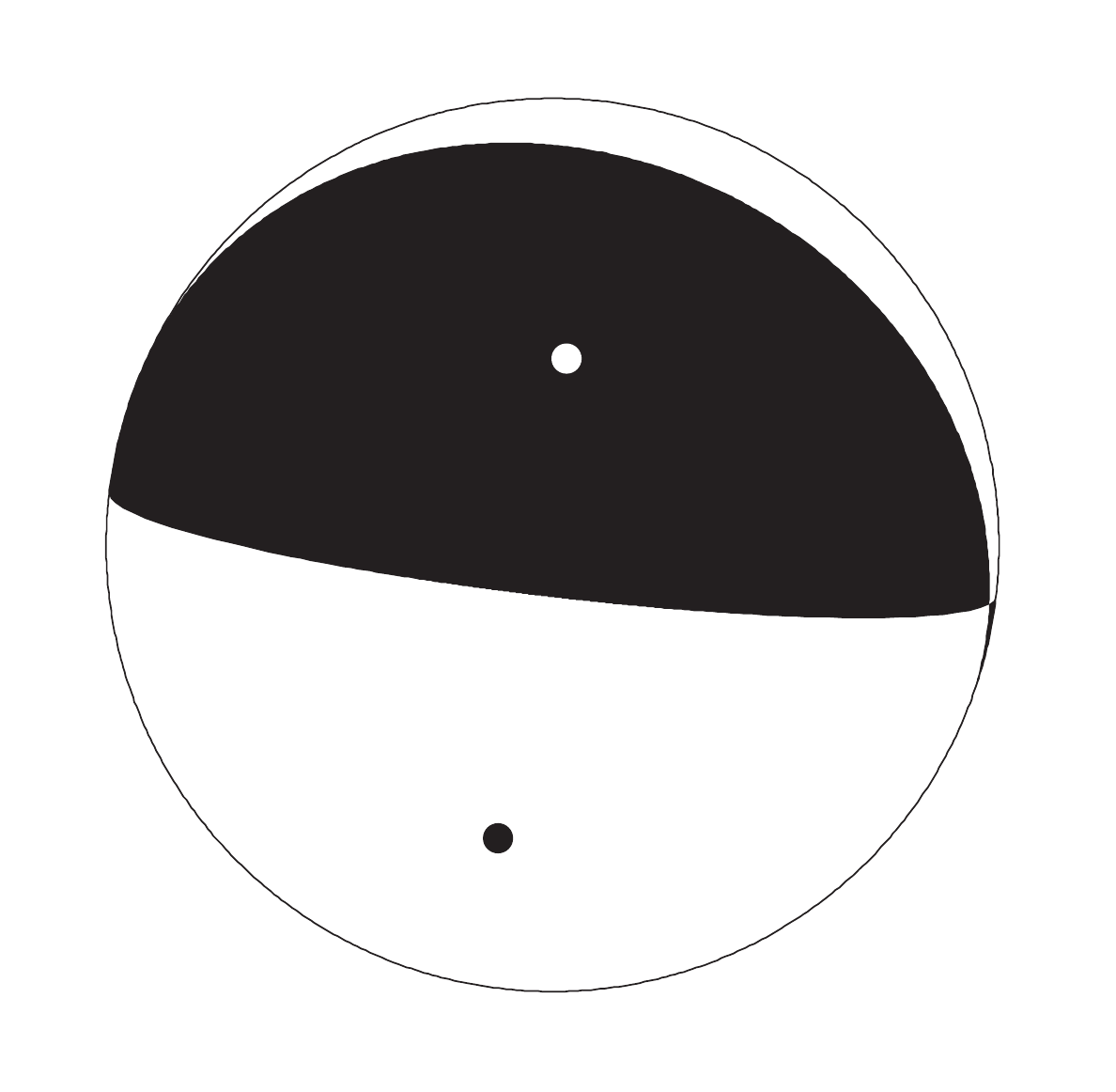}
\caption{Earthquake source, focal mechanism. Left: Sign ($\circ$
  indicates negative and $+$ indicates positive polarity) of
  $\tilde{\nu}_i \hat{M}_{ij} \tilde{\nu}_j$ obtained from (direct
  \textit{P}-wave) data; right: Beach ball plot (acknowledgment:
  Chunquan Yu). The directions of $T$ and $P$ are indicated by dots on
  the beach ball and are given in \eqref{eq:T} and \eqref{eq:P},
  respectively.}
\label{fig:1b}
\end{figure}
the reconstruction of a moment tensor as it
appears in an equivalent point body force displayed using a beach ball plot.

\section{Preliminary formulation of solutions} \label{Prelimsec}

In the next three sections we turn to the rigorous formulation of the mathematical model describing Earth's oscillations.

\subsection{Strong formulation of the system of elastic-gravitational
               equations}
\label{StrongForm}

We will say that $u$ is a strong solution of the elastic-gravitational
equation if $u$ satisfies in the classical sense the following initial
and boundary value problem, for $ g \in H^1(\tilde{X},
\mathbb{C}^3),\ h \in L^2(\tilde{X}, \mathbb{C}^3)$:
 $$\begin{cases} \rho^0 \Big[ \ddot{u} + 2 R_{\Omega} \dot{u}\Big] 
 + \rho^0 u \cdot \nabla \nabla ( \Phi^0+\Psi^s) + \rho^0 \nabla S(u) - \nabla \cdot (\Lambda^{T^0} : \nabla u) = f;\\
  u(0) = g ; \ \ \dot{u}(0) = h \end{cases}.$$
with the boundary conditions for $u$ given in table~\ref{unT}. Here we adopt the point of view that all of the other quantities in \eqref{eq:systemu} except for $u$ are already given. Hence the decoupled boundary conditions from table~\ref{unT} satisfied by $T^0$ should be considered as requirements on the given quantity $T_0$.   

We begin the analysis by introducing a notation for the ``spatial" part of the differential operator involved in the elastic-gravitational equation.
%
 Denote by $A$ the following operator
\begin{equation}\label{Aorig}
u\in \mathcal{C}^{\infty}_0(\tilde{X}) ;\ \ Au :=    u \cdot \nabla \nabla ( \Phi^0+\Psi^s) + \nabla S(u) - \dfrac{1}{\rho^0}\nabla \cdot (\Lambda^{T^0} : \nabla u) .
\end{equation}
The first two terms in the definition of $A$ are formally of order $0$ while the third is a differential operator of order $2$. In terms of $A$, \eqref{eq:systemu} is written as
\begin{equation}\label{eq:Aform}
\rho^0 \Big[ \ddot{u} + 2 R_{\Omega} \dot{u}\Big] + \rho^0 A u = \rho^0 f.
\end{equation}
Denote by $ L^2(\tilde{X} , \rho^0 dV)$ the following weighted $L^2$ Hilbert space: 
\begin{equation}
\begin{aligned}
 L^2(\tilde{X} , \rho^0 dV) := \left \{  u : \int_{\tilde{X}} \lvert u\rvert^2 \rho^0 \, dV  < \infty  \right \} ;\\
 \langle u, v\rangle_{L^2(\tilde{X} , \rho^0 dx) }:= \int_{\tilde{X}} u \ \overline{v}\, \rho^0 \, dV.
 \end{aligned}
\end{equation}
It is a fact that $(A , \mathcal{C}^{\infty}_0(\tilde{X}))$ is an
unbounded closable operator with dense domain in $ L^2(\tilde{X} ,
\rho^0 dV)$.

\medskip

\noindent

\subsection{Preliminary weak formulation of the system of
               elastic-gravitational equations}
\label{WeakForm}

We call $u\in H^2(\tilde{X})$ a weak solution of the system of
elastic-gravitational equations if for all $v \in H^1(\tilde{X}) $ we
have
%
 \begin{equation}\label{Weak1}
 \begin{aligned}
  \dfrac{d}{dt} \int_{\tilde{X}} \rho^0 \dot{u}\cdot \overline{v}\, dV
+ \int_{\tilde{X}} 2\rho^0 R_{\Omega} \dot{u} \cdot \overline{v} \, dV
+ \int_{\tilde{X}} \rho^0\big[ u \cdot \nabla \nabla ( \Phi^0+ \Psi^s) + \nabla S(u) - f \big] \cdot \overline{v} \, dV\\
+ \int_{\tilde{X}} ( \Lambda^{T^0} : \nabla u) : \nabla \overline{v}\, dV 
+ \int_{ \Sigma^{F} }[\nu \cdot ( \Lambda^{T^0} : \nabla u) \cdot \overline{v}]^+_- \, d\Sigma \ = 0.
\end{aligned}
\end{equation}
Recall that $\Sigma^F = \Sigma^{FF} \cup \Sigma^{FS}$ is the union of the fluid-fluid and fluid-solid interfaces. Note that for $u$ a strong solution \eqref{Weak1} can be obtained by multiplying the strong equation by $\overline{v}$ and then integrating by parts. The dynamic boundary conditions at the solid-solid and exterior boundaries then cause those boundary terms to disappear and we are left only with the boundary integral on the fluid region boundaries. This weak formulation leads to the introduction of a corresponding sesquilinear form:
\begin{equation}\label{aOrig}
\begin{aligned}
a_{\text{original}}\big( u, v \big)
= & \int_{\tilde{X}}T^{PK1} : \nabla \overline{v}\, dV 
+ \int_{\tilde{X}} \rho^0 u \cdot \nabla \nabla ( \Phi^0 + \Psi^s) \cdot \overline{v} \, dV +\int_{\tilde{X}} \rho^0 \nabla S(u) \cdot \overline{v} \, d V \\
&+  \int_{\Sigma^{F}} \big[ (\nu \cdot T^{PK1}) \cdot \overline{v}\big]^+_-  \, d\Sigma. 
\end{aligned}
\end{equation}
In terms of $a_{\text{original}}$ then \eqref{Weak1} can be written as
\begin{equation}\label{Weak1.2}
\dfrac{d}{dt} \int_{\tilde{X}} \rho^0 \dot{u}\cdot \overline{v}\, dV
+ \int_{\tilde{X}} 2 \rho^0 R_{\Omega} \dot{u} \cdot \overline{v} \, dV + a_{\text{original}}\big(u, v\big) = \int_{\tilde{X}} \rho^0 f \overline{v}\, dV.
\end{equation}

\subsection{Analysis of the fluid boundary integral} \label{FS:sec}

We can write the fluid boundary integral contribution to the
sesquilinear form $a_{\text{original}}$, when $u$ and $v$ satisfy the
appropriate boundary conditions, in different forms. We begin by
making use of some geometric constructions on the fluid interfaces
$\Sigma^{F}$ including the covariant surface derivative and the Weingarten
operator. For a background on these objects see Appendix
\ref{Wein}. We briefly recall that for a vector $V$ tangential to a
hypersurface $\Sigma$, with $\nu$ denoting a unit normal vector field along
$\Sigma$, the Weingarten operator applied to $V$ is
 $$W (V):= \nabla_V \nu $$
 where $\nabla$ is the (Levi-Civita) covariant derivative.
 \medskip
 
\begin{lemma} \label{aOriga1lem}
For $u$ and $v$ satisfying $[u\cdot\nu]^+_- = [v\cdot\nu]^+_- = 0$ on $\Sigma^{F}$ and given the boundary condition \eqref{BC1}, which is listed here
$$[\nu \cdot  T^{PK1}]^+_- 
 = - \nu \nabla^{\Sigma} \cdot ( p^0 [ u]^+_-) -p^0 W [u]^+_-, $$
 we have
\begin{equation}\label{EquiFormBdy}
\begin{aligned}
\int_{\Sigma^{F}} \big[ (\nu \cdot  T^{PK1}) \cdot \overline{v} \big]^+_-\, d\Sigma = &  
\int_{\Sigma^{F}} p^0 [u]^+_- \cdot \nabla^{\Sigma} (\overline{v} \cdot \nu) 
+p^0 [\overline{v}]^+_- \cdot \nabla^{\Sigma} ( u \cdot \nu) \, d\Sigma\\
-& \int_{\Sigma^{F}} p^0 \big[ W\big( u - (u \cdot \nu) \nu\big) \cdot \big( \overline{v} - (\overline{v}\cdot \nu) \nu\big)  \big]^+_- \, d\Sigma .
\end{aligned}
\end{equation} 
\end{lemma}
\medskip

\begin{proof}
We will use the notation $u^t = u - (u \cdot \nu) \nu$ for the
orthogonal projection of $u$ onto the tangent space of
$\Sigma^{F}$. Also, note that the hypotheses $[u \cdot \nu]_-^+=0$ and
$[v \cdot \nu]_-^+=0$ imply that $[u]_-^+$ and $[v]_-^+$ are tangent to
$\Sigma^{F}$. Moreover, for $\nu \cdot u$ and $\nu \cdot v$ on
$\Sigma^{F}$ it does not depend from which side we approach the interface.

To begin the proof note that
\begin{equation}\label{rewrite1}
\begin{aligned}
&[ (\nu\cdot T^{PK1} )\cdot\overline{v}]^+_- = [ \nu\cdot T^{PK1}]^+_- \cdot \overline{v}^+ + (\nu\cdot T^{PK1} )|_- \cdot [\overline{v}]^+_- .
\end{aligned}
\end{equation}
Now use the boundary condition \eqref{BC1} to rewrite the first term in the above expression as
\begin{align*}
[ \nu\cdot T^{PK1}]^+_- \cdot \overline{v}^+ &= 
 - \nu\cdot  \overline{v}^+ \nabla^{\Sigma} \cdot ( p^0 [ u]^+_-) -p^0 W ([u]^+_-) \cdot \overline{v}^+\\
 & = - \nu\cdot  \overline{v} \nabla^{\Sigma} \cdot ( p^0 [ u]^+_-) -p^0 W ([u]^+_-) \cdot (\overline{v}^t)^+ ,
\end{align*}
where for the second equality we used the fact that $W$ maps vectors tangent to $\Sigma^{F}$ to other tangent vectors. Applying integration by parts inside of $\Sigma^{F}$ and using our assumptions that $\partial\Sigma^{F}_{\text{int}} = \emptyset$ and $p^0 = 0 $ on $\partial \Sigma^{FS}_O$ (which follows from the boundary conditions \eqref{TractCont})  gives
$$\int_{\Sigma^{F}} 
- \nu\cdot  \overline{v} \nabla^{\Sigma} \cdot ( p^0 [ u]^+_-) \, d\Sigma= 
 \int_{\Sigma^{F}} p^0 \nabla^{\Sigma} ( \nu \cdot  \overline{v}) \cdot   [ u]^+_- \, d\Sigma.$$
Here we have also used the fact that $\Sigma^{F}$ is a $C^1$ surface. Next, using \eqref{TPK1} we have from the fluid side
\[
\begin{split}
(T^{PK1}_{ij})_- & = p^0 (\gamma - 1) (\nabla \cdot u_-)\delta_{ij} + p^0 (\nabla u_-)_{ij}.
\end{split}
\]
Hence the second term in \eqref{rewrite1} can be rewritten as
$$( \nu \cdot T^{PK1})|_-\cdot [\overline{v}]^+_-
= p^0 (\gamma-1) (\nabla \cdot u_- ) (\nu \cdot  [\overline{v}]^+_-)
+ p^0(  \nu \cdot \nabla u_- )\cdot  [\overline{v}]^+_-. $$
By the hypotheses the normal component of $[\overline{v}]_-^+$ at $\Sigma^{F}$ vanishes so the first term in the previous equation is zero and we obtain
$$ ( \nu \cdot T^{PK1})|_-\cdot [\overline{v}]^+_- = p^0 ( \nu \cdot \nabla u_- ) \cdot [\overline{v}]^+_-.$$
Next use the product rule and the continuity conditions across the boundary to rewrite the right hand side 
\begin{align*}
 ( \nu \cdot T^{PK1})|_-\cdot [\overline{v}]^+_- &= p^0  \nabla  ( \nu \cdot u_-) \cdot [\overline{v}]^+_- - p^0 ( u_- \cdot \nabla\nu) \cdot [\overline{v}]^+_-\\
 &= p^0  \nabla  ( \nu \cdot u) \cdot [\overline{v}]^+_- - p^0 u_- \cdot W( [\overline{v}]^+_-) \\
  &=p^0  \nabla^{\Sigma}  ( \nu \cdot u) \cdot [\overline{v}]^+_- - p^0 W(u^t_-) \cdot [\overline{v}]^+_- .
  \end{align*}
For the last step recall that $W$ maps vectors tangent to $\Sigma^{F}$ to other tangent vectors, is symmetric and, once again, that $[\overline{v}]^+_-$ is tangent to $\Sigma^{F}$ by the hypotheses.
%
Now putting the previous calculations together
\begin{align*}
& \int_{\Sigma^{F}} [ (\nu \cdot T^{PK1}) \cdot \overline{v}]^+_-\, d\Sigma\\
&= \int_{\Sigma^{F}} \Big(p^0 [u]^+_- \cdot \nabla^{\Sigma} ( \nu\cdot \overline{v}) + p^0  \nabla^{\Sigma}  ( \nu \cdot u) \cdot [\overline{v}]^+_-
 - p^0 W( [u]^+_-) \cdot (\overline{v}^t)^+ 
  - p^0 W(u^t_-) \cdot [\overline{v}]^+_-\Big)\, d\Sigma.
\end{align*} 
This is equivalent to \eqref{EquiFormBdy} and so the proof is complete.
\end{proof}
\medskip

In view of Lemma~\ref{aOriga1lem} it appears that an appropriate weak
formulation of the system of elastic-gravitational equations which
incorporates all of the boundary conditions would be given by
\eqref{Weak1.2} with $a_{\text{original}}$ replaced by $a_1$, see \eqref{a1}
below, since $a_1$ is simply $a_{\text{original}}$ with the fluid-solid
boundary integral term replaced using \eqref{EquiFormBdy}. However, we
are not able to prove that $a_1$ is coercive on an appropriate domain
and this necessitates the introduction of $a_2$ (see \eqref{a2}) and
the corresponding calculations in section~\ref{weak:sec}.

\section{Equivalent weak formulations}\label{weak:sec}

In this section we show how the sesquilinear form $a_{\text{original}}$ may be modified when $u, v$ have enough regularity to take derivatives, traces, and trace of the normal derivative in the classical sense. We must also assume that $T_0$, $u$, and $v$ satisfy certain boundary conditions. The actual domain of the sesquilinear forms will be discussed in subsection \ref{DoSes}.

We will consider two additional sesquilinear forms, defined initially only for $u$ and $v$ sufficiently regular, as follows. First we have a form obtained simply by applying \eqref{EquiFormBdy} to \eqref{Aorig}
\begin{equation}\label{a1}
 \begin{aligned}
a_1\big(u, v \big)
= & \int_{\tilde{X}} ( \Lambda^{T^0} : \nabla u ) : \nabla \overline{v}\, d V
+ \int_{\tilde{X}} \rho^0 u \cdot \nabla \nabla ( \Phi^0 + \Psi^s) \cdot \overline{v}\ dV +\int_{\tilde{X}} \rho^0 \nabla \Phi^1 \cdot \overline{v} \, d V \\
&+ \int_{\Sigma^{F}} p^0 [u]^+_- \cdot \nabla^{\Sigma} (\overline{v} \cdot \nu) 
+p^0 [\overline{v}]^+_- \cdot \nabla^{\Sigma} ( u \cdot \nu)\, d \Sigma \\
-& \int_{\Sigma^{F}} p^0 \big[ W\big( u - (u \cdot \nu) \nu\big) \cdot \big( \overline{v} - (\overline{v}\cdot \nu) \nu\big)  \big]^+_- \, d\Sigma.
\end{aligned}
\end{equation}
Next we introduce a more complicated form, which is the final one used in the next section for the results on well-posedness. First we need some extra notation. Let $\sigma_N$ be any regular scalar function which is $-p^0$ in $\Omega^F$ and $0$ outside of a small neighborhood of $\Omega^F$. We choose this neighborhood sufficiently small so that none of the solid-solid interfaces intersect the support of $\sigma_N$. Also, whenever we have an expression $B(u,\overline{v})$ we will write $\mathfrak{S}B$ for the symmetrization
\[
\mathfrak{S}\{ B(u,\overline{v}) \} = \tfrac{1}{2}\left ( B(u,\overline{v}) + B (\overline{v},u)\right ).
\]
The other sesquilinear form is defined by
\begin{equation}\label{a2}
\begin{aligned}
 a_2\big(u, v \big)
= &  \int_{\Omega^S}\Big(  ( \Lambda^{T^0} : \nabla u ) : \nabla \overline{v} +  \sigma_N  \nabla u : \nabla \overline{v}^T -  \sigma_N (\nabla\cdot u)( \nabla \cdot \overline{v})\Big) \, d V\\
& +\int_{\Omega^S}\Big( -\mathfrak{S} \big\{  (g_0' \cdot u) ( \overline{v}\cdot \nabla \rho^0)  \big\} 
+ \mathfrak{S} \big\{  - ( \nabla \sigma_N  +\rho^0 g_0' ) \cdot u (\nabla \cdot \overline{v}  )\big\}
\\
&+ \mathfrak{S} \big\{     (  \nabla \sigma_N -  \rho^0 g_0' )\cdot \nabla u \cdot \overline{v} \big\}\Big) \, dV - \int_{\Sigma^{SS} } \mathfrak{S} \big\{ [\rho^0]^+_- (u \cdot g_0' )( \overline{v}\cdot\nu)\big\} \, d\Sigma
\\
&+ \int_{\Omega^F} \Big(\dfrac{p^0 \gamma}{(\rho^0)^2}  \Big(\nabla\cdot (\rho^0  u ) 
 - \tilde{s} \cdot u \Big) \Big(\nabla\cdot (\rho^0  v ) 
 - \tilde{s} \cdot \overline{v} \Big)- \tilde{s}\cdot g_0'  \dfrac{(g_0' \cdot u) (\overline{v}\cdot g_0' ) }{\lVert g_0'\rVert^2}\Big) \, dV \\
 &- \int_{ \Sigma^{FF} } [\rho^0]^+_- (g_0' \cdot \nu)(u \cdot \nu)( \overline{v}\cdot\nu) \, d\Sigma\\
&-\int_{\Sigma^{FS}}  \mathfrak{S}\Big\{ (\overline{v}\cdot \nu)\Big( u_+ \cdot [\rho^0] g_0' \Big)\Big\} d\Sigma\\
&-\dfrac{1}{4\pi G} \int_{\RR^3} \nabla S(u) \cdot \nabla S(\overline{v}) \, d V + \int_{\partial \tilde{X}} \mathfrak{S}\{\rho^0 (u \cdot g_0' ) \overline{v}\cdot\nu\}\, d \Sigma
\end{aligned}
\end{equation}
where
$$\tilde{s} = \nabla \rho^0 +\dfrac{g_0' (\rho^0)^2}{p^0 \gamma} \ \text{and} \ g_0' = \nabla ( \Phi^0 + \Psi^s).$$
\
\noindent \begin{remark}
$\tilde{s}$ as defined is related to the Brunt-V\"ais\"al\"a frequency $N^2$ ( see for e.g. \cite{Valette,Valette_86,Dahlen})
$$N^2 = \dfrac{\tilde{s}}{\rho^0} \cdot (-g_0'). $$
 The Brunt-V\"ais\"al\"a frequency $N^2$ arises naturally in the analysis of the local stability of the fluid, providing a simple way to formulate the Schwarzschild criterion (see for e.g. \cite{ChaljubValette}). 
\end{remark}

\medskip
The remainder of this section will be dedicated to proving the equivalence of these forms and $a_{\text{original}}$ given sufficient regularity and boundary conditions. The result is as follows.
\medskip

\begin{lemma}\label{A3}\
If $T^0$ satisfies \eqref{TractCont} and $u,v \in C^\infty$ satisfy along $\Sigma^{F}$
$$[u\cdot \nu]^+_- = [v\cdot \nu]^+_- = 0, $$
as well as
\[
[u]_-^+ = [v]_-^+ = 0
\]
along $\Sigma^{SS}$ then we have
$$ a_1\big(u, v \big) =  a_2\big(u, v \big).$$
In addition, if $u$ and $T^0$ satisfy the boundary conditions in Table \ref{unT} then
$$a_{\text{original}} \big(u, v \big)
= a_1\big(u, v \big) =  a_2\big(u, v \big).$$
\end{lemma}
\medskip

\begin{proof}\

\textbf{Step 1a}: 
Let us assume that $T^0$ satisfies \eqref{TractCont} and $u,v$ satisfy
\begin{equation}\label{condSigmaFS}
[u\cdot \nu]^+_- = [v\cdot \nu]^+_- = 0
\end{equation}
along $\Sigma^{F}$ and 
\begin{equation}\label{condSigmaFFSS}
[u]_-^+ = [v]_-^+ = 0
\end{equation}
 on $\Sigma^{SS}$. Note that \eqref{condSigmaFFSS} implies \eqref{condSigmaFS}. Then, following the idea of \cite[equation (A.41), p.41]{Valette}, we have the following computation
\begin{equation}\label{stepcomp}
\begin{split}
&  \int_{\tilde{X}} \sigma_N \Big( \nabla u : \nabla \overline{v}^T -  (\nabla\cdot u)( \nabla \cdot \overline{v})\Big)\, dV
+ \int_{\tilde{X}}\mathfrak{S}\big\{\big( \nabla \overline{v} \cdot u - u \nabla \cdot \overline{v}\big) \cdot \nabla \sigma_N \big\} \, dV\\
\stackrel{(1)}{=}&\int_{\tilde{X}} \mathfrak{S}\Big\{ \nabla \cdot \big( \sigma_N \nabla \overline{v} \cdot u - \sigma_N u \nabla \cdot \overline{v}\big)  \Big\} \, dV \\
\stackrel{(2)}{=}& \int_{\Sigma^F} p^0 \mathfrak{S} \Big\{\Big[ \nu \cdot \nabla \overline{v} \cdot u - (u\cdot \nu) \nabla \cdot \overline{v}\Big]^+_-\Big\} \, d\Sigma\\
\stackrel{(3)}{=}& \int_{\Sigma^F} \Big(p^0 [u]^+_- \cdot \nabla^{\Sigma} ( \nu\cdot \overline{v}) + p^0  \nabla^{\Sigma}  ( \nu \cdot u) \cdot [\overline{v}]^+_-\Big)\, d\Sigma\,  - \int_{\Sigma^F} p^0  \big[ W( u - (u \cdot \nu) \nu) \cdot \big(  \overline{v} - (\overline{v} \cdot \nu) \nu)\big) \big]^+_-\, d\Sigma\\
& +\int_{\Sigma^F} \mathfrak{S} \big\{ (u\cdot \nu)[\overline{v}]^+_- \cdot \nabla^{\Sigma} p^0 \big\}\, d\Sigma
\end{split}
\end{equation}
Equality (1) follows from the product rule for $\nabla \cdot$.  Equality (2) is an application of the divergence theorem. Here, the boundary term along $\Sigma^{SS}$ vanishes, due to the assumption on the support of $\sigma_N$, while that along $(\partial \widetilde{X})_F$ vanishes because $p^0 = 0 $ there. 

 For equality (3), we will first rewrite the terms $[(u\cdot \nu) \nabla \cdot \overline{v}]^+_-$ and $[ \nu \cdot (\nabla \overline{v} \cdot u )]^+_-$ separately.  Indeed, using the continuity conditions \eqref{condSigmaFS} we have
\begin{equation}\label{rw1}
\begin{aligned}
 \big[(u\cdot \nu) \nabla \cdot \overline{v}\big]^+_-  & = [ u\cdot \nu]_-^+ \nabla \cdot \overline{v}_+ + (u_- \cdot \nu) [\nabla \cdot \overline{v}]_-^+ \\
&= (u \cdot \nu) ( \nabla^{\Sigma} \cdot [\overline{v}]_-^+ + [\nu \cdot \nabla \overline{v} \cdot \nu]_-^+).
\end{aligned}
\end{equation}
on $\Sigma^F$. Note in the second equality we used equation \eqref{tanDi}.

For $[ \nu \cdot (\nabla \overline{v} \cdot u )]_-^+$, the other term at step (2) in \eqref{stepcomp}, we have using the product rule, again the continuity conditions \eqref{condSigmaFS}, and properties of the Weingarten operator $W$ on $\Sigma^F$
\begin{equation}\label{rw2}
\begin{aligned}
\big [ \nu \cdot (\nabla \overline{v} \cdot u )\big ]_-^+
& = \big[ \nu \cdot (\nabla \overline{v} \cdot (u^t + \nu (u \cdot \nu) )\big]^+_- \\
& = \big[ \nu \cdot \nabla \overline{v} \cdot u^t \big]^+_- + [\nu \cdot \nabla \overline{v} \cdot \nu]_-^+ (u\cdot \nu) \\
&= [ \nabla^\Sigma ( \nu \cdot \overline{v}) \cdot u^t - \overline{v}\cdot\nabla \nu \cdot  u^t]_-^+ + [\nu \cdot \nabla \overline{v} \cdot \nu]_-^+ (u\cdot \nu)\\
&= \nabla^\Sigma( \nu \cdot \overline{v}_-) \cdot [u^t]_-^+  - [\overline{v}]_-^+ \cdot \nabla\nu \cdot u^t_+ 
- \overline{v}_- \cdot \nabla \nu \cdot [u^t]_-^+ + [\nu \cdot \nabla \overline{v} \cdot \nu]_-^+ (u\cdot \nu)\\
&= \nabla^{\Sigma} ( \nu \cdot \overline{v}) \cdot [u]_-^+ - [ W(u^t) \cdot \overline{v}^t]_-^+ + [\nu \cdot \nabla \overline{v} \cdot \nu]_-^+ (u\cdot \nu),
\end{aligned}
\end{equation}
where $u^t$ and $v^t$ are the orthogonal projections of $u$ and $v$ onto the tangent space of $\Sigma^F$. Returning to the second equality of \eqref{stepcomp}, subtracting \eqref{rw1} from \eqref{rw2} now gives
\[
p^0 \mathfrak{S} \Big\{\Big[ \nu \cdot \nabla \overline{v} \cdot u - (u\cdot \nu) \nabla \cdot \overline{v}\Big]^+_-\Big\} = p^0 \mathfrak{S}\left \{  \nabla^{\Sigma} ( \nu \cdot \overline{v}) \cdot [u]_-^+ - [ W(u^t) \cdot \overline{v}^t]_-^+ -(u \cdot \nu) \nabla^{\Sigma} \cdot [\overline{v}]_-^+ \right \}.
\]
To complete the proof of step (3) note that integration by parts in $\Sigma^F$ and the product rule gives
\[
- \int_{\Sigma^F} p^0 \mathfrak{S}\left \{  (u \cdot \nu) \nabla^{\Sigma} \cdot [\overline{v}]_-^+ \right \} \, d\Sigma= \int_{\Sigma^F} \Big(p^0 \mathfrak{S}\left \{ \nabla^\Sigma ( u \cdot \nu)\cdot [\overline{v}]_-^+ \right \} + \mathfrak{S}\left \{ ( \nu \cdot u)( [\overline{v}]_-^+ \cdot \nabla^\Sigma p^0) \right \}\Big)\, d\Sigma .
\]
Here, the boundary term along $\partial \Sigma^F = \partial \Sigma^{FF} \cup \partial \Sigma^{FS}_{\text{int}} \cup \partial \Sigma^{FS}_O$vanishes. This is because for the interior intefaces $\partial \Sigma^{FF}= \partial \Sigma^{FS}_{\text{int}} = \emptyset$ , while for ocean edges, although $\partial \Sigma^F_O \neq \emptyset$, we assume that $p^0 = 0$ there. 
This completes the proof of \eqref{stepcomp}.

\textbf{Step 1b} : Rearranging \eqref{stepcomp} we now have
\[
\begin{split}
& \int_{\Sigma^{F}} p^0 [u]^+_- \cdot \nabla^{\Sigma} ( \nu\cdot \overline{v}) + p^0  \nabla^{\Sigma}  ( \nu \cdot u) \cdot [\overline{v}]^+_-\,d\Sigma  - \int_{\Sigma^{F}} p^0  \big[ W( u - (u \cdot \nu) \nu) \cdot \big(  \overline{v} - (\overline{v} \cdot \nu) \nu)\big) \big]^+_-d\Sigma\\
& = \int_{\tilde{X}} \sigma_N \Big( \nabla u : \nabla \overline{v}^T -  (\nabla\cdot u)( \nabla \cdot \overline{v})\Big)\, dV
+ \int_{\tilde{X}} \mathfrak{S}\big\{\big( \nabla \overline{v} \cdot u - u \nabla \cdot \overline{v}\big) \cdot \nabla \sigma_N \big\} \, dV\\
&\hskip2in -  \int_{\Sigma^{FS}} \mathfrak{S} \big\{ (u\cdot \nu)[\overline{v}]^+_- \cdot \nabla^{\Sigma} p^0 \big\} \, d\Sigma.
\end{split}
\]
The final integral in this formula is only over $\Sigma^{FS}$ since along $\Sigma^{FF}$ Lemma \ref{parallellemma} implies that $\nabla p^0$ is parallel to $\nu$ from both sides of $\Sigma^{FF}$ . Hence the term $\nabla^{\Sigma} p^0$ is zero on $\Sigma^{FF}$.

This shows how we can replace the last two lines in the definition \eqref{a1} of $a_1$ so that we have under the given hypotheses
\begin{equation}\label{a1firststep}
 \begin{aligned}
a_1\big(u, v \big)
= & \int_{\tilde{X}} ( \Lambda^{T^0} : \nabla u ) : \nabla \overline{v}\, dV
+ \int_{\tilde{X}} \rho^0 u \cdot \nabla \nabla ( \Phi^0 + \Psi^s) \cdot \overline{v} \, dV+\int_{\tilde{X}} \rho^0 \nabla \Phi^1 \cdot \overline{v} \, dV\\
&+\int_{\tilde{X}} \sigma_N \Big( \nabla u : \nabla \overline{v}^T -  (\nabla\cdot u)( \nabla \cdot \overline{v})\Big)
\, dV+ \int_{\tilde{X}}\mathfrak{S}\big\{\big( \nabla \overline{v} \cdot u - u \nabla \cdot \overline{v}\big) \cdot \nabla \sigma_N \big\} \, dV\\
&-  \int_{\Sigma^{FS}} \mathfrak{S} \big\{ (u\cdot \nu)[\overline{v}]^+_- \cdot \nabla^{\Sigma} p^0 \big\} d\Sigma.
\end{aligned}
\end{equation}

 \textbf{Step 2} : Now let us continue to consider the terms
\begin{equation}\label{twoterms}
\int_{\tilde{X}} \rho^0 u \cdot \nabla \nabla ( \Phi^0 + \Psi^s) \cdot \overline{v} \,dV\ \text{and}\ 
\int_{\tilde{X}} \rho^0 \nabla \Phi^1 \cdot \overline{v}\, dV.
\end{equation}
Recall that
\[
\Phi^1 = S(u)
\]
where $S(u)$ is the solution of the Poisson's equation 
$$ \Delta S(u) =   - 4 \pi G \nabla \cdot ( \rho^0 u) $$
in $\RR^3$ with $\rho^0 = 0 $ in $\RR^3 \setminus \overline{\tilde{X}}$.  Using Lemma~\ref{infiComp}, which says that $S(u) \rightarrow 0$ and $\nabla S(u) \rightarrow 0$ as $|x| \rightarrow \infty$, we can rewrite the second integral using integration by parts twice as follows
\begin{equation}\label{S(u)calc}
\begin{split}
\int_{\tilde{X}} \rho^0 \nabla \Phi^1 \cdot \overline{v} \, dV& = \frac{1}{4 \pi G}\int_{\RR^3} (-4 \pi G \nabla \cdot (\rho^0 \overline{v})) \  \Phi^1\, dV\\
& = \frac{1}{4 \pi G} \int_{\RR^3} \Delta S(\overline{v}) S(u) \, dV\\
& = -\frac{1}{4 \pi G} \int_{\RR^3} \nabla S(u) \cdot \nabla S(\overline{v})\, dV.
\end{split}
\end{equation}
This calculation can alternately be interpreted as simply coming from the definition of the weak derivative, and the fact that $S(u)$ is a weak solution of the Poisson equation given above.

For the first term in \eqref{twoterms} we first recall that $g_0' = \nabla ( \Phi^0 + \Psi^s)$. The term may then be rewritten by first noting that the product rule implies
\begin{align*}
\nabla \cdot \big( \rho^0 (u \cdot g_0' ) \overline{v} \big) 
&= \rho^0 (u \cdot g_0' ) \nabla \cdot \overline{v}
+ \nabla \big( \rho^0 u \cdot g_0'\big) \cdot \overline{v}\\
&= \rho^0 (u \cdot g_0' ) \nabla \cdot \overline{v}
+ (u \cdot g_0')(\nabla \rho^0)  \cdot  \overline{v}
+ \rho^0     (u \cdot \nabla g_0' )\cdot \overline{v}
+ \rho^0 (g_0' \cdot \nabla u ) \cdot \overline{v}.
\end{align*}
Thus the divergence theorem gives
\begin{equation}\label{2ndterm}
\begin{split}
\int_{\tilde{X}} \rho^0     u \cdot \nabla g_0' \cdot \overline{v}\, dV
= & \int_{\tilde{X}} \nabla \cdot \big( \rho^0 (u \cdot g_0' ) \overline{v} \big) \, dV\\
&\hskip-.75in - \int_{\tilde{X}} \Big(\rho^0 (u \cdot g_0' ) \nabla \cdot \overline{v}
+ (u \cdot g_0')\nabla \rho^0  \cdot  \overline{v}
+ \rho^0 g_0' \cdot \nabla u  \cdot \overline{v}\Big)\, dV\\
= & \int_{\partial \tilde{X}} \rho^0 (u \cdot g_0' ) \overline{v}\cdot\nu \, d\Sigma
- \int_{\Sigma^{SS} \cup \Sigma^{FF} } \Big[\rho^0 (u \cdot g_0' ) \overline{v}\cdot\nu\Big]^+_-\, d\Sigma
\\
&\hskip-.75in - \int_{\Sigma^{FS} }\Big[ \rho^0 (u \cdot g_0' ) \overline{v}\cdot\nu\Big]^+_-\, d\Sigma - \int_{\tilde{X}} \Big(\rho^0 (u \cdot g_0' ) \nabla \cdot \overline{v}
+ (u \cdot g_0')\nabla \rho^0  \cdot  \overline{v}
+ \rho^0 g_0' \cdot \nabla u  \cdot \overline{v}\Big)\, dV.
\end{split}
\end{equation}
Symmetrizing \eqref{2ndterm} and substituting it into \eqref{a1firststep} we now have
\begin{equation}\label{thirdR}
\begin{aligned}
a_1\big(u, v \big)
= &  \int_{\tilde{X}} ( \Lambda^{T^0} : \nabla u ) : \nabla \overline{v} \, dV - \dfrac{1}{4\pi G} \int_{\RR^3} \nabla S(u) \cdot \nabla S(\overline{v}) \, d V  +  \int_{\partial \tilde{X}} \mathfrak{S}\big\{ \rho^0 (u \cdot g_0' ) \overline{v}\cdot\nu\big\}\, d\Sigma\\
& - \int_{\Sigma^{SS} \cup \Sigma^{FF} } \mathfrak{S}\left \{ \Big[\rho^0 (u \cdot g_0' ) \overline{v}\cdot\nu\Big]^+_-\right \}\, d\Sigma
- \int_{\Sigma^{FS} } \mathfrak{S} \left \{\Big[ \rho^0 (u \cdot g_0' ) \overline{v}\cdot\nu\Big]^+_- \right \}\, d\Sigma\\
&- \int_{\tilde{X}} \mathfrak{S}\Big\{ \rho^0 (u \cdot g_0' ) \nabla \cdot \overline{v}
+ (u \cdot g_0')\nabla \rho^0  \cdot  \overline{v}
+ \rho^0 g_0' \cdot \nabla u  \cdot \overline{v} \Big\}  \, dV\\
 &+ \int_{\tilde{X}}   \sigma_N \Big( \nabla u : \nabla \overline{v}^T -  (\nabla\cdot u)( \nabla \cdot \overline{v})\Big)\, d V
+ \int_{\tilde{X}} \mathfrak{S}\big\{\big( \nabla \overline{v} \cdot u - u \nabla \cdot \overline{v}\big) \cdot \nabla \sigma_N \big\}\, d V \\
&-  \int_{\Sigma^{FS}} \mathfrak{S} \big\{ (u\cdot \nu)[\overline{v}]^+_- \cdot \nabla^{\Sigma} p^0 \big\} \, d\Sigma.
\end{aligned}
\end{equation}
Now we begin to rearrange the right hand side of \eqref{thirdR} to match the form of \eqref{a2}.

First let us consider the integrals over the fluid-fluid and solid-solid boundaries. Using the boundary condition $[u]^+_- = 0 $ along $\Sigma^{SS}$ we obtain
\begin{equation}\label{SSbdy}
\int_{\Sigma^{SS} } \mathfrak{S}\Big\{ \big[\rho^0 (u \cdot g_0' ) \overline{v}\cdot\nu\big]^+_-\Big\} d\Sigma= \int_{\Sigma^{SS}} \mathfrak{S}\big\{ [\rho^0]^+_- (u \cdot g_0' )( \overline{v}\cdot\nu)\big\}d\Sigma. 
\end{equation} 
Since we assume that the fluid is stratified $g_0' || \nu $ along $\Sigma^{FF}$ and so $g_0' = (g_0'\cdot \nu) \nu$. Using also the condition that $[u \cdot \nu]_-^+ = 0$ on $\Sigma^{FF}$, the boundary integral term along $\Sigma^{FF}$ can thus be written as
 $$\int_{\Sigma^{FF}} [\rho^0]^+_- (g_0' \cdot \nu) (u\cdot \nu)( \overline{v}\cdot\nu) d\Sigma.$$
These are the same integral terms over the fluid-fluid and solid-solid boundaries that appear in \eqref{a2}.

Next we look at the fluid-solid boundary integral terms in \eqref{thirdR}. These are
\begin{align*}
&-\int_{\Sigma^{FS}} \mathfrak{S}\big\{ (u\cdot \nu)[\overline{v}]^+_- \cdot \nabla^{\Sigma} p^0 \big\}\, d\Sigma
-\int_{\Sigma^{FS} } \mathfrak{S} \Big\{\big[ \rho^0 (u \cdot g_0' ) \overline{v}\cdot\nu\big]^+_- \Big\}d\Sigma\\
=&- \int_{\Sigma^{FS}} \mathfrak{S}\Big\{ (\overline{v}\cdot \nu)\big[u  \cdot \big( \nabla^{\Sigma} p^0 + 
 \rho^0 g_0' \big)\big]^+_- \Big\}d\Sigma.
\end{align*}
Now we simplify $ \nabla^{\Sigma} p^0 + \rho^0 g_0'$ using the relationship between $p^0$ and $g_0' $. The mechanical equilibrium condition \eqref{Equi1} gives
$$ \rho^0 g_0' = \nabla \cdot T^0.$$
On the other hand, by \eqref{staticStress} the pretressed state $T^0$ in the fluid region has the form
$$T^0 = -p^0 \mathrm{Id}\ \Rightarrow\ \nabla \cdot T^0 = -\nabla p^0.$$
Hence in the fluid region $\nabla p^0 =- \rho^0 g_0'$ and so using the definition of the tangential gradient
$$ \nabla^{\Sigma} p^0 = \nabla p^0 - ( \nabla p^0 \cdot \nu) \nu 
= - \rho^0_- g_0'  +  (  \rho^0_- g_0' \cdot \nu) \nu.$$
Along $\Sigma^{FS}$ this gives
\begin{equation}\label{bd1}
 [ u \cdot \nabla^{\Sigma} p^0]_-^+ = [u]_-^+ \cdot \nabla^{\Sigma} p^0= [u]_-^+ \cdot \Big(- \rho^0_- g_0'  +  (  \rho^0_- g_0' \cdot \nu) \nu \Big)
= - \rho^0_- [u]_-^+ \cdot  g_0'  \,, \end{equation}
where we used the fact that $[u]_-^+$ is tangent to $\Sigma^{FS}$. For the second term
\begin{equation}\label{bd2}
\begin{aligned}
  \big[ u \cdot   \rho^0 g_0'\big]^+_- &= u_+ \cdot [\rho^0]_-^+ g_0' + [u]_-^+ \cdot \rho^0_- g_0' 
\end{aligned}
\end{equation} 
and putting \eqref{bd1} and \eqref{bd2} together we obtain
$$\big[u  \cdot \big( \nabla^{\Sigma} p^0 +  \rho^0 g_0' \big)\big]^+_- 
=u_+ \cdot [\rho^0]_-^+ g_0' .$$
Therefore the fluid-solid boundary integral term is
\[
-\int_{\Sigma^{FS}} \mathfrak{S}\left \{ (\overline{v}\cdot \nu) (u_+ \cdot [\rho^0]_-^+ g_0') \right \}
\]
which matches the term in \eqref{a2}.

Now we move onto the volume terms. First note that by simply rearranging
\begin{equation}
\begin{aligned}
&- \mathfrak{S}\Big\{ \rho^0 (u \cdot g_0' ) \nabla \cdot \overline{v}
+ (u \cdot g_0')\nabla \rho^0  \cdot  v
+ \rho^0 g_0' \cdot \nabla u  \cdot \overline{v} \Big\}
+ \mathfrak{S}\big\{\big( \nabla \overline{v} \cdot u - u \nabla \cdot \overline{v}\big) \cdot \nabla \sigma_N \big\}\\
=& -\mathfrak{S} \big\{  (g_0' \cdot u) ( v\cdot \nabla \rho^0)  \big\} 
+ \mathfrak{S} \big\{  - ( \nabla \sigma_N  +\rho^0 g_0' ) \cdot u (\nabla \cdot \overline{v}  )\big\}
+ \mathfrak{S}\big\{     (  \nabla \sigma_N -  \rho^0 g_0' )\cdot \nabla u \cdot \overline{v} \big\}
\end{aligned}
\end{equation}
and in the solid region $\Omega^S$ this already matches the form in \eqref{a2}. To complete the proof of the equivalence of $a_1$ and $a_2$ all that remains is to consider the portion of the volume integral over the fluid region $\Omega^F$. We do this in the following lemma.
\medskip

\begin{lemma}[Volume integral in $\Omega^F$]\label{loVt}
\begin{equation}\label{VolumeIntFluid}
\begin{aligned}
&\int_{\Omega^F}\Big( ( \Lambda^{T^0} : \nabla u ) : \nabla \overline{v} +  \sigma_N  \nabla u : \nabla \overline{v}^T -  \sigma_N (\nabla\cdot u)( \nabla \cdot \overline{v}) \Big) \, dV +\int_{\Omega^F}\Big(-\mathfrak{S} \big\{  (g_0' \cdot u) ( \overline{v}\cdot \nabla \rho^0)  \big\} \\
& + \mathfrak{S} \big\{  - ( \nabla \sigma_N  +\rho^0 g_0' ) \cdot u (\nabla \cdot \overline{v}  )\big\}
+ \mathfrak{S}\big\{     (  \nabla \sigma_N -  \rho^0 g_0' )\cdot \nabla u \cdot \overline{v} \big\}\Big) \, dV\\
= & \int_{\Omega^F}\Big(  \dfrac{p^0 \gamma}{(\rho^0)^2}  \Big(\nabla\cdot (\rho^0  u ) 
 - \tilde{s} \cdot u \Big) \Big(\nabla\cdot (\rho^0  \overline{v} ) 
 - \tilde{s} \cdot \overline{v} \Big)- \tilde{s}\cdot g_0'  \dfrac{(g_0' \cdot u) (\overline{v}\cdot g_0' ) }{\lVert g_0'\rVert^2}\Big) \, dV ,
\end{aligned}
\end{equation}
where
$$\tilde{s} = \nabla \rho^0 +\dfrac{g_0' (\rho^0)^2}{p^0 \gamma}.$$
\end{lemma}
\medskip

\begin{proof}
Since  $\sigma_N = - p^0$ in $\Omega^F$ we have as above $\nabla \sigma_N = \rho^0 g_0' $. Thus
\begin{align*}
\int_{\Omega^F}\Big(-\mathfrak{S} \big\{  (g_0' \cdot u) ( \overline{v}\cdot \nabla \rho^0)  \big\} 
& + \mathfrak{S} \big\{  - ( \nabla \sigma_N  +\rho^0 g_0' ) \cdot u (\nabla \cdot \overline{v}  )\big\}
+ \mathfrak{S} \big\{     (  \nabla \sigma_N -  \rho^0 g_0' )\cdot \nabla u \cdot \overline{v} \big\} \Big)\, dV\\
=& -\int_{\Omega^F}\Big(\mathfrak{S} \big\{  (g_0' \cdot u) ( \overline{v}\cdot \nabla \rho^0)  \big\} 
+ \mathfrak{S} \big\{  2\rho^0 g_0'  \cdot u (\nabla \cdot \overline{v}  )\big\}\Big)\, dV.
\end{align*}
Now use the product rule to rewrite
\begin{align*}
\mathfrak{S} \big\{  (g_0' \cdot u) ( \overline{v}\cdot \nabla \rho^0)  \big\} 
& + \mathfrak{S} \big\{   2\rho^0 g_0'  \cdot u (\nabla \cdot \overline{v}  )\big\}\\
=&\ \mathfrak{S} \big\{  (g_0' \cdot u) ( \overline{v}\cdot \nabla \rho^0)  \big\} 
+ \mathfrak{S} \big\{   2 g_0'  \cdot u (\nabla \cdot (\rho^0 \overline{v})   )
- 2 (g_0' \cdot u )( \nabla \rho^0 \cdot \overline{v}) \}\\
=& - \mathfrak{S} \big\{  (g_0' \cdot u) ( \overline{v}\cdot \nabla \rho^0)  \big\} 
 + \mathfrak{S} \big\{   2 g_0'  \cdot u (\nabla \cdot (\rho^0 \overline{v})   )\big\}.
\end{align*}
Using \eqref{EquiFluid} to replace $\nabla \rho^0$ we finally see that the second line in \eqref{VolumeIntFluid} is equivalent to
\[
\int_{\Omega^F}  \Big( \frac{(g_0' \cdot u) ( \overline{v}\cdot g_0')}{\|g_0'\|^2} (\nabla \rho^0 \cdot g_0')  - 2 \mathfrak{S}\{( g_0'\cdot u) (\nabla \cdot ( \rho^0 \overline{v}))\}\Big) \, dV.
\]
On the other hand for the first line in \eqref{VolumeIntFluid} we have by direct computation making use of \eqref{staticStress}, \eqref{ModStiffTensor}, and \eqref{PerFluid}
$$ ( \Lambda^{T^0} : \nabla u ) : \nabla \overline{v} +  \sigma_N  \nabla u : \nabla \overline{v}^T -  \sigma_N (\nabla\cdot u)( \nabla \cdot \overline{v}) 
 = p^0 \gamma ( \nabla\cdot u ) (\nabla \cdot \overline{v})$$
and by the product rule
\[
\begin{split}
p^0 \gamma ( \nabla\cdot u ) (\nabla \cdot \overline{v}) & =  \dfrac{p^0 \gamma}{(\rho^0)^2} \Big(\nabla\cdot (\rho^0  u ) 
 - \nabla \rho^0 \cdot u \Big)\Big(\nabla\cdot (\rho^0  \overline{v} ) 
 - \nabla \rho^0 \cdot \overline{v} \Big) \\
 & =   \dfrac{p^0 \gamma}{(\rho^0)^2} \Big(\nabla\cdot (\rho^0  u ) 
 - \tilde{s} \cdot u + \dfrac{g_0' (\rho^0)^2}{p^0 \gamma} \cdot u \Big)\Big(\nabla\cdot (\rho^0  \overline{v} ) 
 - \tilde{s} \cdot \overline{v}+ \dfrac{g_0' (\rho^0)^2}{p^0 \gamma} \cdot \overline{v} \Big)\\
 & =  \dfrac{p^0 \gamma}{(\rho^0)^2} \Big(\nabla\cdot (\rho^0  u ) 
 - \tilde{s} \cdot u \Big)\Big(\nabla\cdot (\rho^0  \overline{v} ) 
 - \tilde{s} \cdot \overline{v} \Big) + \frac{(\rho^0)^2}{p^0 \gamma} (g_0' \cdot u) (g_0' \cdot \overline{v})\\
 &\hskip.5in  - 2 \mathfrak{S}\{(\tilde{s} \cdot u) (g_0' \cdot \overline{v})\}+ 2 \mathfrak{S}\{(g_0' \cdot \overline{u}) \nabla \cdot(\rho^0 \overline{v})\}.
\end{split}
\]
Now note that using again \eqref{EquiFluid} we have
\[
\tilde{s} \cdot u = \left ( \frac{\nabla \rho^0 \cdot g_0'}{\| g_0'\|^2} + \frac{(\rho^0)^2}{p^0 \gamma} \right ) (g_0' \cdot u) = \frac{\tilde{s} \cdot g_0'}{\|g_0'\|^2} (g_0' \cdot u)
\]
and so the previous calculation gives
\[
\begin{split}
p^0 \gamma ( \nabla\cdot u ) (\nabla \cdot \overline{v}) & =  \dfrac{p^0 \gamma}{(\rho^0)^2} \Big(\nabla\cdot (\rho^0  u ) 
 - \tilde{s} \cdot u \Big)\Big(\nabla\cdot (\rho^0  \overline{v} ) 
 - \tilde{s} \cdot \overline{v} \Big) + \frac{(\rho^0)^2}{p^0 \gamma} (g_0' \cdot u) (g_0' \cdot \overline{v})\\
 &\hskip.5in  - 2 \frac{\tilde{s} \cdot g_0'}{\| g_0'\|^2}(g_0' \cdot u)(g_0' \cdot \overline{v}) + 2 \mathfrak{S}\{(g_0' \cdot \overline{u}) \nabla \cdot(\rho^0 \overline{v})\}.
 \end{split}
\]
Combining the previous formulae finally proves \eqref{VolumeIntFluid}.
\end{proof}
\medskip

Lemma~\ref{loVt} completes the proof of the equivalence of $a_1(u,v)$
and $a_2(u,v)$ in the case when $[u\cdot \nu]_-^+ = [v \cdot \nu]_-^+
=0$ along $\Sigma^{F}$ and $[u]_-^+ = [v]_-^+ = 0$ along $\Sigma^{SS}$. If we additionally assume that $u$
and $T_0$ satisfy all of the boundary conditions in table~\ref{unT}
then Lemma~\ref{aOriga1lem} implies $a_{\text{original}}(u,v) = a_1(u,v)$ and
so the proof is complete.
\end{proof}
\medskip


%

\medskip

%

\section{The variational triple $\left( a_2,E, L^2(\tilde{X}, \rho^0 \,dx) \right)$} \label{functsec}

As mentioned in the introduction, using the idea of \cite{Valette},
one works with the sesquilinear form $(a_2, E, L^2(\tilde{X},
\rho^0 \,d V))$ whose definition is given in \eqref{a2} and
\eqref{DefE}.  In this section we study the properties of the
sesquilinear form $a_2$ in terms of its domain, continuity and
coercitivity.  We will show in Lemma \ref{Propa3} that $a_2$ is
Hermitian, closed, and coercive in Hilbert space $E$ relative to
$L^2(\tilde{X}, \rho^0 \,d V)$, with $E$ dense in $L^2(\tilde{X},
\rho^0 dV)$. By functional analysis, see for e.g. \cite[Section
  VI.3.2.5]{DautrayLionsV2}, this implies that $a_2$ corresponds to a
self-adjoint unbounded operator $(A_2, D(A_2))$ in $L^2(\tilde{X},
\rho^0 dV)$ with $D(A_2)$ dense in $E$ and $L^2(\tilde{X}, \rho^0 \,d
V)$.

\subsection{ Domain of $a_2$}\label{DoSes}

We initially defined $a_2$ in \eqref{a2} for smooth functions, but it can also be extended to $u$ and $v$ with reduced regularity, provided the volume and surface integrals still make sense. For the volume integral in the solid region $\Omega^S$  in \eqref{a2} to make sense, we need that, $u,v,\nabla u, \nabla v \in L^2(\Omega^S)$, that is, $u,v\in H^1(\Omega^S)$. This implies that we can take the trace of $u$ and $v$ from both sides of $\Sigma^{SS}$ by the classical trace theorem; thus this requirement also takes care of the validity of the surface integral along $\Sigma^{SS}$. For the volume integral in the fluid region $\Omega^F$ in \eqref{a2} to make sense, we need $u,v,\nabla \cdot u,\nabla \cdot v \in L^2(\Omega^F)$, the set of such functions is defined below as $H(\Div,\Omega^F)$. For such a function $u$, there is a generalized definition of the trace of $u\cdot \nu$ to the boundary, which lies in $H^{-1/2}(\partial \Omega^F)$, given by Lemma~\ref{tracelem}. This is not sufficient for the integrals over $\Sigma^{FF}$ and $\Sigma^{FS}$ in \eqref{a2} to make sense; we must also require in addition that $u\cdot \nu |_{\Sigma^{FF}}, \overline{v}\cdot \nu| _{\Sigma^{FF}}\in L^2$, the set of such functions are $H\left(\Div, \,\Omega^F,\, L^2\left(\partial \Omega^F \right)\right)$. These above observations give the motivation for the definition the Hilbert space $E$ defined in \eqref{DefE} and its role as the domain of $a_2$. 

We now state more concretely the definitions of the above mentioned function spaces. We have defined in Section 2 the Hilbert space $ L^2(\tilde{X} , \rho^0 dx) $
\begin{align*}
 L^2(\tilde{X} , \rho^0 dx) := \{  u : \int_{\tilde{X}} \lvert u\rvert^2 \rho^0 \, dV  < \infty  \} ;\\
 \langle u, v\rangle_{L^2(\tilde{X} , \rho^0 dx) }:= \int_{\tilde{X}} u \ \overline{v}\, \rho^0 \, d V .
 \end{align*}
For $V$ a bounded domain with Lipschitz boundary we also denote by $L^2(V) $ the usual unweighted space $L^2(V, dx)$ and make the following definitions.
\medskip

\begin{definition}
Denote by $\nu$ the outward point unit normal on $\partial V$.

\begin{equation}\label{SSp}
\begin{aligned}
 H(\Div, V) := \big\{ u \in L^2(V) :  \nabla \cdot u \in L^2(V)\big\}; \\ 
  \langle u, v\rangle _{H(\Div, V)} := \langle u, v\rangle _{L^2(V)} + \langle  \nabla \cdot u, \nabla \cdot v \big\rangle_{L^2(V)}.
  \end{aligned}
  \end{equation}
  \begin{equation}\label{FSp}
  \begin{aligned}
  H\big( \Div, V, L^2(\partial V)\big)
  = \big\{u \in L^2(V) :  \nabla \cdot u \in L^2(V) , \ \left . u \right |_{\partial V} \cdot \nu \in L^2(\partial V) \big\}; \\ 
(u, v)_{H( \Div, V, L^2(\partial V))}
 := \langle u, v\rangle _{L^2(V)} + \langle \nabla \cdot u, \nabla \cdot v \rangle_{L^2(V)}
 + \langle \left . u \right |_{\partial V} \cdot \nu,\left . v \right |_{\partial V}\cdot \nu \rangle_{L^2(\partial V)}  .
\end{aligned}
 \end{equation}
 
  \end{definition}
  \medskip
  
 \begin{lemma}\cite[Proposition 2 and Proposition 3, p 68]{Valette}
 \begin{itemize}
 \item 
 $ \Big( H(\Div, V), \langle \cdot, \cdot\rangle_{H(\Div, V)}\Big)$ is a Hilbert space with 
 $\mathcal{C}^{\infty} (\overline{V})$ dense in $ H(\Div, V)$.
 \item  $\Big( H\big( \Div, V, L^2(\partial V)\big), \langle \cdot, \cdot\rangle_{ H\big( \Div, V, L^2(\partial V)\big)}\Big)$ is a Hilbert space with $\mathcal{C}^{\infty} (\overline{V})$ dense in\linebreak $ H\big( \Div, V, L^2(\partial V)\big)$.
 \end{itemize}
 \end{lemma}
 \medskip
 
 \begin{lemma}\label{tracelem}
 The trace application  $u \in \mathcal{C}^{\infty} (\overline{V}) \mapsto (u\cdot \nu|_{\partial V})$ extends by continuity to a continuous linear map denoted by $\tr$,
 $$  \tr: H(\Div, V ) \rightarrow H^{-1/2}( \partial V) ,$$
defined for $u \in H(\Div, V )$ and $ \phi \in H^{1/2}(\partial V)$ as
\begin{equation}\label{generalizedG}
\langle \tr_u, \phi\rangle_{H^{-1/2}(\partial V), H^{1/2}(\partial V)} := \left( u, \nabla   R(\phi) \right)_{L^2(V)} + \left( \nabla \cdot u,  R(\phi)\right)_{L^2(V)}. \end{equation}
Here, $R$ is an extension operator with $R : H^{1/2}(\partial V) \rightarrow H^1(V) $ bounded and $R(\phi)|_{\partial V}= \phi$. 

 If $u\in H(\Div, V, L^2(\partial V)) $ then $\tr_u \in L^2(\partial V)$, and \eqref{generalizedG} becomes
 $$ (\tr_u, \phi)_{L^2(\partial V)} = \left( u, \nabla   R(\phi) \right)_{L^2(V)} + \left(\nabla \cdot u,  R(\phi)\right)_{L^2(V)} .$$
 \end{lemma}
 \medskip
 
\begin{definition}
\begin{equation}\label{DefE}
\begin{aligned}
E = \left\{ u \in L^2(\tilde{X},\rho^0 \, dx) : \begin{cases}
u|_{\Omega^S} \in H^1(\Omega^S) , \\ 
u|_{\Omega^F} \in H\big( \Div,\, \Omega^F,\, L^2(\partial \Omega^F) \big), \\
 [u\cdot \nu]^+_- = 0  \ \ \text{along}\ \Sigma^{FS} \cup \Sigma^{FF}  \end{cases}\right\};\\
(u, v) _E := \left( u|_{\Omega^S} , v|_{\Omega^S}\right) _{H^1(\Omega^S)} 
 + \left( u|_{\Omega^F}, v|_{\Omega^F}\right)_{H\left( \Div,\, \Omega^F, \, L^2\left(\Sigma^{FF} \cup (\partial \tilde{X})_F \right)\right)}.
 \end{aligned}
\end{equation}
\end{definition}
\medskip

\begin{proposition}\label{PropE}\cite[Proposition 14, p.104]{Valette}
\begin{itemize}
\item  $E$ equipped with inner product $( \cdot, \cdot)_E$ defined by \eqref{DefE}
is a separable Hilbert space.
\item The injective inclusion of $E$ into $L^2(\tilde{X}, \rho^0 \,dx)$ is continuous.
\item $E$ is dense in $L^2(\tilde{X}, \rho^0 \,dx)$.
\end{itemize}
\end{proposition}
\medskip

\begin{remark}\label{sepa2}
As a result, see for e.g. \cite{DautrayLionsV2}, we have the setting of a Hilbert triple 
$$ E\hookrightarrow L^2(\tilde{X}, \rho^0 \,dx) \hookrightarrow E'.$$
where each space is continuously, densely and injectively embedded in the next, denoted by $\hookrightarrow$. In addition, by definition, $E$ is a subspace of $L^2(X)$; as a result $E$ is also separable. This property of $E$ will be needed for the Galerkin method in Section \ref{Galerkin1}. 
\end{remark}

\medskip


\begin{lemma}\label{Propa3}
Suppose that $\rho^0$ is in $W^{1, \infty}$ when restricted to $\Omega^F$ and $\Omega^S$ with $\rho^0$ bounded away from zero, $p^0 \in L^{\infty}(\tilde{X})$ with $p_0$ bounded away from zero, $\gamma \in L^{\infty}(\Omega^F)$ with $\gamma$ bounded away from zero, and $g_0' \in L^{\infty}(\tilde{X})$ with $\|g'_0\|$ bounded away from zero. Then the sesquilinear form $a_2$ defined by \eqref{a2}
$$a_2(\cdot,\cdot) : E \times E \rightarrow \mathbb{C} , $$
is bounded, that is,
$$\lvert a_2(u,v)\rvert \leq C \lVert u\rVert_E \lVert v\rVert_E ;\ u, v\in E$$
for some constant $C >0$, and Hermitian. 
\end{lemma}
\medskip

\begin{proof}
The Hilbert space $E$ is intentionally designed so that $a_2$ is continuous on $E \times E$, and this follows simply by inspection of each of the terms in \eqref{a2}. For the fluid-solid boundary integral term, note that we must invoke the classical trace theorem. Similarly from \eqref{a2} we easily see
\[
a_2(u,v) = \overline{a_2(v,u)}
\]
which means that $a_2$ is Hermitian.
\end{proof}
\medskip


\subsection{Coercivity  of $a_2$}\

We remind the reader that the sesquilinear form $a_2$ is called coercive on the Hilbert space $E$ with respect to $L^2(\tilde{X},\rho^0 dx)$ if there exist constants $\alpha$ and $\beta >0$ such that for any $u \in E$
\[
a_2(u,u) \geq \alpha \| u \|_E^2 - \beta \| u \|^2_{L^2(\tilde{X},\rho^0 dx)}.
\]
Coercivity is a critical element in many proofs of well-posedness for equations such as \eqref{Asos2} and, together with continuity of $a_2$ on $E$, can also be interpreted as the statement that $\sqrt{a_2(u,u) + \beta \| u \|^2_{L^2(\tilde{X},\rho^0 dx)}}$ is a norm on $E$ equivalent to $\|u\|_E$. In this section we show that under certain hypotheses $a_2$ is coercive.

To begin we introduce some notation. First define $\pi$ by
$$\pi_{ijkl} = \Xi_{ijkl} + T^0_{ik} \delta_{jl} + \sigma_N( \delta_{ik} \delta_{jl} - 
 \delta_{ij} \delta_{kl}).$$
We can then rewrite the top order term in the solid region $\Omega^S$ in terms of $\pi_{ijkl}$ as 
$$\int_{\Omega^S}\Big ( ( \Lambda^{T^0} : \nabla u ) : \nabla \overline{v} +  \sigma_N  \nabla u : \nabla \overline{v}^T -  \sigma_N (\nabla\cdot u)( \nabla \cdot \overline{v}) \Big ) dV
 = \int_{\Omega^S} (\pi :\nabla u) : \nabla \overline{v} \ dV.$$ 
For two tensors $C,B$, denote by $\langle C, B\rangle$ the scalar product that is, $\langle C, B\rangle = C_{ij} \overline{B_{ij}}$. 
With this notation, we can write
$$\pi_{ijkl} \partial_j u_i \partial_l \overline{u}_k = \langle \pi : \nabla u , \nabla u \rangle.$$
\medskip 
 
\begin{theorem}\label{CoerciveA2}\
Denote by $\tau^0 = T^0 - \frac{1}{3} \mathrm{tr}(T^0) \mathrm{Id}$ the deviatoric stress. 
We make the following assumptions: 
\begin{itemize}
\item there exist $\mathfrak{c} > 0$ so that for all $2$-tensors $\eta_{ij}$
\begin{equation}\label{Hypo1}
\mathfrak{c} \lvert \eta_{ij} + \eta_{ji} \rvert^2 \leq  (\Xi_{ijkl} - p^0 \delta_{ik} \delta_{jl})  \eta_{kl} \overline{\eta}_{ij} ,
\end{equation}
in the solid region $\Omega^S$,
\item  $\Xi_{ijkl} \in L^\infty(\tilde{X})$, 
\item $\rho^0$ is piece-wise $W^{1, \infty}$ with $\rho^0$ bounded away from zero on $\tilde{X}$, 
\item $p^0 \in L^\infty(\tilde{X})$ with $p^0$ bounded away from zero, and $\nabla p^0 \in L^\infty(U)$ for $U$ a neighborhood of $\Sigma^{FS}$,
\item  $g_0' \in L^{\infty}(\tilde{X})$ with $\|g_0'\|$ bounded away from zero,
\item $\gamma \in L^\infty(\Omega^F)$ with $\gamma$ bounded away from zero,
\item  $[\rho^0]_-^+ \,(g_0' \cdot \nu) < -C < 0$ along $\Sigma^{FF}$ and $\rho^0 \,g_0' \cdot \nu >C >  0$ along $(\partial \widetilde{X})_F$ for a constant $C$; see Remark \ref{SignP} for further comments on this assumption.
\end{itemize}
For $\lVert \tau^0 \rVert_{L^{\infty}(\tilde{X})}$ sufficiently small,
$\sigma_N$ such that $$\lVert \sigma_N\rVert_{L^{\infty}(\Omega^S)}
\leq\lVert \tau_0\rVert_{L^{\infty}(\tilde{X})} ,\ \text{and}\ \| \nabla
\sigma_N \|_{L^\infty(\Omega^S)} \leq\| \nabla p^0 \|_{L^\infty(U)} ,$$
there exist $\alpha, \beta > 0$ such that
$$ a_2(u,u) \geq \alpha \lVert u \rVert^2_E - \beta \lVert u \rVert^2_{L^2(\tilde{X}, \rho^0 \,dx)}, \ \forall u \in E. $$
In other words, $a_2$ is $E$ coercive relative to $L^2(\widetilde{X}, \rho^0 \,dx)$. 
\end{theorem}
\medskip

\begin{remark}
Before beginning the proof we remark that the hypotheses of Theorem~\ref{CoerciveA2} may be modified so that a different portion of the tensor $\pi$ is assumed to satisfy the condition \eqref{Hypo1} and the remaining portion is assumed to be sufficiently small. There are many possible ways to do this choosing for example different values of $a$ and $b$ in \eqref{ElasticTensor} then using $\Gamma_{ijkl}$ instead of $\Xi_{ijkl}$.
\end{remark}

\begin{remark}\label{SignP}
Since $\rho^0$ on the exterior of $\widetilde{X} $ is zero, with the convention that the normal vector points outward the condition along $(\partial \widetilde{X})_F$ can be rewritten as
$$\rho^0 \,g_0' \cdot \nu = - [\rho^0]^+_-\, (g_0'\cdot \nu)  > C \Leftrightarrow [\rho^0]^+_- (g_0'\cdot \nu)  < -C.$$
Hence, the assumptions $\rho^0 g_0' \cdot \nu > 0$ along $(\partial \widetilde{X})_F$ and 
$[\rho^0]_-^+ (g_0' \cdot \nu) < 0$ along $\Sigma^{FF}$ are consistent with each other and the choice of orientation of normal vectors at each interface. In fact, these two assumptions can be rewritten as
$$ [\rho^0]^+_- \,(g_0'\cdot \nu) < -C , \ \ \text{along}\  \Sigma^{FF} \cup (\partial \widetilde{X})_F.$$
This condition on the fluid-fluid boundaries appears to be related to the local stability of the fluid. Indeed, on the fluid-fluid boundaries this is stating that a positive jump in density must occur in the same direction as the force due to gravity and rotation, a condition that, at least at an intuitive level, must be satisfied by a stably stratified fluid. On the surface $(\partial \widetilde{X})_F$ the condition is stating that the total force due to gravity and rotation must point downward.
\end{remark}
\medskip

\begin{proof}
We will break the proof of Theorem~\ref{CoerciveA2} into the following steps given by separate lemmas:
\begin{itemize}
\item Showing that the highest order term of the volume integral in the solid region satisfies G\aa rding's inequality. The result is given by 
Lemma \ref{CoerciveSolid2}.
\item Establishing an upper bound for the lower order terms in the solid region and along $\Sigma^{SS}$ given by Lemma \ref{LowerSolid}
\item Establishing a lower bound for the volume term in the fluid region and the boundary integral along $\Sigma^{FF}$ given by by Lemma \ref{CoerciveFluid}. 
\item Establishing an upper bound for the gravity and exterior boundary $\partial \tilde{X}$ terms. 
\end{itemize}
And now we present the lemmas. In the proofs we will use positive constants $C$ and $D$ which may change from step to step.
\medskip

 \begin{lemma}\label{CoerciveSolid2}
If $\mathfrak{c}$ satisfies \eqref{Hypo1}, $\|\sigma_N \|_{L^\infty(\Omega_S)}  \leq \|\tau_0\|_{L^\infty(\tilde{X})}$, and $\|\tau_0\|_{L^\infty(\tilde{X})}$ is sufficiently small then there exists $C$ and $D > 0$ so that 
  $$\int_{\Omega^S} \langle \pi : \nabla u, \nabla u \rangle \ dV
\geq C \lVert u\rVert^2_{H^1(\Omega^S)} - D  \lVert u\rVert_{L^2(\Omega^S)}^2.$$
 \end{lemma}
 \medskip
 
\begin{proof}
This lemma follows from Korn's Lemma (e.g. see \cite{Marsden}) and the point-wise estimate
\begin{equation}
\label{coerciveptwise}
\langle \pi : \nabla u, \nabla u \rangle \geq \mathfrak{c} \lvert \nabla u +  \nabla u^T \rvert^2 
- C \lVert \tau^0 \rVert_{L^\infty(\tilde{X})} \lvert \nabla u\rvert^2
\end{equation}
which holds for some constant $C>0$. This inequality can be established by noting that
\[
\langle \pi : \nabla u, \nabla u \rangle = \langle (\Xi_{ijkl} + T^0_{ik} \delta_{jl} + \sigma_N (\delta_{ik} \delta_{jl} - \delta_{ij} \delta_{kl})) : \nabla u, \nabla u \rangle,
\]
in the solid region $T^0_{ik} = - p_0 \delta_{ik} + \tau^0_{ik}$, and using $\lVert \sigma_N\rVert_{L^{\infty}(\Omega^S )} \leq\lVert \tau^0\rVert_{L^{\infty}(\tilde{X})}$ as well as the hypothesis \eqref{Hypo1}. By Korn's Lemma \eqref{coerciveptwise} implies that for some constant $c >0$
\[
\int_{\Omega^{S}} \langle \pi : \nabla u, \nabla u \rangle \ dV \geq c \lVert  u \rVert_{H^1(\Omega^S)}^2 - \lVert u \rVert_{L^2(\Omega^S)}^2  
- C \lVert \tau^0 \rVert_{L^\infty(\tilde{X})} \| \nabla u\|_{L^2(\Omega^S)}^2.
\]
Therefore if $\lVert \tau^0 \rVert_{L^\infty(\tilde{X})} < c/C$ the result is proven.
\end{proof}
\medskip

Now we continue to deal with the lower order terms in the solid region and the solid-solid boundary terms in the next lemma.
\medskip

\begin{lemma}
\label{LowerSolid}
Under the hypotheses of Theorem~\ref{CoerciveA2} we have for $u \in H^1(\Omega^S)$ and any $\epsilon >0$
$$\Bigg |\int_{\Omega^S} \Big ( -\mathfrak{S} \big\{  (g_0' \cdot u) ( \overline{u}\cdot \nabla \rho^0)  \big\} 
+ \mathfrak{S} \big\{  - ( \nabla \sigma_N  +\rho^0 g_0' ) \cdot u (\nabla \cdot \overline{u}  )\big\}
+ \mathfrak{S} \big\{     (  \nabla \sigma_N -  \rho^0 g_0' )\cdot \nabla u \cdot \overline{u} \big\}\Big ) dV;$$
$$- \int_{\Sigma^{SS} } \mathfrak{S} \big\{ [\rho^0]^+_- (u \cdot g_0' )( \overline{u}\cdot\nu)\big\} d \Sigma \Bigg | 
\leq C \epsilon \lVert u \rVert^2_{H^1(\Omega^S)} + \dfrac{1}{d(\epsilon)} \lVert u \rVert^2_{L^2(\tilde{X}, \rho^0 \,dx)}; $$
where $d(\epsilon)>0$ depends continuously on $\epsilon >0$ and $C>0$ is a constant.
\end{lemma}
\medskip

\begin{remark}
Throughout the proofs of Lemmas \ref{LowerSolid} and \ref{CoerciveFluid}, and Theorem \ref{CoerciveA2} we will write $d(\epsilon)$ for any positive function depending continuously on $\epsilon > 0$. Note that $d(\epsilon)$ may change from step to step.
\end{remark}

\begin{proof}
For general $f$ and $g \in H^1(\Omega^S)$ we have from the Cauchy-Schwarz inequality
$$\left \lvert \int_{\Omega^S} f \ \nabla g \ d V\right \rvert \leq \lVert f\rVert_{ L^2(\Omega^S)} \lVert \nabla g\rVert_{L^2(\Omega^S)}
\leq  \lVert f\rVert_{L^2(\Omega^S)} \lVert  g\rVert_{H^1(\Omega^S)}
\leq \epsilon  \lVert g\rVert_{H^1(\Omega^S)}^2 + \dfrac{1}{d(\epsilon)} \lVert f\rVert^2_{L^2(\Omega^S)}$$
for any $\epsilon >0$. Combining this with the hypotheses from Theorem~\ref{CoerciveA2} for $\rho^0$, $\sigma_N$, and $g_0'$ bounds the volume integral as required. For the surface integral we have from the classical trace theorem and the hypotheses on $g_0'$ and $\rho^0$
\[
\Bigg | \int_{\Sigma^{SS} } \mathfrak{S} \big\{ [\rho^0]^+_- (u \cdot g_0' )( \overline{u}\cdot\nu)\big\} \ d \Sigma \Bigg | \leq C \| u \|_{H^s(\Omega^S)}^2
\]
for some constant $C>0$ and any $s> 1/2$. If also $s<1$ then by Sobolev space interpolation and Young's inequality we have
\[
\|u \|^2_{H^s(\Omega^S)} \leq \| u \|^{2(1-s)}_{H^1(\Omega^S)} \| u \|^{2s}_{L^2(\Omega^S)} \leq \epsilon \| u \|^2_{H^1(\Omega^S)} + \frac{1}{d(\epsilon)} \| u\|^2_{L^2(\Omega^S)}
\]
and this can be used to find the required bound for the surface integral term.
\end{proof}
\medskip

Now we move to estimates for the integrals over the fluid region and the fluid-fluid interfaces.
\medskip

\begin{lemma}\label{CoerciveFluid}
Assuming the hypotheses of Theorem~\ref{CoerciveA2} and that $[\rho^0]^+_- (g_0' \cdot \nu)<0$ along $\Sigma^{FF}$ there are constants $C$ and $D >0$ such that for any $u \in E$
$$\int_{\Omega^F} \Bigg ( \dfrac{p^0 \gamma}{(\rho^0)^2}  \Big(\nabla\cdot (\rho^0  u ) 
 - \tilde{s} \cdot u \Big) \Big(\nabla\cdot (\rho^0  u ) 
 - \tilde{s} \cdot \overline{u} \Big)- \tilde{s}\cdot g_0'  \dfrac{(g_0' \cdot u) (\overline{u}\cdot g_0' ) }{\lVert g_0'\rVert^2} \Bigg ) \ dV
 -  \int_{\Sigma^{FF}} [\rho^0]^+_- (g_0' \cdot \nu)(u\cdot \nu) ( \overline{u}\cdot \nu) \ d \Sigma$$
 $$\geq  C \left ( \lVert \nabla \cdot u \rVert_{L^2(\Omega^F)}^2  + \| u \cdot \nu\|^2_{L^2(\Sigma^{FF})} \right ) - D \lVert u \rVert^2_{L^2(\tilde{X}, \rho^0 \,dV)}.$$ 
\end{lemma}
\medskip

\begin{proof}
The highest order term of the volume integral over $\Omega^F$ is easily seen to be bounded below by 
\[
C\int_{\Omega^F} \lvert \nabla \cdot u \rvert^2 \, dV
\]
for a constant $C>0$, and by the Cauchy-Schwarz inequality as in the proof of Lemma~\ref{LowerSolid} the lower order terms can all be bounded above by an expression
\[
C \epsilon  \int_{\Omega^F} \lvert \nabla \cdot u \rvert^2 \, dV + \frac{1}{d(\epsilon)} \lVert u \rVert^2_{L^2(\tilde{X}, \rho^0 \,dV)}
\]
for any $\epsilon >0$. Combining these and taking $\epsilon$ sufficiently small gives the required bound for the integral over $\Omega^F$. The estimate of the boundary term from below follows immediately from the hypotheses and thus the proof is complete.
\end{proof}

For the fluid-solid boundary term we have the estimate
 $$\int_{\Sigma^{FS}} 
 \mathfrak{S}\big\{  (\overline{u} \cdot\nu) \big( u_+ \cdot [\rho^0] g_0'\big)  \big\} d \Sigma
 \leq C \epsilon \lVert u\rVert^2_{H^1(\Omega^S)} + \tfrac{1}{d(\epsilon)} \lVert u \rVert^2_{L^2(\Omega^S)}$$
 for $u \in E$ and any $\epsilon >0$. This follows in the same way as the estimate of the solid-solid boundary terms in the proof of Lemma~\ref{LowerSolid}.

Next, for the gravitation term we have from \eqref{S(u)calc}
\[
( \nabla S(u), u)_{L^2(\tilde{X}, \rho^0 dV)} = -\frac{1}{4 \pi G}\| \nabla S(u) \|^2_{L^2(\mathbb{R}^3)}.
\]
The Cauchy-Schwarz inequality and bound for $\rho^0$ therefore imply that
\[
\| \nabla S(u) \|_{L^2(\mathbb{R}^3)} \leq C \|u \|_{L^2(\tilde{X}, \rho^0 dV)}
\]
for a constant $C$ independent of $u$.

Finally, the exterior boundary integral is decomposed into 
$$\int_{\partial \tilde{X}} \rho^0 (u\cdot g_0') (\overline{u}\cdot \nu) \ d \Sigma
 = \int_{(\partial \tilde{X})_S} \rho^0 (u\cdot g_0') (\overline{u}\cdot \nu) \ d \Sigma
 + \int_{(\partial \tilde{X})_F} \rho^0 (u\cdot g_0') (\overline{u}\cdot \nu) \ d \Sigma.$$
The part along $(\partial \tilde{X})_S$ can be estimated in the same manner as the solid-solid interface integrals by
\[
\Bigg | \int_{\partial \tilde{X}} \rho^0 (u\cdot g_0') (\overline{u}\cdot \nu) \ d \Sigma
\Bigg | \leq \epsilon \lVert u\rVert^2_{H^1(\Omega^S)}
+ \dfrac{1}{d(\epsilon)} \lVert u \rVert^2_{L^2(\tilde{X}, \rho^0\, dV)}
\]
for any $\epsilon >0$.

As for the part along $(\partial \tilde{X})_F$, using the fact that along $(\partial \tilde{X})_F$ we have $g_0' = (g_0' \cdot \nu) \nu$, we rewrite this term as
\begin{equation}\label{IntF}
 \int_{(\partial \tilde{X})_F} \rho^0 (u\cdot g_0') (\overline{u}\cdot \nu) \ d \Sigma
=  \int_{(\partial \tilde{X})_F} \rho^0 (g_0'\cdot \nu) (u\cdot \nu) (\overline{u}\cdot \nu) \ d \Sigma
\end{equation}
With the assumption that 
$$ \rho^0( g_0' \cdot \nu)\ > C > 0 , \ \text{along} \ (\partial \widetilde{X})_F, $$
the term \eqref{IntF} contributes directly to the $E$ coercitivity since it is bounded below by
$$ > C \lVert u\cdot \nu\rVert^2_{L^2( (\partial \tilde{X})_F)}.$$
Combining all of the previous estimates and taking the $\epsilon$'s to be sufficiently small when necessary completes the proof of Theorem~\ref{CoerciveA2}.
\end{proof}
\medskip

\subsection{The unbounded operator $A_2$ defined from $(E, L^2(\widetilde{X}, \rho^0\,dx), a_2)$}\label{A2Op}

We give a brief description of the unbounded operator on
$L^2(\widetilde{X}, \rho^0\,dx)$ defined from the variational triple
$(E, L^2(\widetilde{X}, \rho^0\,dx), a_2)$. This correspondance is
described in more detail in \cite[Section VI.2.5]{DautrayLionsV2}.

As a result of Proposition~\ref{PropE}, we are in the setting of the following Hilbert triplet, 
$$E \hookrightarrow L^2(\widetilde{X}, \rho^0\,dx)\hookrightarrow E',$$
where $E$ is continuously embeded in $L^2(\widetilde{X}, \rho^0\,dx)$ with dense image. Here, we denote by $E'$ the Banach dual of $E$. We also have, from Theorem~\ref{CoerciveA2}, that $a_2$ is coercive on $E$ with respect to $L^2(\widetilde{X}, \rho^0\,dx)$; that is, there are $\alpha > 0$ and $\beta \in \RR$ such that
\begin{equation}\label{a2coerc}
a_2(u,u) + \beta \lVert u\rVert_{L^2(\widetilde{X}, \rho^0\,dx)}^2 
\geq \alpha \lVert u\rVert_E^2,\  \forall \,u \in E.
\end{equation}
Define
\begin{equation}\label{a2beta}
(a_2+\beta)(u,v):= a_2(u,v) + \beta\ (u,v)_{L^2(\widetilde{X}, \rho^0\,dx)} .
\end{equation}
Then $a_2+\beta$ is a bounded sesquilinear form on $E\times E$ and is $E$ coercive. We will assume in the remainder of the paper that the hypotheses of Theorem~\ref{CoerciveA2} are satisfied, and so \eqref{a2coerc} holds.

Now, \cite[Theorem 6, p.368]{DautrayLionsV2} gives an isomorphism between sesquilinear forms bounded on $E\times E$ and $\mathcal{L}(E, E')$. Here we denote by $\mathcal{L}(E,E')$ the set of bounded linear operators from $E$ to $E'$. Thus corresponding to $a_2+\beta$ is operator $A_2 +\beta = A_2 + \beta \Id \in \mathcal{L}(E,E')$ defined by
\[(a_2+\beta)(u,v) = \langle (A_2 +\beta \Id) u, v\rangle_{E', E} ,\ \ \forall u , v\in E,\]
where $\langle, \rangle_{E',E}$ is the duality pairing between $E'$ and $E$. By the Lax-Milgram Theorem, since $(a_2+\beta)$ is $E$ coercive, $A_2+\beta : E \rightarrow E'$ is an isomorphism, see for e.g.
\cite[Theorem 7, p 368]{DautrayLionsV2}. We can restrict $A_2+\beta \in \mathcal{L}(E,E')$ to $D(A_2+\beta)$ defined as
\begin{equation}\label{DoDef}
D(A_2+\beta ) :=  \{ u \in E :  v \mapsto (a_2+\beta) (u,v) \ \text{is continuous on} \ E \ \text{with the topology of}\ L^2(\widetilde{X}, \rho^0\,dx) \}.
\end{equation}
and obtain the unbounded operator $(A_2+\beta, D(A_2+\beta))$ on $L^2(\widetilde{X}, \rho^0\,dx)$. Since $a_2+\beta$ is $E$ coercive and Hermitian, $A_2+\beta$ enjoys the following properties, see for e.g. \cite[Propositions 9 and 10 , p 371]{DautrayLionsV2},
\medskip
\begin{proposition}\label{PropOp}
The following holds true
\begin{itemize}
\item $A_2+\beta \in \mathcal{L}(E,E')$ and $A_2+\beta : E \rightarrow E'$ is an isomorphism. In fact $A_2 + \beta : D(A_2+ \beta) \rightarrow L^2(\tilde{X},\rho^0\,dx)$ is an isomorphism with $(A_2+\beta)^{-1} : L^2(\tilde{X},\rho^0\,dx)\rightarrow L^2(\tilde{X},\rho^0\,dx)$ . 
\item $D(A_2+\beta)$ is dense in $L^2(\tilde{X},\rho^0\,dx)$ and in $E$.
\item $(A_2+\beta, D(A_2+\beta))$ is closed and self-adjoint as an unbounded operator on $L^2(\tilde{X},\rho^0\,dx)$.
\end{itemize}
\medskip
\end{proposition}
From the definition of $a_2+\beta$ in \eqref{a2beta} and the definition of $D(A_2+\beta)$ in \eqref{DoDef}, we have $D(A_2+\beta) = D(A_2)$. Since $\beta \in \RR$, $(A_2, D(A_2))$ inherits the same properties as $A_2 + \beta$, which we summarize in the following proposition.
\medskip
\begin{proposition}\label{DomA2}
The following holds true
\begin{itemize}
\item $A_2 \in \mathcal{L}(E,E')$. 
\item $D(A_2)$ is dense in $L^2(\tilde{X},\rho^0\,dx)$ and in $E$.
\item $(A_2, D(A_2))$ is closed and self-adjoint as an unbounded operator on $L^2(\tilde{X},\rho^0\,dx)$.
\end{itemize}
\end{proposition}

\medskip
\section{Well-posedness of the weak formulation via semi-group method}\label{semigroup:sec}

In the previous section we obtained the setting
$E\hookrightarrow H \hookrightarrow E'$ with operators $A_2 \in \mathcal{L}(E,E')$, and $R_{\Omega} \in \mathcal{L}(H,H)$. With $T > 0$ we can extend the classical problem to the space of vector-valued distributions by looking for solutions to
 \begin{equation}\label{2ndorderevo}
 \ddot{u} + 2 R_{\Omega} \dot{u} + A_2 u = f, \ \text{in} \ \mathcal{D}'(0,T; E').
 \end{equation}
Here $R_{\Omega}\, u$ is the matrix multiplication between $R_{\Omega}$ and $u$ with $R_{\Omega}$ defined by \eqref{ROm}. We would also like to add the initial conditions $u(0) = g$ and $\dot u(0) = h$, but it is not apparent a priori how to incorporate these conditions with no additional regularity assumed. This issue will be addressed in detail in section \ref{AbsCauchy}. One can either study Equation \eqref{2ndorderevo} directly, or the associated first-order system obtained from \eqref{2ndorderevo} by reduction of order. Well-posedness for the first-order system can be obtained either via semi-group theory, discussed in sections \ref{SemiGroupReview} and \ref{InfGen}, or by Galerkin method, cf. section \ref{Galerkin1}.
 
 Let us first write out the first-order problem obtained from \eqref{2ndorderevo} by reduction of order. We define the following product space
$$\mathcal{H} := E \times H;\ \   H = L^2(\widetilde{X}, \rho^0\,dx),$$
equipped with the scalar product $(,)_{\mathcal{H}}$ defined by  
\begin{equation}\label{SPH}
\left ( \begin{pmatrix} u_1 \\ u_2\end{pmatrix} , \begin{pmatrix} v_1 \\ v_2 \end{pmatrix} \right )_{\mathcal{H}} :=  (a_2+\beta)(u_1, v_1)+   \left( u_2, v_2 \right)_H.
\end{equation}
Here $\beta$ is the coercivity constant of $a_2$ which gives that $a_2(u,u) \geq \alpha \| u \|_E^2 - \beta \| u \|^2_H$ for $u\in D(A_2)$, given by Theorem \ref{CoerciveA2}. In addition $a_2$ is Hermitian, see Lemma \ref{Propa3}. As a result, $\mathcal{H}$ equipped with the scalar product \eqref{SPH} is a Hilbert space, and the norm induced by the scalar product $(a_2 +\beta)$ is equivalent to the norm $\lVert \cdot \rVert_E$. For $(u_1, u_2)^T \in D(A_2) \times  E \subset \mathcal{H}$, we define the following unbounded operator on $\mathcal{H}$
$$\tilde{A}_2 \begin{pmatrix} u_1 \\ u_2 \end{pmatrix}:= \begin{pmatrix}  0 & \Id \\ -A_2  & -2R_{\Omega} \end{pmatrix} \begin{pmatrix} u_1 \\ u_2 \end{pmatrix} = \begin{pmatrix} u_2 \\ -A_2 u_1 -2R_{\Omega}  u_2 \end{pmatrix} ,$$
with domain $D(\tilde{A}_2) = D(A_2) \times E$. Since $D(A_2)$ is dense in $E$ and in $H$ by Proposition \ref{DomA2} , and $E$ is dense in $H$ by Proposition \ref{PropE}, $D(\tilde{A}_2) $ is dense in $\mathcal{H}$; in other words, $\tilde{A}_2$ is an unbounded densely defined operator on $\mathcal{H}$.  Via reduction of order, we can formally write \eqref{2ndorderevo} as the following first-order abstract Cauchy problem on $\mathcal{H}$ 
  \begin{equation}\label{Asos2}
 \dot{U} = \tilde{A}_2U+ F  , \ \text{in} \ \mathcal{D}'(0,T; H \times E')
 \end{equation}
also with an initial condition $U(0) = U_0$. Here we have formally set $U = (u,\dot{u})^T$, $F = (0,f)^T$, and $U_0 = (u(0),\dot{u}(0))^T$. The two formulations, \eqref{2ndorderevo} and \eqref{Asos2}, are equivalent in the sense that a solution of either one will give a solution of the other. The formulation \eqref{Asos2} is best adapted to analysis via semi-group method, and this is the approach we follow in this section. In section \ref{From1to2} we look in detail how to prove that a solution of \eqref{Asos2} gives a corresponding solution of \eqref{2ndorderevo}.

For the further analysis, we will also need to know the adjoint of $\widetilde{A}_2$ which is given in the following lemma.

\medskip
\begin{lemma}\label{A2*adjoint}
The following holds true
\[
\Big ( \widetilde{A}_2^*,\ D(\widetilde{A}^*_2) \Big ) = \Bigg (
\left ( \begin{matrix}
0 & - \mathrm{Id} + \beta (A_2 + \beta)^{-1}\\
A_2 + \beta & 2 R_{\Omega} 
\end{matrix}
\right ),\
D(\widetilde{A}_2)
\Bigg )
\]
\end{lemma}
\begin{proof}
We first check that the domain is correct. The domain of the adjoint is defined to be
\[
D(\widetilde{A}^*_2) = \{ U \in \mathcal{H} \ : \ D(\widetilde{A}_2) \ni V \mapsto (\widetilde{A}_2 V, U )_{\mathcal{H}} \mbox{ extends to a continuous functional on $\mathcal{H}$} \}.
\]
For $V \in D(\widetilde{A}_2)$ we have, using the notation $V = (v_1,v_2)$ and $U = (u_1,u_2)$,
\begin{equation}\label{toextend1}
\begin{split}
(\widetilde{A}_2 V, U )_{\mathcal{H}} & = \left ( \left ( \begin{matrix}
v_2 \\ -A_2 v_1 - 2 R_\Omega v_2
\end{matrix} \right ), \ \left (\begin{matrix} u_1\\u_2
\end{matrix} \right ) \right )_{H}\\
& = a_2(v_2,u_1) + \beta (v_2,u_1)_{H} - (A_2 v_1,u_2)_{H} - 2 (R_\Omega v_2,u_2)_{H}
\end{split}
\end{equation}
First note that if $U \in D(\widetilde{A}_2) = D(A_2) \times E$, then using that $a_2$ is Hermitian, we have from \eqref{toextend1}
\[
(\widetilde{A}_2 V, U )_{\mathcal{H}} =  \overline{(A_2 u_1,v_2)_H} + \beta (v_2,u_1)_{H} - a_2( v_1,u_2) - 2 (R_\Omega v_2,u_2)_{H}
\]
which is a bounded functional for $V \in \mathcal{H}$. Thus $D(\widetilde{A}_2) \subset D(\widetilde{A}_2^*)$. For the opposite inclusion we begin by taking $v_1 = 0$ to find that if $U \in D(\widetilde{A}^*_2)$, then
\[
E \ni v_2 \mapsto a_2(v_2,u_1) + \beta (v_2,u_1)_{H}
\]
extends to a bounded linear functional on $H$. Therefore $u_1 \in D(A_2 + \beta) = D(A_2)$. On the other hand, taking $v_2 = 0$ we see that if $U \in D(\widetilde{A}^*_2)$, then
\begin{equation}\label{toextend}
D(A_2) \ni v_1 \mapsto (A_2 v_1,u_2)_{H}
\end{equation}
extends to a bounded linear functional on $E$. This implies that $u_2 \in E$ as we now argue. Indeed, suppose that $u_2 \in H$ and let $l_{u_2}$ be the bounded functional on $E$ which extends \eqref{toextend}. Then since $a_2+\beta$ is $E$ coercive, by the Lax-Milgram theorem there exists a $\widetilde{u}_2 \in E$ such that
\[
a_2(v_1,\widetilde{u}_2) + \beta (v_1, \widetilde{u}_2) = l_{u_2}(v_1)+ \beta (v_1,u_2)_H
\]
for all $v_1 \in E$. Suppose we set $v_1 = (A_2 + \beta)^{-1} (u_2-\widetilde{u}_2)$. Then $a_2(v_1,\widetilde{u}_2) = (A_2 v_1, \widetilde{u}_2)_H$ since $v_1 \in D(A_2)$, and so we find that
\[
\big ( (A_2 + \beta) v_1, \widetilde{u}_2 \big )_H = \big ( (A_2+\beta) v_1,u_2 \big )_{H} \Rightarrow \big (u_2-\widetilde{u}_2,u_2-\widetilde{u}_2\big )_H = 0.
\]
Therefore $u_2 = \widetilde{u}_2 \in E$. Thus $D(\widetilde{A}^*_2) \subset D(\widetilde{A}_2)$, and so we have proven that $D(\widetilde{A}^*_2) = D(\widetilde{A}_2)$.

To prove that the given formula is correct for the adjoint operator we calculate for $V \in D(\widetilde{A}_2)$ and $U \in D(\widetilde{A}^*_2)$
\begin{equation} \label{step}
\begin{split}
\left ( \left ( 
\begin{matrix}
v_1\\
v_2
\end{matrix}
\right ), \left (
\begin{matrix}
-u_2 + \beta (A_2 + \beta)^{-1} u_2\\
(A_2 + \beta)u_1 + 2 R_\Omega u_2
\end{matrix}
\right ) \right )_{\mathcal{H}}
& = -a_2(v_1,u_2) - \beta (v_1,u_2)_{H} + (v_2,(A_2 + \beta)u_1)_{H}\\
& + 2 ( v_2,R_\Omega u_2)_{H} + \beta (a_2 + \beta)( v_1 , (A_2 + \beta)^{-1} u_2)\\
& = a_2(v_2,u_1) + \beta (v_2,u_1)_{H} - (A_2 v_1,u_2)_{H} - 2 (R_\Omega v_2,u_2)_{H}.
\end{split}
\end{equation}
Note that we have used $R_\Omega^* = -R_{\Omega}$ with respect to the inner product on $H$. Since this coincides with \eqref{toextend1} the proof is complete.
\end{proof}
 

\medskip
The rest of the section is dedicated to obtaining a well-posedness result for this problem via semi-group theory. We will first show in subsection \ref{InfGen} and Theorem \ref{WellPosedNess} that $\tilde{A}_2$ is the infinitesimal generator of a semigroup of class $\mathcal{C}^0$. In subsection \ref{AbsCauchy}, we will discuss how one `solves' the abstract Cauchy problem \eqref{Asos2}, given such a property of $\tilde{A}_2$. After this is done, we will discuss how a solution of the first-order system obtained in this manner gives a solution of \eqref{2ndorderevo}, cf. subsection \ref{From1to2}. 
Before tackling these tasks, we make a short digression in subsection \ref{SemiGroupReview} to state the essential definitions and facts we will need from semigroup theory.

\subsection{Summary of semigroup theory}\label{SemiGroupReview}
 Our discussion follows \cite[Chapter XVII]{DautrayLionsV5}, where proofs and more details can be found. First we state the definition of a semigroup of class $\mathcal{C}^0$.
\medskip  
  \begin{definition}
Let $\mathcal{B}$ be a real or complex Banach space provided with norm $\lVert\cdot\rVert_\mathcal{B}$. 
A family $\{ G(t)\}_{t\geq 0}$ of elements $G(t) \in \mathcal{L}(\mathcal{B})$ for $t\geq 0$ is said to form a semigroup of class $\mathcal{C}^0$, or just a $\mathcal{C}^0$ semigroup, on $\mathcal{B}$ if it satisfies the following conditions:
 \begin{enumerate}
 \item $G(t+s) = G(t) \circ G(s)$ for $t, s\geq 0$ .
 \item $G(0) = \Id$ in $\mathcal{L}(\mathcal{B})$. 
 \item $\lim_{t\rightarrow 0^+} \lVert G(t)X - X\rVert_{\mathcal{B}} = 0$, for all $X \in \mathcal{B}$. 
 \end{enumerate}
\end{definition}
Following from its definition, for a semigroup $\{G(t)\}_{t\geq 0}$ of class $\mathcal{C}^0$ over Banach space $\mathcal{B}$, there exist real constants $\beta$ and $M$ such that $\lVert G(t)\rVert \leq Me^{\beta t}$ for $t \geq 0$.  When $M = 1$, the semigroup is called quasi-contractive; in the special case where $M=1$ and $\beta = 0 $, the semigroup is called contractive. Next we state the definition of the infinitesimal  generator of a $\mathcal{C}^0$ semigroup.
\medskip
\begin{definition}
The infinitesimal generator $\tilde{A}$ of semigroup $\{G(t)\}_{t\geq 0}$ of class $\mathcal{C}^0$ is the linear unbounded operator given by 
\medskip
\begin{equation}\label{limitSem}
\tilde{A}(X) := \lim_{t\rightarrow 0^+} \dfrac{G(t)X - X}{t}\end{equation}
on the domain $D(\tilde{A})$ defined to be the set of those $X \in \mathcal{B}$ such that the limit \eqref{limitSem} exists in $\mathcal{B}$. 
\end{definition}
From its definition, it follows that $D(\tilde{A})$ is a dense vector subspace of $\mathcal{B}$. 
\smallskip
\begin{remark}[Well-posedness and infinitesimal generator]\label{DiffSet}
If $X\in D(\tilde{A})$, then the function $t\mapsto G(t) X$ is once continuously differentiable from $[0, \infty)$ to $\mathcal{B}$ and one has
$$\dot{G}(t) X = \tilde{A} G(t) X = G(t) \tilde{A}X .$$
Hence $G(t)X$ is a strong solution to $\dot{V} = AV$, $V(0) = X$. In fact it is also the unique strong solution. 

Conversely, suppose for each $X \in D(\tilde{A})$, the problem
$\dot{U}= \tilde{A} U , U(0) = X$ has a unique solution, denoted by $G(t)X$, then 
$\{G(t)\}$ is a semigroup of class $\mathcal{C}^0$, with $\tilde{A}$ being its infinitesimal generator. 
\end{remark}
\smallskip
 \begin{remark}[Relation between quasi-contractive and contractive semi-groups]\label{QuasiCon}
Suppose that $\tilde{A}-c$ is the infinitesimal generator of a contraction semigroup of class $\mathcal{C}^0$ denoted by $G(t)$. For all $X \in D(\tilde{A}-c) = D(\tilde{A})$, we have $V(t):= G(t)X$ solves
\begin{equation}\label{eqn2}
\dot{V} = (\tilde{A}- c) V  ; \ \ V(0) = X.
\end{equation}
Now, $V(t)$ is a solution of \eqref{eqn2} if and only if 
 $U(t):= e^{ct} V  = e^{ct} G(t) x$ solves
 \begin{equation}\label{eqn1}
 \dot{U} = \tilde{A} U ,  \ \ U(0) = X.
 \end{equation}
Thus \eqref{eqn1} has existence and uniqueness of solution for all initial conditions in $D(\tilde{A})$. In addition, with $\{G(t)\}$ being contractive, we have
 $$\lVert V(t) X\rVert_{\mathcal{B}} \leq \lVert X \rVert_{\mathcal{B}} \Rightarrow \lVert U(t) X\rVert_{\mathcal{B}}  \leq e^{ct} \lVert X\rVert_{\mathcal{B}}  .$$
This implies, by Remark \ref{DiffSet}, that $\tilde{A}$ is the infinitesimal generator of a quasi-contractive semigroup of class $\mathcal{C}^0$, namely $e^{ct} G(t)$.   
 
\end{remark}
\medskip
The Hille-Yosida Theorem, cf. \cite[Chapter XVII, Section 3]{DautrayLionsV5}, gives the necessary and sufficient conditions for an operator densely defined on a Banach space to be the infinitesimal generator of a semigroup of class $\mathcal{C}^0$. For our problem, we will use a corollary of the Hille-Yosida Theorem, the Hille-Phillips Theorem \cite{HillePhilips}. Here we follow the version of the theorem given in \cite[Chapter XVII, Section 3,Theorem 7]{DautrayLionsV5}, 
\begin{theorem}[Hille-Phillips \cite{HillePhilips}] \label{HPtheorem}
Let $\tilde{A}$ be an unbounded operator with domain $D(\tilde{A})$ dense in a Hilbert space $\mathcal{H}$. Then $\tilde{A}$ is the infinitesimal generator of a contractive semigroup of class $\mathcal{C}^0$ if and only if
\begin{enumerate}
\item $\tilde{A}$ is dissipative, that is, $\text{Re} (\tilde{A}X,X)_\mathcal{H}\leq 0$ for all $X \in D(\tilde{A})$; and
\item the image of $D(\tilde{A})$ by $\text{Id} - \tilde{A}$ is equal to $\mathcal{H}$.
\end{enumerate}
\end{theorem}
\smallskip
\begin{remark}
A proof of the theorem can be found in \cite[Chapter XVII, Section 3, Theorem 7]{DautrayLionsV5}. By the proof, one can in fact replace the second condition by one that requires the existence of $\lambda > 0$ such that the image of $D(\tilde{A})$ by $\lambda \text{Id} - \tilde{A}$ is equal to $\mathcal{H}$, cf. \cite[Chapter XVII, Section 3, Remark 7]{DautrayLionsV5}. In fact, this is the version that we will use to show that  the operator $(\tilde{A}_2 - c, D(\tilde{A}_2))$ is the infinitesimal generator of a quasi-contractive semigroup of class $\mathcal{C}^0$, for some $c > 0$. 
\end{remark}

\subsection{$\tilde{A}_2$ is the infinitesimal generator of a quasi-contractive semigroup}\label{InfGen} 

For some $c >0$ we will show that $\tilde{A}_2-c$ generates a contractive semigroup by showing that it satisfies the hypothesis of Theorem \ref{HPtheorem}. Then Remark \ref{QuasiCon} gives the result for $\tilde{A}_2$. Recall that we assume the hypotheses of Theorem \ref{CoerciveA2} hold, and that $\alpha$ and $\beta$ are the coercivity constants for $a_2$.
 \begin{lemma}\label{WellPosedNess}
 \begin{enumerate}
 \item $D( \tilde{A}_2) $ is dense in $\mathcal{H} = E \times H$
\item There exists $c  > 0$ such that for all  $U \in D( \tilde{A}_2)$
$$\text{Re}( \tilde{A}_2 \, U, U )_{\mathcal{H}} \leq c ( U , U)_{\mathcal{H}}. $$ 
\item $\lambda - \tilde{A}_2$ is onto $\mathcal{H}$ for sufficiently large $\lambda$.
\end{enumerate}
  \end{lemma}
  \medskip

\begin{proof}\

The first criterion follows simply from the fact that $D(A_2)$ is dense in $E$ by Proposition \ref{DomA2},  and so we have $D( \tilde{A}_{2}) = D(A_2) \times E$ is dense in $\mathcal{H} = E \times H $. 

For the second criterion we calculate for $U = (u_1,u_2) \in D(A_2)$
\begin{align*}
\text{Re}\left ( \tilde{A}_2\begin{pmatrix} u_1 \\ u_2\end{pmatrix} , \begin{pmatrix} u_1 \\ u_2\end{pmatrix} \right )_{\mathcal{H}}
&= \text{Re}\left [(a_2+\beta)( u_2, u_1) + (-A_2 u_1, u_2 )_H -  (  2 R_{\Omega} u_2, u_2 )_H \right ] \\
& = \text{Re}\left [ a_2 ( u_2, u_1) + \beta ( u_2, u_1 )_{H} - a_2 ( u_1, u_2 )  \right ]\\[0.2cm]
& =  \text{Re}\left [\beta ( u_1, u_2 )_{H} \right ]\\
&\leq \dfrac{1}{2} \beta \lVert u_1 \rVert_H^2 + \dfrac{1}{2} \beta \lVert u_2 \rVert_H^2 \\
& \leq \dfrac{1}{2} \beta  \lVert u_1 \rVert_E^2 + \dfrac{1}{2} \beta  \lVert u_2 \rVert^2_H  \\
&\leq c \big(  \lVert u_1 \rVert_{a_2+\beta}^2 +  \lVert u_2 \rVert^2_H \big) 
= c \, \left ( \begin{pmatrix} u_1 \\ u_2\end{pmatrix},  \begin{pmatrix} u_1 \\ u_2\end{pmatrix} \right )_{\mathcal{H}}
\end{align*}
on $E$.

Finally for the third criterion we let $\begin{pmatrix} f \\ \tilde{f}\end{pmatrix}\in \mathcal{H}$ and consider the solvability of
$$\lambda \begin{pmatrix} u_1 \\ u_2\end{pmatrix} -\tilde{A}_2\begin{pmatrix} u_1 \\ u_2\end{pmatrix} 
  = \begin{pmatrix} f \\ \tilde{f}\end{pmatrix}.$$
This last formula holds if and only if
  $$
  \begin{cases} u_2 = \lambda u_1 - f\\
  \lambda u_2 + A_2 u_1 + 2R_{\Omega} u_2 = \tilde{f} \end{cases}\Leftrightarrow \begin{cases}   u_2 = \lambda u_1 - f \\
A_2 u_1 + 2 R_{\Omega}(\lambda u_1 - f)  + \lambda (\lambda u_1- f) =  \tilde{f}
  \end{cases}.$$
This solvability follows from that of the following problem
  $$ \Big( A_2 + \lambda 2 R_{\Omega} + \lambda^2  \Big) u_1 =  \tilde{f} + 2 R_{\Omega} f + \lambda f \in H$$
which is guaranteed by the following Lemma \ref{invertible} for $\lambda > \mathbf{c}$.
\end{proof}
\medskip

 \begin{lemma}\label{invertible}
 There exists $\mathbf{c} > 0$ so that for any $\lambda > \mathbf{c}$ and $g \in H$ there exists $u \in D(A_2)$ such that
 $$\big( A_2 + \lambda 2 R_{\Omega} \big)u + \lambda^2 u  =  g .$$
 \end{lemma}
 \medskip
 
 \begin{proof} 
 For $r, s\in \RR$ and $g \in H$ consider the solvability of the problem
 \begin{equation}\label{InterStep}
   A_2 u +  r 2R_{\Omega} u + s u = g.
   \end{equation}
For a fixed $r\in \RR$, we apply the result of Proposition \ref{PropOp}. To begin we introduce the bilinear form
$$ (a_2 + r 2 R_{\Omega}) (u,v):= a_2(u,v) + r ( 2 R_{\Omega} u, v)_H$$
which is E coercive relative to $H$ with constants $\alpha$ and $\beta+\lvert r\vert \lVert 2R_{\Omega}\rVert:$
$$(a_2 + r2R_{\Omega} ) (u,u) \geq {\alpha} \lVert u \rVert^2_E 
- \big(\beta+\lvert r\vert \lVert 2R_{\Omega}\rVert \big) \lVert u \rVert^2_H. $$
As in Proposition \ref{PropOp}, by  \cite[Propositions 9 and 10 , p 371]{DautrayLionsV2}, the operator
\[
A_2 + r\ R_\Omega + s\ \mathrm{Id}
\]
maps $D(A_2)$ onto $H$ for
$$ s> \beta + |r| \lVert 2 R_{\Omega}\rVert. $$
Now note that if $\lambda \in \RR$ satisfies $\lambda > \mathbf{c}$ where $\mathbf{c}$ is given by
 $$\mathbf{c} = \dfrac{1}{2}\Big(\lVert 2 R_{\Omega}\rVert + \sqrt{ \lVert 2 R_{\Omega}\rVert^2 + 4 {\beta} }\Big)$$
 then $\lambda$ satisfies
  $$\lambda^2 >{\beta} +  \lvert \lambda \rvert \lVert 2 R_{\Omega}\rVert .$$ 
Combining the previous conclusions we find that for $\lambda > \mathbf{c}$, with $r = \lambda^2$ and $s =  \lambda$, the operator
  $$A_2 + \lambda 2 R_{\Omega} + \lambda^2  : D(A_2) \rightarrow H$$
  is onto.
 \end{proof}
 \medskip

\noindent As a result of the Hille-Phillips Theorem and Remark \ref{QuasiCon}, Lemma~\ref{WellPosedNess} gives
\medskip

\begin{theorem}\label{qcontA2}
If the hypotheses of Theorem \ref{CoerciveA2} are satisfied, then there exists a $c > 0$ such that {$(\tilde{A}_2 -c , D(\tilde{A}_2))$ generates a contractive semigroup on $\mathcal{H}$. Hence $(\tilde{A}_2 , D(\tilde{A}_2))$ generates a quasi-contractive semigroup on $\mathcal{H}$.} 
\end{theorem}
\medskip

\begin{remark}
The resolvent set is the set of values of $\lambda$
for which $[F(\lambda) + A_2]$, with $F(\lambda) = \lambda 2
R_{\Omega} + \lambda^2$, has a bounded inverse. The spectrum,
$\sigma$, is given by its complement in $\mathbb{C}$. While noting that
$\lambda \to F(\lambda)$ is bounded on $H$ for all $\lambda \in \mathbb{C}$,
we obtain the general property
\[
   \lambda \in \sigma \Rightarrow
   |\mathrm{Re}\ \lambda | \le c
\]
where $c$ is the constant from part 2 of Lemma \ref{WellPosedNess}. To obtain eigenfrequencies, one identifies $\lambda = i \omega$.
\end{remark}
\medskip

\subsection{Well-posedness results for the first-order problem}\label{AbsCauchy}

We begin by discussing the various notions of solution to the abstract Cauchy problem \eqref{Asos2}. We will consider weak and strong solutions on both infinite and finite time intervals. First, we introduce a class of distributions taking values in a Banach space.
\medskip

\begin{definition}
Let $\mathcal{D}_-$ denote the subset of $C^\infty(\mathbb{R})$ consisting of functions with support limited to the right (that is, $f \in \mathcal{D}_-$ if $f \in C^\infty(\mathbb{R})$ and $f(t) = 0$ for all $t$ sufficiently large). For a Banach space $Y$, denote by $\mathcal{D}_+'(Y)$ the space of continuous linear maps from $\mathcal{D}_-$ to Y. 
\end{definition}
\medskip

\noindent For $U_0 \in \mathcal{H}$ and $F \in \mathcal{D}'_+(\mathcal{H})$, we wish consider the problem \eqref{Asos2} with initial condition $U(0) = U_0$, which must be properly reformulated to allow for distributional solutions. Indeed, if $U$ is a classical solution of \eqref{Asos2} extended as zero for $t < 0$, and $F$ is also extended as zero there, then
\begin{equation} \label{LaplaceCauchy} 
 \dot{U} - \tilde{A}_2 U = F + \delta \otimes U_0
 \end{equation}
holds in the sense of distributions $\mathcal{D}'_+(\mathcal{H})$. 
\medskip

\begin{definition}\label{DefVectorVaSol}
 By a vector-valued distribution solution of \eqref{Asos2}, we mean a vector-valued distribution $U \in \mathcal{D}'_+ (D(\tilde{A}_2))$ satisfying \eqref{LaplaceCauchy}. Here $D(\tilde{A}_2)$ is made into a Banach space using the graph norm.
\end{definition}
\medskip

By Theorem \ref{qcontA2}, $\tilde{A}_2$ generates a quasi-contractive semigroup on $\mathcal{H}$, which we denote by $\{G(t)\}_{t\geq 0}$.   
This property of $\tilde{A}_2$ allows one to use the Method of Laplace transform, e.g. see \cite[Chapter XVI, Section 3, Theorem 2]{DautrayLionsV5}, to obtain a vector-valued distributional solution. 
\medskip

\begin{proposition}\label{VectorVaSol}
With $U_0 \in \mathcal{H}$ and $F \in \mathcal{D}'_+(\mathcal{H})$
$$\dot{U} = \tilde{A}_2U+ F;  \ \  U(0) = U_0$$
 has a unique solution in the sense of Definition \ref{DefVectorVaSol}, that is, in the space of distributions with values in $D(\tilde{A}_2)$. This solution is given by
 $$U = G \overset{t}{\ast} ( F + \delta \otimes U_0).$$ 
If we further assume that $F \in L^1_{\text{loc}}(0,\infty; \mathcal{H}) $ then the above solution can be written as,
\begin{equation}\label{VectorVSol}
U(t) = \int_0^t G(t-s) F(s) \, ds + G(t) U_0 .
\end{equation}
\end{proposition}
\medskip

\noindent It is straightforward to prove the following ``energy estimate" for $U(t)$ based on \eqref{VectorVSol} and the fact that $\|G(t)\|_{\mathcal{H} \rightarrow \mathcal{H}} \leq e^{ct}$.
\medskip

\begin{corollary}\label{Energyest:cor}
If $F \in L^1_{\text{loc}}(0,\infty; \mathcal{H})$, $U_0 \in \mathcal{H}$, and $U$ is the solution given by \eqref{VectorVSol}, then
\[
\|U(t)\|_{\mathcal{H}} \leq e^{ct} \big ( \|F\|_{L^1(0,t; \mathcal{H})} + \|U_0\|_{\mathcal{H}} \big ).
\]
The constant $c$ is the same as the constant in the second part of Lemma \ref{WellPosedNess}.
\end{corollary}
\medskip

One can also impose more hypotheses on the data $U_0$ and $F$, so that vector-valued distribution solution obtained in \eqref{VectorVSol} has more regularity, and thus solves the equation in a `stronger' sense. However, whether one is going to obtain more regularity or not, uniqueness of solution is guaranteed by Proposition \ref{VectorVaSol}.  

One scenario in which one has improved regularity is when the initial condition $U_0 \in D(\tilde{A}_2)$.  By Remark \ref{DiffSet}, under this hypothesis, the map $t\mapsto G(t)U_0$ is differentiable in the classical sense; as a result, $G(t)U_0$ is the classical solution to $\dot{U} = \tilde{A}_2 U$, $U(0) = U_0$. To consider the regularity of the other term in \eqref{VectorVaSol}, define the space $H^1_{loc}(0,\infty;\mathcal{H})$ to by
\[
H^1_{loc}(0,\infty;\mathcal{H}) = \left \{ F: [0,\infty) \rightarrow \mathcal{H} \ : \ \|F\|_{L^2(0,T;\mathcal{H})} , \ \|\dot{F}\|_{L^2(0,T;\mathcal{H})} < \infty \quad \mbox{for all $T>0$} \right \}.
\]
A straightforward extension of \cite[ChapterXVII, Part B, Proposition
  1]{DautrayLionsV5} shows that if $F\in H^1_{loc}(0,\infty;
\mathcal{H})$ then the function
$$t \mapsto \int_0^t G(t-s) F(s)\, ds $$
has a derivative in the classical sense in $\mathcal{H}$. This means that with $U_0 \in D(\tilde{A}_2)$ and $F\in H_{loc}^1(0,\infty; \mathcal{H})$, the solution defined by \eqref{VectorVSol} actually defines a strong (classical) solution. We summarize the discussion in the form of Definition \ref{StrongSolution} and Proposition \ref{StrongOnHalfline}.
\smallskip
\begin{definition}\label{StrongSolution}
$U$ is a strong (classical) solution of \eqref{Asos2} if $U$ satisfies the following requirements:
\begin{enumerate}[label= \arabic*.]
\item $U \in \mathcal{C}^0( [ 0, \infty) , \mathcal{H}) \cap \mathcal{C}^1( (0, \infty) ; \mathcal{H})$. 
\item $U(t) \in D(\widetilde{A}_2)$ for $t > 0$. 
\item $U$ satisfies \eqref{Asos2} pointwise. 
\end{enumerate}
\end{definition}
\medskip
\begin{proposition}\label{StrongOnHalfline}
Suppose that the hypotheses of Theorem~\ref{CoerciveA2} are satisfied so that by Theorem \ref{qcontA2} $\tilde{A}_2$ generates a quasi-contractive semigroup on $\mathcal{H}$ denoted $\{G(t)\}_{t\geq 0}$. If $U_0 \in D(\tilde{A}_2)$ and $F \in H^1_{loc}(0,\infty; \mathcal{H}) $, then \eqref{Asos2} has a unique strong solution in the sense of Definition \ref{StrongSolution}. This solution is given by \eqref{VectorVSol} which is differentiable in the classical sense in $\mathcal{H}$. 
\end{proposition}\

\medskip
\noindent Another case in which we have a classical solution occurs when $F$ is continuous taking values in $D(\tilde{A}_2)$, and $\tilde{A}_2F$ is locally integrable (see \cite[ChapterXVII, Part B, Proposition 2]{DautrayLionsV5}).

 Our next question is whether the distributional solution will have better regularity when $U_0$ is not in $D(\tilde{A}_2)$. In the case where $F=0$, the unique distributional solution in $\mathcal{D}'_+(\mathcal{H})$ is $G(t)U_0$ by Proposition \ref{VectorVaSol}. This unique solution would be a strong one if and only if the semigroup $\{G(t)\}_{t\geq 0}$ is differentiable.  We do not expect this to be true as our equation is hyperbolic, and we expect that singularities will propagate (that is, the range $G(t)$ is not contained in $D(\tilde{A}_2)$ for any $t$, and so $\{G(t)\}_{t\geq 0}$ is not differentiable). In other words, when $U_0\notin D(\tilde{A}_2)$, one cannot generally hope for a classical solution. On the other hand, when the problem is posed on a finite time interval $(0,T)$ and with the assumption that $F \in L^1(0,T; \mathcal{H})$, for all $U_0 \in \mathcal{H}$, 
$U$ defined by \eqref{VectorVSol} has a meaning as a continuous function from $[0,T]$ to $\mathcal{H}$. In fact, expression \eqref{VectorVSol} will be called a weak solution of \eqref{Asos2}. To facilitate further discussion, we first state the definition of strong solution and various notions of weak solution of \eqref{Asos2} on the finite interval $(0,T)$, $T >0$.
\medskip

\begin{definition}[Strong solution \cite{DautrayLionsV5}]\label{StrongFinite}
Suppose $F\in L^1(0,T; \mathcal{H})$. A distribution $U$ is called a strong solution of \eqref{Asos2} on $(0,T)$ if it satisfies the following properties,
\begin{enumerate}[label= \arabic*.]
\item $U\in \mathcal{C}^0([0,T]; \mathcal{H}) , \dot{U} \in L^1(0,T; \mathcal{H}),$
\item $U(0) = U_0$,
\item $U(t) \in D(\tilde{A}_2), \ \text{almost all}\  t \in [0, T],$
\item $ U(t)\ \text{satisfies}\  \eqref{Asos2}\  \text{in}\  L^1(0,T; \mathcal{H}).$
\end{enumerate}
Here the derivative $\dot{U}$ is taken in the sense of distributions.
\end{definition}
\medskip
\begin{definition}[Definition 3.1, p.64 in \cite{Brezis} - Weak solution]\label{WeakBrezis}
Suppose $F\in L^1(0,T; \mathcal{H})$. 
$U\in \mathcal{C}^0([0,T],\mathcal{H})$ is called a weak solution of \eqref{Asos2} on $(0,T)$ if $U(0) = U_0$, and there exist sequences $\{F_n\}\in L^1(0,T;\mathcal{H})$, $\{U_n\} \in \mathcal{C}^0([0,T],\mathcal{H})$ such that 
\begin{enumerate}[label= \arabic*.]
\item 
$U_n$ are strong solutions, in the sense of Definition \ref{StrongFinite}, to
$$\dot{U}_n = \tilde{A}_2 U_n + F_n ;$$ 
\item $F_n$ converges to $F$ in $L^1(0,T;\mathcal{H})$;
\item $U_n$ converges to $U$ in $\mathcal{C}^0([0,T],\mathcal{H})$.
\end{enumerate}
\end{definition}
\medskip
\begin{definition}[Ball \cite{Ball} and Balakrisknan \cite{Balakrishnan} - Weak solution]\label{WeakSG}
Suppose $F(t) \in L^1(0,T; \mathcal{H})$.\\ 
 $U \in \mathcal{C}^0([0,T]; \mathcal{H})$ is a weak solution of \eqref{Asos2} on $(0,T)$
 if and only if $U(0) = U_0$ and for every $V \in D(\tilde{A}_2^*)$ the function $t\mapsto ( U(t), V)_{\mathcal{H}}$ is absolutely continuous on $[0,T]$ and
 $$ \dfrac{d}{dt} (U(t), V)_{\mathcal{H}} = ( U(t), \tilde{A}_2^* V)_{\mathcal{H}} + (F(t), V)_{\mathcal{H}},  \text{for almost all}\  t \in [0,T].$$
 \end{definition}

We combine the result of Proposition \ref{StrongOnHalfline} (restricted to $[0,T]$) and the theory found in \cite[Chapter XVII, Part B]{DautrayLionsV5}, and the second part of \cite[Theorem 4.8.3]{Balakrishnan}, or \cite{Ball}, to obtain the following result describing the well-posedness of \eqref{Asos2} on a finite time interval.

\medskip
\begin{proposition}\label{AbsCauchyWP}
Suppose that the hypotheses of Theorem~\ref{CoerciveA2} are satisfied, so that $\tilde{A}_2$ generates a quasi-contractive semigroup on $\mathcal{H}$ denoted $\{G(t)\}_{t\geq 0}$.  If $F \in L^1( 0,T;\mathcal{H})$, then any weak or strong solution of $\dot{U} = \tilde{A}_2 U + F$ and $U(0) = U_0$ is given by \eqref{VectorVSol}

If $U_0 \in \mathcal{H}$ and $ F \in L^1( 0,T;\mathcal{H})$, then the solution given by \eqref{VectorVSol} is the unique weak solution in the sense of Definitions \ref{WeakBrezis} and \ref{WeakSG}.

If $U_0 \in D(\tilde{A}_2)$ and $ F \in L^1( 0,T;\mathcal{H})$, then the solution given by \eqref{VectorVSol} is a strong solution in the sense of Definition \ref{StrongFinite}

If $U_0 \in D(\tilde{A}_2)$ and $F \in H^1(0,T;\mathcal{H}) $, then the solution given by \eqref{VectorVSol} is a strong solution in the sense of Definition \ref{StrongSolution} restricted to the interval $[0,T]$. In this case, the strong solution satisfies the equation in the classical sense. 
\end{proposition} 
\medskip

\subsection{Existence of solution for second order equation}\label{From1to2}
Having established results for the first-order equation \eqref{Asos2}, we now go back to the original second order equation \eqref{2ndorderevo}. We start by analyzing what a weak solution of \eqref{Asos2} gives in terms of a solution for \eqref{VartionalP2}. Recall that there are two equivalent definitions of weak solution of \eqref{Asos2}, given by Definition \ref{WeakSG} and \ref{WeakBrezis}. In this section we explicitly consider only the case of a finite time interval, but similar results can be found for the case of $\mathbb{R}^+$.
\medskip

\begin{lemma}\label{WeakSol1to2}
Suppose $U\in \mathcal{C}^0([0,T];\mathcal{H})$ is a weak solution of \eqref{Asos2} with $F(t) = \begin{pmatrix} 0 \\ f\end{pmatrix} \in L^1(0,T;\mathcal{H})$. Denote by $u(t)$ the first component of $U$. Then $u(t) \in  \mathcal{C}^0([0,T]; E) \cap \mathcal{C}^1([0,T];H)$ satisfies
$$\forall v \in E ,\ \  \dfrac{d}{dt} ( \dot{u}, v) = -a_2(u,v)
- ( 2R_{\Omega} \dot{u},  v)_H + 
 (f, v)_H ,\ \ \text{in}\ \mathcal{D}'(0,T). $$
\end{lemma}

\begin{proof}
 By Definition \ref{WeakSG}, 
 a weak solution to \eqref{Asos2} is $U\in \mathcal{C}^0([0,T];\mathcal{H})$ that satisfies : 
\begin{align*}
\forall \ V\in D(\tilde{A}_2^*),\ \ 
  \dfrac{d}{dt} (U(t), V)_{\mathcal{H}} = ( U(t), \tilde{A}_2^* V)_{\mathcal{H}} +& (F(t), V)_{\mathcal{H}}, \\
  & \text{for almost all}\  t \in [0,T].
  \end{align*}
In this definition, $(U(t),V)_{\mathcal{H}}$ is required to be absolutely continuous, and thus $\tfrac{d}{dt}$ can be taken in the distribution sense, or in the sense of the pointwise derivative. In other words, the above expression can be written as
\begin{equation}\label{weakSolA*} 
\forall \ V\in D(\tilde{A}_2^*),\ \  \dfrac{d}{dt} \left(U(t), V\right)_{\mathcal{H}} = \left( U(t), \tilde{A}_2^* V\right)_{\mathcal{H}} + \left(F(t), V\right)_{\mathcal{H}},  \ \text{in}\ \mathcal{D}'( 0,T).
 \end{equation}
Using the notation $U(t) = (u_1(t),u_2(t))$ and $V(t) = (v_1(t),v_2(t))$, by equation \eqref{step} in the proof of Lemma \ref{A2*adjoint}, \eqref{weakSolA*} is equivalent to
 \begin{equation}\label{ComponentWeakSol}
\begin{aligned}
\dfrac{d}{dt} 
(a_2+\beta)(u_1,v_1) 
+ \dfrac{d}{dt} (u_2, v_2)_H =& -a_2(u_1, v_2)
+ \left(u_2, (A_2+\beta) v_1 + 2R_{\Omega} v_2\right)_H \\
&+ 
 (f, v_2) ,\ \ \text{in}\  \mathcal{D}'(0,T).
 \end{aligned}
 \end{equation}
Choosing $v_2 = 0 $ in \eqref{ComponentWeakSol}, we obtain
$$ \dfrac{d}{dt} (a_2+\beta)( u_1,v_1) 
= \left(u_2 ,(A_2+\beta) v_1\right)_H  \ \text{in}\  \mathcal{D}'(0,T).$$
Since $v_1\in D(A_2)$ we have
$$ (a_2+\beta)( u_1,v_1)  = \overline{ (a_2+\beta)( v_1,u_1) } 
= \overline{( (A_2+\beta) v_1, u_1)}_H = (u_1, (A_2+\beta)v_1)_H.$$
Thus,
$$\forall v_1\in D(A_2) ,   
\dfrac{d}{dt} \left(u_1, (A_2+\beta\right)v_1)_H  = \left(u_2 ,(A_2+\beta) v_1\right)_H   , \ \text{in}\  \mathcal{D}'(0,T). $$
We recall that $(A_2+\beta): D(A_2) \rightarrow H$ is onto, and so the last equation implies
$$\forall w \in H, \dfrac{d}{dt}\, (u_1, w)_H  = \left(u_2 ,w\right)_H   , \ \text{in}\  \mathcal{D}'(0,T). $$
As a result of Lemma \ref{commute}, we have $u_1 \in \mathcal{C}^1([0,T],H)$ and $\dot{u}_1 = u_2$. Now choosing $v_1 = 0 $ in \eqref{ComponentWeakSol} we obtain
$$\forall v_2 \in E ,\ \  \dfrac{d}{dt} (u_2, v_2)_H = -a_2( u_1, v_2)
+ \left(u_2,  2 R_{\Omega}^*v_2\right)_H + 
 (f, v_2)_H, \ \text{in}\  \mathcal{D}'(0,T).$$
Substitute $\dot{u}_1 = u_2$ in the above expression, and note $u = u_1$,  to obtain
 $$\forall v \in E ,\ \  \dfrac{d}{dt} ( \dot{u}, v)_H = -a_2(u, v)
- ( 2 R_{\Omega} \dot{u},  v)_H + 
 (f, v)_H ,\ \ \text{in}\ \mathcal{D}'((0,T)). $$
This completes the proof.
\end{proof}
\medskip

Next, we consider strong solutions showing what a strong solution of
\eqref{Asos2} gives in terms of \eqref{VartionalP2}.
  
\medskip
\begin{proposition} \label{StrongSolution2}
Suppose that the hypotheses of Theorem~\ref{CoerciveA2} are
satisfied. If $f \in H^1(0,T;H),\ g \in D(A_2) $, and $h\in E$, then
problem \eqref{VartionalP2} has a unique strong solution $u \in
\mathcal{C}^1([0,T], H) \,\cap \,\mathcal{C}^2((0,T), H)$, which
satisfies \eqref{IVP} pointwise on $(0,T)$. If we denote by $U(t) = (
u ,\dot{u})$, then $U(t)$ is given by
\eqref{VectorVSol}. Additionally, for every $t \in (0,T)$, $u(t) \in
D(A_2)$, $\dot{u}(t) \in E$, and $u(t)$ is classically
differentiable as a map into $E$.
\end{proposition}
\medskip

\begin{proof}

\textbf{Existence}: Under the current assumption on $f, g$ and $h$, we have $(0,f)\in H^1(0,T, \mathcal{H})$ and $( g,h ) \in D(\tilde{A}_2)$. By Proposition \ref{AbsCauchyWP}, the problem
\begin{equation}\label{sample}
\dfrac{d}{dt} U(t) = \tilde{A}_2 U + \begin{pmatrix} 0 \\ f \end{pmatrix} ; U(0) = \begin{pmatrix} g\\ h \end{pmatrix}
\end{equation}
has a unique strong solution $U(t) \in \mathcal{C}^0( [0,T] , \mathcal{H}) \cap \mathcal{C}^1( (0,T) , \mathcal{H})$ given by \eqref{VectorVSol}; that is $U(t) \in D(\tilde{A}_2)$ for all $t\in (0,T)$ and
$$ \dfrac{d}{dt} U(t)  = \tilde{A}_2 U + \begin{pmatrix} 0 \\ f \end{pmatrix}  \ \text{in}\  \mathcal{C}^0( (0,T), \mathcal{H})$$
where the derivative is taken in the classical sense. Using the notation $U(t) = (u_1(t),u_2(t))$, we have for all $t \in (0,T)$, $u_1(t) \in D(A_2)$, $u_2(t) \in E$, and
$$ u_1 \in \mathcal{C}^0([0,T]; E) ;\  \  u_2 \in \mathcal{C}^0([0,T], H)  ;$$
$$ u_1 \in \mathcal{C}^1((0,T); E);\ \  u_2 \in  \mathcal{C}^1((0,T), H)  .$$ 
When $\tilde{A}_2 = \begin{pmatrix} 0 & \Id \\ -A_2 & -2R_{\Omega} \end{pmatrix} $, the Cauchy problem \eqref{sample} can be writen as
$$\begin{pmatrix} \dot{u}_1 \\ \dot{u}_2 \end{pmatrix} = \begin{pmatrix} u_2\\ -A_2 u_1- R_{\Omega} u_2\end{pmatrix}
+ \begin{pmatrix} 0\\ f\end{pmatrix},\  \text{in}\  \mathcal{C}^0((0,T), \mathcal{H}),$$
which componentwise gives
$$\begin{cases}  
 \dot{u}_1 = u_2\  \text{in}\ \  \mathcal{C}^0((0,T), E ) \\
 \dot{u}_2 +  A_2 u_1 + R_{\Omega} u_2  = f \ \text{in} \ \mathcal{C}^0((0,T), H).
 \end{cases}.$$
Since $\dot{u}_1 = u_2 \in \mathcal{C}^1((0,T), H)$ can replace $u_2$ by $\dot{u}_1$ and $\dot{u}_2$ by $\ddot{u}_1$, with all derivatives taken in the classical sense, and obtain
$$\begin{cases} \ddot{u}_1 +  A_2 u_1 + 2R_{\Omega} \dot{u}_1  = f \ \text{in} \ \mathcal{C}^0((0,T), H);\\
u_1(0) = g ;\  \dot{u}_1(0) = h.  \end{cases}.$$
Taking $u = u_1$ we see we have proven that a solution with claimed properties exists.

\textbf{Uniqueness}: On the other hand, under the same assumptions on $f, g$ and $h$, if
 we have a classical solution of \eqref{IVP} that is, $u(t) \in  \mathcal{C}^0([0,T], E) \cap \mathcal{C}^2((0,T), H)$ which satisfies \eqref{IVP}  pointwise on $(0,T)$, then $U(t) = \begin{pmatrix} u(t) \\ \dot{u}(t) \end{pmatrix}$ is a classical solution to \eqref{Asos2}  with initial condition $\begin{pmatrix} g\\ h \end{pmatrix} $ and inhomogeneous term $\begin{pmatrix} 0 \\ f\end{pmatrix}$. The uniqueness of such a $U(t)$ is guaranteed by Proposition \ref{AbsCauchyWP}, which then implies the uniqueness of $u$. 
 \end{proof}
 \medskip

To finish this section we make a few remarks on conservation of energy. Let us define the energy associated to our problem as
\begin{equation}\label{Energy}
E(t) = \frac{1}{2} \left ( \| \dot{u} \|_H^2 + a_2(u,u) \right )
\end{equation}
which would be $\|(u,\dot{u})\|_{\mathcal{H}}$ upon setting $\beta = 0$. If $u$ is the strong solution for problem \eqref{VartionalP2} with $f = 0$ shown to exist in Proposition \ref{StrongSolution2}, then making use of the regularity properties given in the proposition  and the fact $u$ satisfies \eqref{IVP} pointwise we have
\[
\begin{split}
\dot{E}(t) & = \mathrm{Re} \big [ ( \ddot u, \dot u)_H + a_2(\dot u , u) \big ]\\
& = \mathrm{Re} \big [ ( -2 R_\Omega \dot u - A_2 u , \dot u)_H + a_2(\dot u , u) \big ]\\
& = 0.
\end{split}
\]
Thus we see that $E(t)$ is constant. By using the convergence of strong solutions to weak solutions as in Definition \eqref{WeakBrezis} we can establish the same result for weak solutions of problem \eqref{VartionalP2}. Note that this is in some ways an improvement on the energy estimate given in Corollary \ref{Energyest:cor} in the $\mathcal{H}$ norm.

The total physical energy \cite[(3.226)-(3.227)]{Dahlen} is
given in Appendix~C. To show the equivalence of the
total physical energy conservation and \eqref{Energy}, we have to resort to a
second-order tangential slip condition at fluid-solid boundaries,
$\Sigma^{F}$, given in \eqref{2ndorderBC} replacing the linearized one \eqref{FcontinuityBC} used in the main text. A proof of the equivalence of the total physical energy and \eqref{Energy} under the additional nonlinear slip condition is given in Appendix~C.

\section{Galerkin Approximation}\label{Galerkin1}

We restate our original problem which was defined in \eqref{VartionalP2}
\begin{equation}\label{Galerkinprob1}
\left \{
\begin{split}
& \text{Given} \ u_0 \in E ,\ \dot{u}_0 \in H,\  f\in L^2(0,T; H)\\[0.2cm] 
&\text{Find}\ u\ \text{satisfying}: u \in \mathcal{C}^0( [0,T]; E),\ \dot{u} \in \mathcal{C}^0([0,T]; H),\ \text{such that} \\
 &\forall v \in E ,  \dfrac{d}{dt} (\dot{u}, v)_H + (2R_{\Omega} \dot{u} , v)_H + a_2(u,v) = (f, v)_H   , \ \text{in}\ \mathcal{D}'(0,T);\\
& u(0) = u_0\  ; \ \dot{u}(0) = \dot{u}_0 . 
\end{split}
\right .
 \end{equation}
Existence and uniqueness of solutions to \eqref{Galerkinprob1} can be proven via the Faedo-Galerkin method combined with parabolic regularization as shown in \cite[Theorem 3, p572, Theorem 4, p574]{DautrayLionsV5}. The problem \eqref{Galerkinprob1} can also be written in its equivalent form as, cf. \eqref{IVP},
\begin{equation}\label{eqnInDist}
\ddot{u} + 2 R_{\Omega} \dot{u} + A_2 u = f ,\ \text{in}\ \mathcal{D}'(0,T; E')
\end{equation}
where elements of $H$ are identified with elements of $E'$ via the $H$-inner product. As we have done in the semigroup approach, we can also replace \eqref{eqnInDist} formally by a system of first-order equations. A precise formulation of \eqref{Galerkinprob1} as a first-order system is
\begin{equation}\label{Galerkinprob2}
\left \{
\begin{split}
& \text{Given} \ U_0 \in \mathcal{H},\  (0,f) \in L^2(0,T; \mathcal{H})\\[0.1cm] 
&\text{Find}\ U\ \text{satisfying}: U \in \mathcal{C}^0( [0,T]; \mathcal{H}),\ \text{such that} \\
 &\dot{U} = \widetilde{A}_2 U + F   , \ \text{in}\ \mathcal{D}'(0,T;H \times E');\\
& U(0) = U_0.
\end{split}
\right .
\end{equation} 
We remind the reader that we are using the notation $\mathcal{H} = E \times H$ as in section \ref{semigroup:sec}. We will also take the inner product on $\mathcal{H}$ to be that given by \eqref{SPH}. As discussed in section \ref{semigroup:sec}, existence and uniqueness of solutions for \eqref{Galerkinprob2} follows from existence and uniqueness of solutions for \eqref{Galerkinprob1}. Our plan in this section is to review the Faedo-Galerkin method applied directly to \eqref{Galerkinprob2}, and to provide a proof of strong convergence of the Galerkin approximations in this case. Alternatively, well-posedness of both \eqref{Galerkinprob2} and \eqref{Galerkinprob1} can be deduced from the results of section \ref{semigroup:sec}. 
Convergence of the Galerkin approximations for many classes of abstract problems closely related to the problem considered in this paper have been established previously (e.g. see \cite{Joly} for a review, or \cite{Dupont} for a classical account of the standard approach), and the method of proof that we use below is essentially the same.


Standard Galerkin approximation for \eqref{Galerkinprob2} would proceed as follows. Let $\{X_m\}$ be an increasing sequence of finite dimensional subspaces of $E$ such that $\cup_{m=1}^\infty X_m$ is dense in $E$, and define $\mathcal{E}_m = X_m \times X_m$. Suppose that for every $m$, $\{W_{j,m}\}_{j=1}^{d_m}$ is a basis for $\mathcal{E}_m$, and introduce the matrices
\[
M_{ij} = ( W_{j,m},W_{i,m})_{H\times H}, \quad K_{ij} = \langle \tilde{A}_2 W_{j,m}, W_{i,m} \rangle_{(E\times E)', (E\times E)}.
\]
For the definition of $K_{ij}$, we identify elements of $H \times H$ with elements of $(E \times E)'$ via the $H$-inner product. The corresponding Galerkin approximations for \eqref{Galerkinprob2} would then be
\[
U_m(t) = \sum_{j=1}^{d_m} g_{m,j}(t) W_{j,m}
\]
where the coefficients $g_{m,j}$ are obtained by solving the finite dimensional problems
\begin{equation}\label{Galerkinapprox}
\sum_{j = 1}^{d_m} M_{ij}\ \dot{g}_{m,j}(t) = \sum_{j=1}^{d_m} K_{ij}\ g_{m,j}(t) + (F, W_{i,m})_{H\times H}, \quad \sum_{j = 1}^{d_m} ( W_{j,m},W_{i,m})_{\mathcal{H}}\ g_{m,j}(0) = (U_0,W_{i,m})_{\mathcal{H}} 
\end{equation}
Since $\{W_{j,m}\}$ is linearly independent, $M_{ij}$ is symmetric positive definite, and this finite system has a unique solution for every $m$. Of course application of this method in practice requires much more work including proper design of the approximation spaces $\mathcal{E}_m$, and construction of the matrix $K_{ij}$ which will involve boundary integrals as well as a nonlocal contribution from the effects of self-gravitation. We do not comment further on these issues, although do note that the Volterra equation method outlined in section \ref{Voltsec} provides a way in which the nonlocal self-gravitation effects could be dealt with separately. 

The Faedo-Galerkin method involves demonstrating that the approximations $U_m$ converge to a solution of \eqref{Galerkinprob2}. The first step is to establish bounds, so-called energy estimates, on $U_m$ which will imply weak convergence of subsequences. Note that since $\{W_{j,m}\}$ is a basis for $\mathcal{E}_m$, we have from the definition of the Galerkin approximation that 
\[
\left (\dot{U}_m, V \right )_{H\times H} = \left \langle \tilde{A}_2 U_m, V\right \rangle_{(E \times E)', E \times E} + (F, V)_{H \times H}
\]
for all $V \in \mathcal{E}_m$. We will prove in the next lemma an energy estimate which establishes bounds for $\|U_m\|_{\mathcal{H}}$.
\medskip

\begin{lemma} \label{Energy1}
Suppose that $f \in L^2(0,T; H)$, 
$W \in L^2(0,T;\mathcal{E}_m)$ and
\begin{equation}\label{Galerkinenergy}
-\int_0^T \left (W(t), \dot{V}(t) \right )_{H\times H} \ \mathrm{d} t = \int_0^T\left \langle \tilde{A}_2 W(t), V(t) \right \rangle_{(E \times E)', E \times E} + (F(t), V(t))_{H \times H} \ \mathrm{d} t, \quad 
\end{equation}
where $F = (0,f)$, for all $V \in \mathcal{C}^1([0,T];\mathcal{E}_m)$ with $V(0) = V(T) = 0$. Then $\dot{W} \in L^2(0,T; \mathcal{H})$ and there is a constant $C$ independent of $m$ such that for all $t \in [0,T]$
\[
\|W (t) \|^2_{\mathcal{H}} \leq C \left ( \|f\|^2_{L^2(0,t;H)} 
+ \|W(0)\|^2_{\mathcal{H}} \right ).
\]
\end{lemma}
\begin{proof}
First of all, we claim that $\dot{W} \in L^2(0,T;\mathcal{H})$. This can be established by expanding $W(t)$ in the basis $W_{m,j}$ for $\mathcal{E}_m$, and then noting that the coefficients satisfy 
the Galerkin system \eqref{Galerkinapprox}. From this we see that the coefficients have derivatives in $L^2(0,T)$, and this establishes the initial claim that $\dot{W} \in L^2(0,T;\mathcal{H})$. From this
\[
\begin{split}
\|W (t) \|^2_{\mathcal{H}} & =\|W (0) \|^2_{\mathcal{H}}+ 2\int_0^t \mathrm{Re}\ \left (\dot{W}(s) , W(s) \right )_{\mathcal{H}}\ \mathrm{d} s\\
& =\|W (0) \|^2_{\mathcal{H}}+ 2\int_0^t \mathrm{Re}\ \Big [ a_2(w_{1},\dot{w}_{1}) + \beta (w_{1},\dot{w}_{1})_H + (w_{2},\dot{w}_{2})_H \Big ] \ \mathrm{d} s
\end{split}
\]
where we have used the notation $W = (w_{1}, w_{2})$. Note that this formula would be true by the fundamental theorem of calculus if $W \in \mathcal{C}^1([0,T];\mathcal{H})$, and then follows in our case by density.

Next, setting $V = (v_1, v_2) \in \mathcal{E}_m$, \eqref{Galerkinenergy} gives
\[
(\dot{w}_{1}(t), v_1)_H = (w_{2}(t), v_1)_H
\]
and
\[
(\dot{w}_{2}(t), v_2)_H = -a_2(w_{1}(t),v_2) - (2 R_\Omega w_{2}(t), v_2)_H + (f(t), v_2)_H
\]
for all $v_1$, $v_2 \in X_m$ and almost every $t \in [0,T]$. 

Now let $\pi_{X_m}^H$ be the $H$-orthogonal projection onto $X_m$. Choosing $v_2 = \dot{w}_{1}(t)$, and $v_1$ either $\dot{w}_{2}(t)$, $w_{1}(t)$, $\pi_{X_m}^H f(t)$, or $\pi_{X_m}^H 2 R_\Omega w_{2}(t)$, and using the skew symmetry of $R_\Omega$, we obtain, after some calculation,
\[
\begin{split}
\|W (t) \|^2_{\mathcal{H}}  & = \|W (0) \|^2_{\mathcal{H}} + 2\int_0^t \mathrm{Re}\ \Big [
\beta (w_1, w_2)_H  
+ (f,w_2)_H
\Big ] \ \mathrm{d}s\\
\end{split}
\]
Applying the Cauchy-Schwarz inequality and the inequality $2a b\leq a^2 + b^2$, 
we have
\[
\begin{split}
\|W (t) \|^2_{\mathcal{H}} & \lesssim \|f\|_{L^2(0,t;H)}^2 
+ \| W (0) \|^2_{\mathcal{H}}  + \int_0^t \|W(s) \|^2_{\mathcal{H}} \ \mathrm{d} s
\end{split}
\]
where we use the notation $A\lesssim B$ to mean $A \leq C B$ for some constant $C$ depending only on $T$ and $\beta$. Finally, 
Gronwall's lemma implies the result.
\end{proof}

Taking $W = U_m$ and observing that $U_m(0)$ is the $\mathcal{H}$-orthogonal projection of $U_0$ onto $\mathcal{E}_m$, Lemma \ref{Energy1} implies that 
\[
\|U_m(t) \|^2_{\mathcal{H}} \lesssim \|f\|_{L^2(0,t;H)}^2 + \|U_0\|_{\mathcal{H}}^2.
\]
Thus we see that the sequence $\{U_m\}$ is uniformly bounded in $L^\infty(0,T;\mathcal{H})$, and from this we can conclude that subsequences of $\{U_m\}$ converge weakly in $L^2(0,T;\mathcal{H})$ and weak* in $L^\infty(0,T;\mathcal{H})$. This is actually insufficient to establish existence and uniqueness of solutions for \eqref{Galerkinprob2} since we cannot conclude, based on this alone, that the weak limits are continuous with values in $\mathcal{H}$. This necessitates the so-called parabolic regularization in which the operator $\tilde{A}_2$ is changed to
\[
M_\epsilon = \begin{pmatrix}  0 & \Id \\ -A_2  & -2R_{\Omega} + \epsilon (A_2 + \beta \Id)  \end{pmatrix}.
\]
Assuming that $a_2$ is coercive, for every $\epsilon >0$ the problem \eqref{Galerkinprob2} with $\tilde{A}_2$ replaced by $M_\epsilon$ is in fact parabolic, and existence and uniqueness can be established from an appropriate analog of Lemma \ref{Energy1}. It then remains to prove that as $\epsilon \rightarrow 0^+$ a solution of \eqref{Galerkinprob2} is obtained. We will not expand on this further, but for details we refer the reader to \cite[Theorem 3, p.572, Theorem 4, p.574]{DautrayLionsV5}. In fact, well-posedness for \eqref{Galerkinprob2} follows from the results of section \ref{semigroup:sec} and so there is no rigor lost in omitting a proof of the same based only on Galerkin approximation. In a broader context, Galerkin approximation as a method to prove well-posedness enjoys the advantage that results can be obtained in the case when the operator $A_2$ depends on $t$, but we will not consider this here as it is not important in our application. Instead, we will now establish that the Galerkin approximations $U_m$ converge strongly in $\mathcal{C}^0([0,T];\mathcal{H})$ to the solution $U$ of \eqref{Galerkinprob2}.
\begin{theorem}
If $U_0 \in \mathcal{H}$ and $f \in L^2(0,T;H)$, then the Galerkin approximations defined by \eqref{Galerkinapprox} converge to the solution of \eqref{Galerkinprob1} in the $\mathcal{C}^0([0,T];\mathcal{H})$ norm.
\end{theorem}
\begin{proof}
Assume for the moment that $U\in \mathcal{C}^1([0,T];\mathcal{H})$ is the solution of \eqref{Galerkinprob2} and that $U_m$ are the Galerkin approximations to $U$ defined above. Note we are assuming here that $U$ has more regularity than should be expected. For convenience we will use the notation $E_m = U - U_m = (e_{1,m}, e_{2,m})$ for the approximation error and $U_m = (u_{1,m},u_{2,m})$ for the approximations. With the assumed regularity, we have, just as in the proof of Lemma \ref{Energy1},
\[
\| E_m(t) \|^2_{\mathcal{H}} =  \|E_m (0) \|^2_{\mathcal{H}} + 2\int_0^t \mathrm{Re}\ \Big [ a_2(e_{1,m},\dot{e}_{1,m}) + \beta (e_{1,m},\dot{e}_{1,m})_H + (e_{2,m},\dot{e}_{2,m})_H \Big ] \ \mathrm{d} s.
\]
The equations satisfied by $U$ and $U_m$ imply that
\[
a_2( e_{1,m}, v_2 ) + (\dot{e}_{2,m},v_2)_H = -2 (R_\Omega e_{2,m},v_2 )_H, \quad \mbox{and} \quad (\dot{e}_{1,m},v_1)_H = (e_{2,m},v_1)_H
\]
for all $v_2 \in X_m$ and $v_1 \in H$. First setting either $v_1 = \dot{e}_{2,m}$ or $v_1 = e_{1,m}$ we have
\[
\begin{split}
\| E_m(t) \|^2_{\mathcal{H}} & =  \|E_m (0) \|^2_{\mathcal{H}} + 2\int_0^t \mathrm{Re}\ \Big [ a_2(e_{1,m},\dot{e}_{1,m}) + \beta (e_{1,m},e_{2,m})_H + (\dot{e}_{1,m},\dot{e}_{2,m})_H \Big ] \ \mathrm{d} s\\
& = \|E_m (0) \|^2_{\mathcal{H}} + 2\int_0^t \mathrm{Re}\ \Big [ a_2(e_{1,m},\dot{u}_1 - \dot{u}_{1,m}) + (\dot{e}_{2,m},\dot{u}_1 - \dot{u}_{1,m})_H + \beta (e_{1,m},e_{2,m})_H \Big ] \ \mathrm{d} s
\end{split}
\]
Next we set $v_2 = \dot{u}_{1,m}$, and then for an arbitrary $v_m \in \mathcal{C}^2([0,T];X_m)$ set $v_2 = \dot{v}_m$ (this is using ``Galerkin orthogonality"). In this way, also using $(e_{2,m},\dot{u}_1)_H = (e_{2,m},u_2)_H$, we establish
\[
\begin{split}
\| E_m(t) \|^2_{\mathcal{H}} & = \|E_m (0) \|^2_{\mathcal{H}} + 2\int_0^t \mathrm{Re}\ \Big [ a_2(e_{1,m},\dot{u}_1 - \dot{v}_m) + (\dot{e}_{2,m},u_2 - \dot{v}_m)_H + \beta (e_{1,m},e_{2,m})_H\\
&\hskip1.5in - 2 (R_\Omega e_{2,m}, \dot{u}_{1,m}- \dot{v}_m)_H \Big ] \ \mathrm{d} s\\
& = \|E_m (0) \|^2_{\mathcal{H}} + 2\ \mathrm{Re}\Big [(e_{2,m},u_2 - \dot{v}_m)_H|_0^t\Big ] \\
& \hskip-.5in + 2\int_0^t \mathrm{Re}\ \Big [ a_2(e_{1,m},\dot{u}_1 - \dot{v}_m) - (e_{2,m},\dot{u}_2 - \ddot{v}_m)_H + \beta (e_{1,m},e_{2,m})_H - 2 (R_\Omega e_{2,m}, \dot{u}_{1,m}- \dot{v}_m)_H \Big ] \ \mathrm{d} s.
\end{split}
\]
Note that in the last term $\dot{u}_{1,m}$ may be replaced by $\dot{u}_{1}$ since $R_\Omega$ is skew symmetric and so
\[
(R_\Omega e_{2,m}, \dot{u}_{1,m}- \dot{v}_m)_H = (R_\Omega e_{2,m}, -e_{2,m} + \dot{u}_{1,m} - \dot{v}_m)_H = (R_\Omega e_{2,m}, \dot{u}_{1}- \dot{v}_m)_H.
\]
Now, using the Cauchy-Schwarz inequality, the inequality $2 ab \leq a^2 + b^2$, the boundedness of $a_2$, and a trace theorem, we obtain an estimate
\[
\| E_m(t) \|^2_{\mathcal{H}} \lesssim \|E_m (0) \|^2_{\mathcal{H}} + \int_0^t \| E_m \|^2_{\mathcal{H}} + \|\dot{u}_1 - \dot{v}_m\|^2_E + \|\dot{u}_{2} - \ddot{v}_m\|^2_H \ \mathrm{d} s.
\]
Gronwall's inequality then implies
\[
\sup_{t \in [0,T]} \| E_m(t) \|^2_{\mathcal{H}} \lesssim \|E_m (0) \|^2_{\mathcal{H}} + \int_0^t \|\dot{u}_1 - \dot{v}_m\|^2_E + \|\dot{u}_{2} - \ddot{v}_m\|^2_H \ \mathrm{d} s.
\]
We note that this holds for any $v_m \in \mathcal{C}^2([0,T];X_m)$, and the fact that $\cup_m X_m$ is dense in both $E$ and $H$ then proves that
\[
\lim_{m\rightarrow \infty} \sup_{t \in [0,T]} \| E_m(t) \|^2_{\mathcal{H}} = 0,
\]
or, equivalently, that the Galerkin approximations converge in the $\mathcal{C}^0([0,T];\mathcal{H})$ norm to the true solution.

At this point we have proven the theorem under the assumption that $U\in \mathcal{C}^1([0,T];\mathcal{H})$. The regularity of $U$ depends on that of $U_0$ and $f$, and can, more or less, be read off from the equation \eqref{VectorVSol} for $U$. Indeed, based on the comments following \eqref{VectorVSol}, if we assume that $U_0 \in D(\tilde{A}_2)$ and $f \in H^1(0,T;H)$, then for the solution we have $U \in \mathcal{C}^1([0,T],\mathcal{H})$. To obtain the proof in the general case then, we can use the facts that $D(\tilde{A}_2)$ is dense in $\mathcal{H}$, and $H^1(0,T;H)$ is dense in $L^2(0,T;H)$, as well as using Lemma \ref{Energy1} to bound the difference between the Galerkin approximation of the true problem, and the Galerkin approximation for a problem with regularized $U_0$ and $f$.
\end{proof}

 \section{Volterra equation method}\label{Voltsec}
 
 In this section we present one further method of proving well-posedness for equation \eqref{globalseis} via conversion to a Volterra equation. This method enjoys the additional benefit that we may analyse in more detail the error associated with neglecting the self-gravitation. Additionally it allows more general lower order operators to be added which could for example include viscoelastic effects. To accomplish this let us begin by defining another bilinear form
\[
a_3(u,v) = a_2(u,v) + \frac{1}{4 \pi G} \int_{\mathbb{R}^3} \nabla S(u) \cdot \nabla S(\overline{v}) \ dV.
\]
Here, we remove the first-order perturbation in the gravitational
potential but retain the initial potential, $\Phi^0$. (The acoustic
version of this results in the so-called Cowling approximation.) The
non-gravitating ``limit'' is obtained by also removing terms
containing $\Phi^0$ and $\Psi^s$; it should be carefully noted that this is not equivalent to setting $g_0' = 0$ in the expression \eqref{a2} for $a_2$. This is because of the initial stress $T_0$ which appears in the highest order terms of $a_2$, and already includes some contribution from the gravitation and rotation. Indeed, if, as with $a_3$, we just subtract the terms involving $\Phi^0$ and $\Psi^s$ we obtain the form
\begin{equation} \label{a4}
\begin{aligned}
a_4\big(u, v \big)
= a_2(u,v) - \int_{\tilde{X}} \rho^0 u \cdot \nabla \nabla ( \Phi^0 + \Psi^s) \cdot \overline{v}\ dV + \frac{1}{4 \pi G} \int_{\mathbb{R}^3} \nabla S(u) \cdot \nabla S(\overline{v}) \ dV 
\end{aligned}
\end{equation}
which generates the system of equations describing acousto-elastic waves including the effect of the initial stress $T_0$. Since the two terms which are removed are both bounded on $H = L^2(\widetilde{X}, \rho^0 \,dV)$, $a_4$ is coercive on $E$ relative to $H$, and the same theory follows for $a_4$ as the cases of $a_2$ and $a_3$. Unfortunately the removal of the two terms in $a_4$ does not result in a significantly simpler formula than \eqref{a2}. Indeed, one might imagine starting from the formula \eqref{a1} for  $a_1$, removing the two terms involving $\Phi^0$, $\Psi^s$, and $\Phi^1$, and then following a similar calculation to that in section \ref{weak:sec} in order to find a coercive Hermitian form. This plan results precisely in $a_4$ given by \eqref{a4} where $a_2$ is given by \eqref{a2}. In order to have a simpler formula we would need to remove the contribution of the initial stress $T_0$.

The domains of $a_3$ and $a_4$ are again $E$, the domain of $a_2$, and the proof of Theorem \ref{CoerciveA2} still applies to show that $a_3$ and $a_4$ are $E$ coercive relative to $L^2(\widetilde{X}, \rho^0 \,dV)$. We can define $A_3$ and $A_4$ in the same way that $A_2$ was defined in section \ref{A2Op}, and Propositions \ref{PropOp} and \ref{DomA2} still hold with $A_2$ replaced by $A_3$ or $A_4$. We consider now the initial value problem
\begin{equation} \label{volt_IVP}
\ddot{u}  + A_3 u + 2R_\Omega \dot{u} + \rho^0 \nabla S(u) =  f, \quad u(0) = u_0, \quad \dot{u}(0) = \dot{u}_0,
\end{equation}
or the analogous problem using $A_4$. The distinction from what was done in sections \ref{semigroup:sec} and \ref{Galerkin1} is that here we have separated out the contributions of self-gravitation for $A_3$, and gravitation as well as rotation for $A_4$. We will show that any ``lower order perturbation" such as the self-gravitation does not affect the well-posedness of the problem, and give an expression for the difference between the solution without the effects of self-gravitation, and the solution including those effects. The other lower order terms may also be separated out in the same way, but we highlight the self-gravitation term here as it is the only nonlocal term.

Let us consider the more general problem
\begin{equation} \label{volt_genIVP}
\ddot{u}  + A u + Q_1 \dot{u} + Q_0 u =  f, \quad u(0) = u_0, \quad \dot{u}(0) = \dot{u}_0
\end{equation}
on the interval $[0,T]$ where we are considering that the operator $A$ may be $A_3$ or $A_4$. We will work in a rather general setting in which $Q_1: L^1(0,t; E) \rightarrow L^1(0,t; E)$ is bounded uniformly for any $t$, and can be extended to act on the same spaces with $E$ replaced by $H$. For $Q_0$, we assume $Q_0: L^1(0,t; D(A)) \rightarrow L^1(0,t; E)$ uniformly for any $t$ where $D(A)$ is given the graph norm for $A$, and $Q_0$ also extends to a bounded operator $L^1(0,t; E) \rightarrow L^1(0,t; H)$ uniformly in $t$. This degree of generality allows the treatment of viscoelasticity, although we will not discuss that topic any further. Also, it can be shown that $Q_0 = \nabla S$ satisfies these properties. We first establish the following result concerning strong solutions of \eqref{volt_genIVP}.

\begin{theorem} \label{volt_thm}
Suppose that the hypotheses of Theorem \ref{CoerciveA2} are satisfied, and that $Q_0$ and $Q_1$ are any linear operators with the mapping properties described above. Then if $u_0 \in D(A)$, $\dot{u}_0 \in E$, and $f \in C^0([0,T]; E)$ for any $T>0$ there exists a unique strong solution $u \in C^2([0,T];H) \cap C^1([0,T];E) \cap C^0([0,T];D(A))$ of \eqref{volt_genIVP}. The solution satisfies the estimate
\[
\|u\|_{C([0,T];E)} \leq Ce^{CT} (\|u_0\|_E + \|\dot{u}_0\|_H + \|f\|_{L^1(0,T;H)})
\]
where the constant $C$ does not depend on $T$.
\end{theorem}

\begin{remark}
The constant $C$ is given explicitly in terms of operator norms during the proof.
\end{remark}

\begin{proof}
We will, as in section \ref{semigroup:sec}, define $\mathcal{H} = E \times H$, and use the same inner product defined by \eqref{SPH} on $\mathcal{H}$, although the constant $\beta$ in this case comes from the coercivity inequality for $a_3$ or $a_4$ rather than $a_2$. The unbounded operator $\widetilde{A}$ with domain $D(A) \times E \subset \mathcal{H}$ is then defined by
\[
\widetilde{A} = \left (
\begin{matrix}
0 & \mathrm{Id}\\
- A & 0
\end{matrix}
\right ).
\]
Using the same method as in section \ref{InfGen} we see that $\widetilde{A}$ is the generator of a quasi-contraction semi-group on $\mathcal{H}$ which we will label here $G(t)$. Using $G(t)$ we define two additional families of operators $C_{A}(t)$ and $S_{A}(t)$ by
\[
C_{A}(t)u = \pi_1 \circ G(t) \left (
\begin{matrix}
u \\
 0
\end{matrix}
\right )
\]
and
\[
S_{A}(t)u = \pi_1 \circ G(t) \left (
\begin{matrix}
0 \\
 u
\end{matrix}
\right )
\]
where $\pi_1$ is the projection onto the $E$ component in $\mathcal{H} = E \times H$. Using properties of the semi-group $G(t)$ we can easily establish the properties of $C_{A}(t)$ and $S_{A}(t)$ given in the following lemma.
\medskip

\begin{lemma} \label{CS_lem}
Making all the assumptions of Theorem \ref{CoerciveA2} we have
\begin{enumerate}
\item $C_{A}(0) = \mathrm{Id}$, and $S_{A}(0) = 0$.
\item $C_{A}(t): E \rightarrow E$, and $S_{A}(t): H \rightarrow E$ for all $t \in [0,\infty)$.
\item $C_{A}(t)$ and $S_{A}(t)$ are bounded maps of $E$ into $D(A)$ (with the graph norm) for all $t \in [0,\infty)$.
\end{enumerate}
If in addition $u \in D(A)$, then the map $t \rightarrow C_{A}(t) u$ is $C^2$ with 
\begin{enumerate}
\item $\dot{C}_{A}(0)u = 0$,
\item $\dot{C}_{A}(t): D(A) \rightarrow E$, and $\dot{C}_{A}(t)$ extends to a continuous operator $\dot{C}_{A}(t): E \rightarrow H$,
\item $\ddot{C}_{A}(t) u = - A C_{A}(t) u$ for all $t \in [0,\infty)$.
\end{enumerate}
If $u \in E$, then the map $t \rightarrow S_{A}(t) u$ is $C^2$ with 
\begin{enumerate}
\item $\dot{S}_{A}(0) u = u$, 
\item $\dot{S}_{A}(t) : E \rightarrow E$ and extends to a continuous operator $\dot{S}_{A}(t): H \rightarrow H$,
\item $\ddot{S}_{A}(t) u = -A S_{A}(t) u$ for all $t \in [0,\infty)$.
\end{enumerate}
\end{lemma}
\medskip

\noindent Using these operators we introduce the following ansatz
\begin{equation}\label{volt_ansatz}
u(t) = \int_0^t S_{A}(t-s) g(s) \ \mathrm{d} s + C_{A}(t) u_0 + S_{A}(t) \dot u_0
\end{equation}
which is well-defined on $[0,T]$ provided $u_0 \in E$, $\dot u_0 \in H$, and $g \in L^1(0,T; H)$. If in fact $u_0 \in D(A)$, $\dot{u}_0 \in E$, and $g \in L^1(0,T; E)$, then using Lemma \ref{CS_lem} $u(t)$ is $C^2$,
\[
\dot{u}(t) = \int_0^t \dot{S}_{A}(t-s) g(s) \ \mathrm{d} s + \dot{C}_{A}(t) u_0 + \dot{S}_{A}(t) \dot{u}_0,
\]
and
\[
\ddot{u}(t) = g(t) - \int_0^t A S_{A}(t-s) g(s) \ \mathrm{d} s - A C_{A}(t) u_0 - A S_{A}(t).
\]
If $u$ is to satisfy \eqref{volt_genIVP} then we find that
\begin{equation} \label{intg}
\begin{split}
g(t) +  Q_1 \int_0^t \dot{S}_{A}(t-s) g(s) \ \mathrm{d} s + &Q_0 \int_0^t S_{A}(t-s)  g(s) \ \mathrm{d} s \\
&+ (Q_1 \dot{C}_{A}(t) + Q_0 C_{A}(t)) u_0 + (Q_1 \dot{S}_{A}(t) + Q_0 S_{A}(t)) \dot{u}_0 = f(t).
\end{split}
\end{equation}
This is essentially a Volterra equation for $g$, although the general form of the operators $Q_0$ and $Q_1$ makes it slightly more complicated.

Now we introduce the operator
 \[
 \mathcal{K}[g](t) =  -Q_1 \int_0^t \dot{S}_{A}(t-s) g(s) \ \mathrm{d} s - Q_0 \int_0^t S_{A}(t-s)  g(s) \ \mathrm{d} s,
 \]
 and set
 \[
 \mathbf{F}(t) = f(t) - (Q_1 \dot{C}_{A}(t) + Q_0 C_{A}(t)) u_0 - (Q_1 \dot{S}_{A}(t) + Q_0 S_{A}(t)) \dot{u}_0.
 \]
From Lemma \ref{CS_lem}, and our assumptions on $Q_1$ and $Q_0$ we have that
\[
 \mathcal{K}: L^1(0,t;E) \rightarrow L^1(0,t;E)
\]
uniformly in $t$. An induction argument shows further that
\begin{equation} \label{Kjest}
\|\mathcal{K}^j \|_{L^1(0,T;E) \rightarrow L^1(0,T; E)} \leq C_E^j \frac{T^{j}}{j!}
\end{equation}
where
\[
\begin{split}
C_E = \max_{t \in [0,T]} \big \{& \|\dot{S}_{A}(t)\|_{E \rightarrow E} \big \} \max_{t \in [0,T]} \big \{ \|Q_1\|_{L^1(0,t;E)\rightarrow L^1(0,t;E)}\big \} \\
& + \max_{t \in [0,T]} \big \{ \|S_{A}(t)\|_{E \rightarrow D(A) } \big \} \max_{t \in[0,T]} \big \{ \|Q_0\|_{L^1(0,t;D(A))\rightarrow L^1(0,t;E)} \big \}.
\end{split}
\] 
Also, using different parts of Lemma \ref{CS_lem} we can establish the same estimate with $E$ replaced by $H$ and the constant $C_E$ replaced by $C_H$ given by
\begin{equation}\label{C_H}
\begin{split}
C_H =  \max_{t \in [0,T]} \big \{& \|\dot{S}_{A}(t)\|_{H \rightarrow H} \big \} \max_{t \in [0,T]} \big \{ \|Q_1\|_{L^1(0,t;H)\rightarrow L^1(0,t;H)}\big \} \\
& + \max_{t \in [0,T]} \big \{ \|S_{A}(t)\|_{H \rightarrow E } \big \} \max_{t \in[0,T]} \big \{ \|Q_0\|_{L^1(0,t;E)\rightarrow L^1(0,t;H)} \big \}.
\end{split}
\end{equation}


Now, the integral equation \eqref{intg} for $g$ is simplified using the operators just introduced into the form
 \begin{equation} \label{volt}
 g(t) = \mathbf{F}(t) + \mathcal{K}[g](t).
 \end{equation}
The solution of this equation is formally
\begin{equation}\label{volt_sol}
g(s) = \sum_{j=0}^\infty \mathcal{K}^j \mathbf{F}.
\end{equation}
Using \eqref{Kjest}, if $u_0 \in D(A)$, $\dot{u}_0\in E$, and $f \in L^1(0,T;E)$, then this series converges in $L^1(0,T;E)$, and the resulting $g$ satisfies the equation \eqref{intg}. Also, using \eqref{Kjest} with $E$ replaced by $H$, if $u_0 \in E$, $\dot{u}_0 \in H$, and $f \in L^1(0,T;H)$, then the series converges in $L^1(0,T;H)$, and the resulting $g$ still satisfies \eqref{intg}. Further, we have the estimate
\begin{equation}\label{gest}
\|g(t) \|_{L^1(0,T;H)} \leq e^{C_H T} \|\mathbf{F}\|_{L^1(0,T;H)} \leq A e^{C_H T} \Big (\|u_0\|_{E} + \|\dot{u}_0\|_H + \|f \|_{L^1(0,T;H)} \Big )
\end{equation}
where
\begin{equation}\label{Adef}
\begin{split}
A = 3 \max \Bigg \{ \max_{t \in [0,T]} \big \{& \|\dot{C}_{A}(t)\|_{E \rightarrow H} \big \} \max_{t \in [0,T]} \big \{ \|Q_1\|_{L^1(0,t;H)\rightarrow L^1(0,t;H)}\big \} \\
& + \max_{t \in [0,T]} \big \{ \|C_{A}(t)\|_{E \rightarrow E } \big \} \max_{t \in[0,T]} \big \{ \|Q_0\|_{L^1(0,t;E)\rightarrow L^1(0,t;H)} \big \}, C_H, 1 \Bigg \}.
\end{split}
\end{equation}
This establishes the existence of strong solutions to \eqref{volt_genIVP} when $u_0 \in D(A)$, $\dot{u}_0\in E$, and $f \in L^1(0,T;E)$, as well as the statements on regularity and the estimate of $u$ given in the theorem. It remains to prove uniqueness of the solution.

To establish uniqueness of the solution let us consider $u$ satisfying \eqref{volt_genIVP} with $u_0 = \dot{u}_0 = f = 0$, and the given regularity properties. Then set
\[
g(t) = \ddot{u}(t) + A u(t).
\]
From \eqref{volt_genIVP} and the regularity properties of $u$, we see that $g(s) \in L^1(0,T; E)$. Using the identity
\[
S_{A}(t-s) (\ddot{u}(s) + A u(s)) = \frac{\mathrm{d}}{\mathrm{d} s} \left ( S_{A}(t-s) \dot{u}(s) + \dot{S}_{A}(t - s) u(s) \right )
\]
we see that
\[
u(t) = \int_0^t S_{A}(t-s) g(s) \ \mathrm{d} s.
\]
Thus the remainder of the calculation from above holds and we see that $g$ must satisfy the equation \eqref{volt} with $\mathbf{F} = 0$. Since $\mathrm{Id} - \mathcal{K}$ is invertible on $L^1(0,T;E)$, this completes the proof.
\end{proof}

Even in the case when we only assume $u_0 \in E$, $\dot{u}_0 \in H$, and $f \in L^1(0,T;H)$ the solution constructed in the proof of Theorem \ref{volt_thm} using equations \eqref{volt_sol} and \eqref{volt_ansatz} is well-defined and lies in the space $\mathcal{C}^0([0,T];H)$. This is the unique weak solution in a sense adapted from Definition \ref{WeakBrezis} as we now demonstrate. Note that Definition \ref{WeakBrezis} does not apply in this case directly, even if we reformulate \eqref{volt_genIVP} as a first-order system, since we have allowed $Q_0$ and $Q_1$ to act on $L^1(0,T;H)$. The definition of weak solution we apply in this case will be as follows.
\begin{definition}\label{weak_voltsol}
Suppose that $u_0 \in E$, $\dot{u}_0 \in H$, and $f \in L^1(0,T;H)$. We say that $u \in \mathcal{C}^0([0,T];E)$ is a weak solution of \eqref{volt_genIVP} if there are sequences $f_n \in \mathcal{C}^0([0,T],E)$, $u_{0,n} \in D(A)$, $\dot{u}_{0,n} \in E$, and $u_n \in C^2([0,T];H) \cap C^1([0,T];E) \cap C^0([0,T];D(A))$ such that
\begin{enumerate}
\item $u_n$ is the unique strong solution of
\[
\ddot{u}_n  + A u_n + Q_1 \dot{u}_n + Q_0 u_n =  f, \quad u_n(0) = u_{0,n}, \quad \dot{u}_n(0) = \dot{u}_{0,n}
\]
for all $n$.
\item $f_n \rightarrow f$ in $L^1(0,T;H)$
\item $u_{0,n} \rightarrow u_0$ in $E$.
\item $\dot{u}_{0,n} \rightarrow \dot{u}_0$ in $H$.
\item $u_n \rightarrow u$ in $C(0,T;E)$.
\end{enumerate}
\end{definition}

With this definition we have the following result for weak solutions which follows easily from the estimate given in Theorem \ref{volt_thm}.

\begin{theorem} \label{volt_weakthm}
If $u_0 \in E$, $\dot{u}_0 \in H$, and $f \in L^1(0,T;H)$ the unique weak solution of \eqref{volt_genIVP} is given by formulas \eqref{volt_sol} and \eqref{volt_ansatz}.
\end{theorem}

Finally, we can compare the solution of \eqref{volt_genIVP} with the solution of the corresponding equation if the effects of the lower order operators, $Q_0$ and $Q_1$, are neglected by noting that the solution $u_r$ of \eqref{volt_genIVP} with $Q_0$ and $Q_1$ set to zero is, according to \eqref{volt_sol} and \eqref{volt_ansatz}
\[
u_r(t) = \int_0^t S_{A}(t-s) f(s) \ \mathrm{d} s + C(t) u_0 + S_{A}(t) \dot{u}_0.
\]
Therefore if $u$ is the solution of \eqref{volt_genIVP} then
\[
u(t) - u_r(t)= \int_0^t S_{A}(t-s) \left ( \sum_{j=1}^\infty \mathcal{K}^j \mathbf{F} (s) \right ) \mathrm{d} s.
\]
From this we have the estimate
\[
\|u - u_r\|_{C([0,T];E)} \leq C_H \max_{t \in [0,T]} \{ \|S_{A}(t)\|_{H\rightarrow E} \} T A e^{C_H T} \Big (\|u_0\|_{E} + \|\dot{u}_0\|_H + \|f \|_{L^1(0,T;H)} \Big )
\]
where $C_H$ and $A$ are defined by \eqref{C_H} and \eqref{Adef}. Note that this error is $O(C_H)$, and so decreases linearly with $\|Q_1\|_{L^1(0,t;H)\rightarrow L^1(0,t;H)}$ and $\|Q_0\|_{L^1(0,t;E)\rightarrow L^1(0,t;H)}$. Actually, once well-posedness for \eqref{volt_genIVP} is known, then this error estimate also follows from Corollary \ref{Energyest:cor}.

\section*{Acknowledgements}

This research was carried out in part while MVdH was a Schlumberger
Visiting Professor in the Department of Earth, Atmospheric and
Planetary Sciences at MIT, and was supported in part by the members of
the Geo-Mathematical Imaging Group, currently at Rice University.

\appendix

\renewcommand{\theequation}{\Alph{section}.\arabic{equation}}

\setcounter{equation}{0}
\section{Tangential Divergence and Gradient, and Weingarten operator}\label{Wein}

We take advantage of a view facts from geometry in Sections~\ref{FS:sec} and~\ref{weak:sec} when showing the equivalence of the different weak formulations. For more details, see \cite{ONeill}. We let $\Omega \subset \mathbb{R}^3$ be a domain and $\Sigma$ be a $C^2$ surface in $\Omega$. We let $\nabla$ be the Euclidean covariant derivative on $\Omega$ and $\nu$ a unit normal vector field for $\Sigma$.

First we recall the definition of the shape and Weingarten operators.
\medskip

 \begin{definition}
For $p \in \Sigma$ we have operators the shape and Weingarten operators $S$ and $W$ on $T_p \Sigma$ defined by
 $$S(v_p) := - \nabla_{v_p} \nu ;\ \ W (v_p):= \nabla_{v_p} \nu. $$
 \end{definition}
\medskip

\noindent Both $S$ and $W$ are self-adjoint. Next, we have the second fundamental form.
\medskip

\begin{definition}
For $p \in \Sigma$ the second fundamental form of $\Sigma$ at $p$ is the symmetric bilinear form on $T_p \Sigma$ defined by
$$II(v_p, w_p) := \langle v_p, S(w_p)\rangle = \langle S(v_p), w_p\rangle$$
\end{definition}
\medskip

\begin{definition}
For $X_p \in T_p\Sigma$ and $Y$ a tangent vector field on $\Sigma$ the Levi-Civita connection on $\Sigma$ is given by
$$\nabla^{\Sigma}_{X_p} Y := \pi_{T_pM} \nabla_{X_p} Y $$
where $\pi_{T_p\Sigma} : T_p\RR^n \rightarrow T_p \Sigma$ is orthogonal projection onto $T_p \Sigma$. 
\end{definition}
\medskip

\noindent From this definition we obtain the following formula for $\nabla^\Sigma$
$$\nabla^{\Sigma}_{X_p} Y= \nabla_{X_p} Y - \langle S (X_p), Y_p\rangle \nu_p
= \nabla_{X_p} Y  + II(X_p, Y_p) \nu_p.$$
For a function $f$ defined on $\Sigma$ the tangential gradient $\nabla^{\Sigma} f $ is defined to be the dual of $df$ via the metric induced on $\Sigma$ by the Euclidean metric. This is given by the formula
$$\nabla^{\Sigma} f = \nabla f - \langle \nabla f , \nu \rangle \nu.$$
The tangential divergence of a vector field on $\Sigma$ is defined as the trace of the covariant derivative corresponding to the Levi-Civita connection on $\Sigma$. We have the following identity relating the tangential divergence and the Euclidean divergence in $\Omega$.
 \begin{equation}\label{tanDi}
 \nabla \cdot u = \nabla^{\Sigma} \cdot(  u - (u\cdot \nu) \nu )
 + \nu \cdot \nabla u \cdot \nu +     ( c_1 + c_2) (u\cdot \nu) 
 \end{equation}
 where $c_1, c_2$ are the eigenvalues of the Weingarten operator $W$.

\setcounter{equation}{0}
\section{Estimates for gravitational potential}

\begin{lemma}\label{infiComp}
If $f$ has compact support with $\supp f \subset B(0,R)$ then
$$\int \dfrac{1}{\lvert x-y\rvert} f(y) \, dy \rightarrow 0 , \ \lvert x\rvert \rightarrow \infty$$
$$\nabla \int \dfrac{1}{\lvert x-y\rvert} f(y) \, dy \rightarrow 0 , \ \lvert x\rvert \rightarrow \infty$$
\end{lemma}
\begin{proof}
Since $f$ has compact support we have
$$\int_{\RR^3} \dfrac{1}{\lvert x-y\rvert} f(y) \, dy 
= \int_{B(0,R)} \dfrac{1}{\lvert x-y\rvert} f(y) \, dy $$
Using the following inequality
$$\lvert x-y\rvert  > \lvert x\rvert - \lvert y \rvert = \lvert x\rvert - R, $$
we can establish the bound
$$\int_{\RR^3} \dfrac{1}{\lvert x-y\rvert} f(y) \, dy \leq
\dfrac{1}{\lvert x\rvert - R} \int_{B(0,R)}  f(y) \, dy 
\leq \dfrac{1}{\lvert x\rvert - R}  \lVert f\rVert_{L^1}
\rightarrow 0, \ \lvert x\rvert \rightarrow \infty.$$
Similarly we have
$$\nabla\int_{\RR^3} \dfrac{1}{\lvert x-y\rvert} f(y) \, dy
= \int_{\RR^3} \dfrac{x-y}{\lvert x-y\rvert^3} f(y) \, dy$$
\begin{align*}
\Rightarrow \Big|\nabla \int_{\RR^3} \dfrac{1}{\lvert x-y\rvert} f(y) \, dy \Big|
&= \Big|  \int_{\RR^3} \dfrac{x-y}{\lvert x-y\rvert} f(y) \, dy\Big|\\
&\leq \int_{B(0,R)} \dfrac{1}{\lvert x-y\rvert^2} f(y) \, dy\\
&\leq \dfrac{1}{(\lvert x\rvert - R)^2} \int_{B(0,R)}  f(y) \, dy \\
&\leq \dfrac{1}{(\lvert x\rvert - R)^2}  \lVert f\rVert_{L^1}\end{align*}
$$\Rightarrow \Big|\nabla \int_{\RR^3} \dfrac{1}{\lvert x-y\rvert} f(y) \, dy \Big|
\leq \dfrac{1}{(\lvert x\rvert - R)^2}  \lVert f\rVert_{L^1}\rightarrow 0, \ \lvert x\rvert \rightarrow \infty.$$
\end{proof}

\setcounter{equation}{0}
\section{Conservation of physical energy}\label{energy:sec} 

In this section we relate the total physical energy of the
rotating earth to our formulation, and in particular the energy given in \eqref{Energy}. We use \cite{Dahlen} as a reference
for the physical energy. The kinetic energy of the freely deforming
earth, minus that of the initial uniformly rotating earth, is given by
\begin{equation}
   E^{\mathrm{kin}}
        = \int_{\tilde{X}} \rho^0 [\tfrac{1}{2} \dot{u} \cdot \dot{u}
     + u \cdot \nabla \Psi^s
     + \tfrac{1}{2} u \cdot \nabla\nabla \Psi^s \cdot u] \, dV .
\end{equation}
The stored elastic energy in the deformed earth, relative to that in
the initial undeformed earth, is given by
\begin{equation}
   E^{\mathrm{el}}
        = \int_{\tilde{X}} [T^0 : \tfrac{1}{2} [\nabla u + (\nabla u)^T]
               + \tfrac{1}{2} \nabla u : \Lambda^{T^0} : \nabla u]
                 \, dV .
\end{equation}
The gravitational potential energy associated with a deformation $u$
is given by
\begin{equation}
   E^{\mathrm{g}}
        = \int_{\tilde{X}} \rho^0 [u \cdot \nabla \Phi^0
               + \tfrac{1}{2} u \cdot \nabla S(u)
               + \tfrac{1}{2} u \cdot \nabla\nabla \Phi^0 \cdot u]
                 \, dV .
\end{equation}
The total energy then takes the form
\begin{equation}
\begin{split}
   E & = E^{\mathrm{kin}} + E^{\mathrm{el}} + E^{\mathrm{g}} \\
   & = \int_{\tilde{X}} \rho^0 [\tfrac{1}{2} \dot{u} \cdot \dot{u}
     + u \cdot g_0'
     + \tfrac{1}{2} u \cdot \nabla g_0' \cdot u + \tfrac{1}{2} u \cdot \nabla S(u)] \\
               &\hskip1.5in  + [T^0 : \tfrac{1}{2} [\nabla u + (\nabla u)^T]
               + \tfrac{1}{2} \nabla u : \Lambda^{T^0} : \nabla u] \,  dV.
\end{split}
\end{equation}
Now we simplify the terms that are first order in $u$ by integrating by
parts, using the boundary conditions for $T^0$, and applying the
mechanical equilibrium equation \eqref{Equi1} to obtain
\[
\int_{\tilde{X}} [T^0 : \tfrac{1}{2} [\nabla \dot{u} + (\nabla u)^T] + \rho^0 u \cdot g_0'] \, d V = \int_{\Sigma^{F}} [\nu \cdot T^0 \cdot u]_-^+ d \Sigma = -\int_{\Sigma^{F}} p^0 [\nu \cdot u]_-^+ \ d \Sigma.
\]
The second equality follows from the boundary condition for $T^0$ and
the fact that $T^0 = -p^0 \mathrm{Id}$ in the fluid region. If we
assume further that $u$ satisfies the \textit{second-order} tangential
slip condition \cite[Formula (3.95)]{Dahlen}
\begin{equation} \label{2ndorderBC}
 \left[ \nu \cdot u - u \cdot \nabla^{\Sigma} (\nu \cdot u)
    + \tfrac{1}{2} u \cdot (\nabla^{\Sigma} \nu) \cdot u
          \right]_-^+ = 0,
\end{equation}
then
\[
\begin{split}
\int_{\Sigma^{F}} p^0 [\nu \cdot u]_-^+ \ d \Sigma &= \int_{\Sigma^{F}} p^0 \left [ \tfrac{1}{2} W\big( u - (u \cdot \nu) \nu\big) \cdot \big( u - (u\cdot \nu) \nu\big)  - u \cdot \nabla^{\Sigma}(\nu \cdot u) \right ]_-^+ \ d \Sigma
\end{split}
\]
Therefore we may rewrite the total energy as
\[
\begin{split}
   E & = \int_{\tilde{X}} \rho^0 [\tfrac{1}{2} \dot{u} \cdot \dot{u}
     + \tfrac{1}{2} u \cdot \nabla g_0' \cdot u + \tfrac{1}{2} u \cdot \nabla S(u) + \tfrac{1}{2} \nabla u : \Lambda^{T^0} : \nabla u]\, dV \\
               &\hskip1.5in  -  \int_{\Sigma^{F}} p^0 \left [ \tfrac{1}{2} W\big( u - (u \cdot \nu) \nu\big) \cdot \big( u - (u\cdot \nu) \nu\big)  - u \cdot \nabla^{\Sigma}(\nu \cdot u) \right ]_-^+ \ d \Sigma.
\end{split}
\]
Now suppose that $u$ is sufficiently regular on $\Omega^S \cup \Omega^F$, satisfies all of the boundary conditions appearing in table~\ref{unT}, and satisfies the equation \eqref{globalseis} with $f = 0$. We point out that the assumption that $u$ satisfies the second and first order tangential slip conditions here appears unphysical, but is in fact the same as used for example in \cite{Dahlen}. In this case, by Lemma~\ref{aOriga1lem}
\[
E = \frac{1}{2} (\|\dot{u}\|_{L^2(\tilde{X}, \rho^0 dV)}^2 + a_{original}(u,u)),
\]
and then, using the fact that $u$ satisfies the necessary boundary conditions by Lemma \ref{A3}
\[
E = \frac{1}{2} (\|\dot{u}\|_{L^2(\tilde{X}, \rho^0 dV)}^2 + a_{2}(u,u)).
\]
This is the same as \eqref{Energy}, and since $u$ satisfies \eqref{globalseis} with sufficient regularity and satisfying the required boundary conditions the discussion following \eqref{Energy} shows that $\dot{E} = 0$, and energy is conserved. We stress here that for the total physical energy and the energy given in \eqref{Energy} to be equivalent we must use both the nonlinear second order slip conditions on $\Sigma^F$, and all the boundary conditions in table~\ref{unT}.

\setcounter{equation}{0}
\section{Some facts on vector-valued distributions}

The following technical lemma is used in some of the proofs in section
\ref{semigroup:sec}.
\medskip

\begin{lemma}\label{commute}
Let $H$ be a Hilbert space and $u_1$, $u_2 \in \mathcal{C}^0([0,T], H)$. 
Suppose 
\begin{equation}\label{asumFirstD}
\forall w \in H,\ \dfrac{d}{dt} (u_1, w)_H = (u_2, w)_H , \ \text{in} \ \mathcal{D}'(0,T).
\end{equation}
Then $ u_1 \in \mathcal{C}^1([0,T], H)$ and $\dot{u}_1 = u_2$ with the derivative taken in the classical sense.
\end{lemma}

\begin{proof}
Note that $u_1 \in \mathcal{C}^0([0,T], H)$ and so we can differentiate $u_1$ in the sense of distributions to get $\tfrac{d}{dt} u_1 \in \mathcal{D}'( 0,T; H)$. We will show that under assumption \eqref{asumFirstD}, in fact $\tfrac{d}{dt} u_1  = u_2 \in D'(0,T; H)$. This will complete the proof because $u_2 \in \mathcal{C}^0([0,T], H)$ and $u_1(t)$ then differs only by a constant from
\[
\int_0^t u_2(s) \ \mathrm{d} s \in \mathcal{C}^1([0,T], H).
\]
For all $w \in H$ and $\varphi \in \mathcal{D}(0,T)$ 
\begin{align*}
\left \langle \dfrac{d}{dt} (u_1, w)_H, \varphi \right \rangle_{\mathcal{D}', \mathcal{D}}
 & = - \int_0^T (u_1, w)_H \dot{\varphi} \, dt \\
&  = - \int_0^T ( \varphi' u_1, w)_H \, dt  =    \left(  -\int_0^T u_1 \dot{\varphi} \, dt  , w\right)_H .
\end{align*}
 The second to last equality follows from \cite[Corollary 1, p.470]{DautrayLionsV5}. 
On the other hand  by assumption \eqref{asumFirstD}
\begin{align*}
\left\langle \dfrac{d}{dt} (u_1, w)_H, \varphi \right\rangle_{\mathcal{D}', \mathcal{D}}
   &= \langle (u_2, w)_H, \varphi\rangle_{\mathcal{D}', \mathcal{D}}  \stackrel{(*)}{=} \int_0^T (u_2, w)_H \varphi\, dt \\
   & = \int_0^T (\varphi u_2, w)_H \, dt  = \left( \int_0^T u_2 \varphi \, dt     , w\right)_H. 
\end{align*}
    since $(*)$ is due to $u_2\in \mathcal{C}^0( [0,T], H)$ and the second to last 
    follows again from  \cite[Corollary 1, p.470]{DautrayLionsV5}. 
Hence we have
$$\forall w \in H , \forall \varphi \in \mathcal{D}, \ \     \left( \int_0^T u_2 \varphi \, dt     , w\right)_H =   \left(  -\int_0^T u_1 \dot{\varphi} \, dt  , w\right)_H $$
$$\Rightarrow \forall \varphi\in \mathcal{D} (0,T) , \  \int_0^T u_2 \varphi \, dt  =  -\int_0^T u_1 \dot{\varphi} \, dt . $$
This means that 
$$ \dfrac{d}{dt} u_1  = u_2 \in  \mathcal{D}'( 0,T; H)$$
which completes the proof.
\end{proof}

\bibliographystyle{siam}
\bibliography{global_Bibliography}

 \end{document}